\documentclass[11pt]{article}
\usepackage{csquotes}

\usepackage[margin=1in]{geometry}

\usepackage[T1]{fontenc}
\usepackage[utf8]{inputenc}
\usepackage{geometry}
\usepackage{authblk}
\usepackage{amsmath}
\usepackage{amssymb}
\usepackage{amsthm}
\usepackage{thmtools, thm-restate}
\usepackage{mathabx}
\usepackage{fullpage}
\usepackage{mdwlist}
\usepackage{graphicx}
\usepackage[most]{tcolorbox}
\usepackage{bbm}
\usepackage{bm}
\usepackage{autobreak}

\usepackage[hidelinks]{hyperref}
\usepackage[nameinlink]{cleveref}
\usepackage{url}            

\usepackage{stackengine}
\usepackage{imakeidx}
\usepackage{xcolor}
\usepackage{graphicx}
\usepackage{caption}
\usepackage{subcaption}
\usepackage{rotating}
\usepackage{quiver}

\definecolor{thmcolor}{RGB}{235, 235, 235}
\definecolor{citecolor}{RGB}{1, 210, 56}
\newtcolorbox{thmbox}{colback=thmcolor!5!white,colframe=black!75!black}

\allowdisplaybreaks
\hypersetup{
	colorlinks,
	linkcolor={blue},
	citecolor={citecolor},
	urlcolor={blue}
}

\usepackage{tikz}
\usetikzlibrary{decorations.pathmorphing}

\newtheorem{theorem}{Theorem}[section]
\newtheorem{corollary}[theorem]{Corollary}
\newtheorem{lemma}[theorem]{Lemma}
\newtheorem{observation}[theorem]{Observation}
\newtheorem{proposition}[theorem]{Proposition}
\newtheorem{definition}[theorem]{Definition}
\newtheorem{claim}[theorem]{Claim}
\newtheorem{fact}[theorem]{Fact}
\newtheorem{remark}[theorem]{Remark}

\newcommand{\poly}{\mathop{\mathrm{poly}}}

\newcommand{\Boo}{\{0,1 \}}

\newcommand{\bigO}{\mathcal{O}}

\newcommand{\paren}[1]{\left( #1 \right)}
\newcommand{\brac}[1]{\left[ #1 \right]}
\newcommand{\set}[1]{\left\{ #1 \right\}}

\newcommand{\setcond}[2]{\left\{ #1 \;\middle\vert\; #2 \right\}}

\newcommand{\F}{\mathbb{F}}

\newcommand{\IPS}{\mathrm{IPS}}
\newcommand{\IPSLIN}{\mathrm{IPS}_{\mathrm{LIN}}}
\newcommand{\IPSLINp}{\mathrm{IPS}_{\mathrm{LIN}'}}
\newcommand{\roABP}{\mathrm{roABP}}
\newcommand{\ml}{\mathsf{ml}}

\newcommand{\xbar}{\mathbf{x}}
\newcommand{\ACz}{\text{AC}^0}

\usepackage[
backend=biber,
style=alphabetic,
sorting=nyt,
backref=true,
maxcitenames = 8,
mincitenames = 5,
maxalphanames = 8,
minalphanames = 5,
maxnames = 10,
minnames = 5
]{biblatex}

\begin{document}
\title{New Bounds for the Ideal Proof System in Positive Characteristic}
\author{
{Amik Raj Behera \thanks{Department of Computer Science, University of Copenhagen, Denmark, \texttt{Email: ambe@di.ku.dk}. Supported by Srikanth Srinivasan's start-up grant from the University of Copenhagen.}} 
\and 
{Nutan Limaye \thanks{IT University of Copenhagen, Denmark, \texttt{Email: nuli@itu.dk}. Supported by Independent Research Fund Denmark (grant agreement No. 10.46540/3103-00116B) and is also supported by the Basic Algorithms Research Copenhagen (BARC), funded by VILLUM Foundation Grant 54451.}} 
\and 
{Varun Ramanathan \thanks{School of Technology and Computer Science, Tata Institute of Fundamental Research, Mumbai, India, \texttt{Email: varun.ramanathan@tifr.res.in}. Supported by the Department of Atomic Energy, Government of India, under project number RTI400112. A part of the work was done when the author was visiting the University of Copenhagen and was supported by the European Research Council (ERC) under grant agreement no. 101125652 (ALBA).}} 
\and 
{Srikanth Srinivasan \thanks{Department of Computer Science, University of Copenhagen, Denmark, \texttt{Email: srsr@di.ku.dk}. Supported by the European Research Council (ERC) under grant agreement no. 101125652 (ALBA).}}
}
\date{\today}
\maketitle

\begin{abstract}
In this work, we prove upper and lower bounds over fields of positive characteristics for several fragments of the Ideal Proof System (IPS), an algebraic proof system introduced by Grochow and Pitassi~(J. ACM 2018). Our results extend the works of Forbes, Shpilka, Tzameret, and Wigderson (Theory of Computing 2021) and also of Govindasamy, Hakoniemi, and Tzameret (FOCS 2022). These works primarily focused on proof systems over fields of characteristic 0, and we are able to extend these results to positive characteristic.\\ 

The question of proving general $\mathrm{IPS}$ lower bounds over positive characteristic is motivated by the important question of proving $\mathsf{AC}^{0}[p]$-Frege lower bounds. This connection was observed by Grochow and Pitassi (J. ACM 2018). Additional motivation comes from recent developments in algebraic complexity theory due to Forbes (CCC 2024)  who showed how to extend previous lower bounds over characteristic $0$ to positive characteristic.\\ 

In our work, we adapt the functional lower bound method of Forbes et al. (Theory of Computing 2021) to prove exponential-size lower bounds for various subsystems of $\mathrm{IPS}$. In order to establish these size lower bounds, we first prove a tight degree lower bound for a variant of \emph{Subset Sum} over positive characteristic. This forms the core of all our lower bounds.\\ 

Additionally, we derive upper bounds for the instances presented above. We show that they have efficient constant-depth $\mathrm{IPS}$ refutations. This demonstrates that constant-depth $\mathrm{IPS}$ refutations are stronger than the proof systems considered above even in positive characteristic. We also show that constant-depth $\mathrm{IPS}$ can efficiently refute a general class of instances, namely all symmetric instances, thereby further uncovering the strength of these algebraic proofs in positive characteristic.\\ 

Notably, our lower bounds hold for fields of arbitrary characteristic but require the field size to be $n^{\omega(1)}$. In a concurrent work, Elbaz, Govindasamy, Lu, and Tzameret have shown lower bounds against restricted classes of $\mathrm{IPS}$ over finite fields of any size by considering different hard instances.
\end{abstract}

        \newpage 
        
\tableofcontents

        \newpage

\section{Introduction}
\label{sec:intro}
\paragraph{Propositional Proof Systems.} A proof system consists of a set of axioms and inference rules. The goal is to start with the given set of axioms and apply the inference rules repeatedly to prove theorems (tautologies) within the proof system.  A proof system is \emph{sound} if it proves only true statements and it is \emph{complete} if it proves all true statements. The area of \emph{Propositional Proof Complexity} aims to understand the strength of different proof systems in the propositional setting. In a foundational work, Cook and Reckhow \cite{CR79} showed that if we could prove that there exist tautologies such that they require exponential proof size (i.e., vaguely the number of times different inference rules are applied in the proof) in any proof system, then it would resolve the famous NP vs. coNP question in computational complexity theory. 


Apart from the connection to this central question in complexity theory, understanding the power of different proof systems is also fundamental to mathematical reasoning. This has motivated a lot of research in the area for the last five decades. (See for instance these reference texts for more context~\cite{krabook95,CloteKSurvey02,krabook2019}.) There are many different kinds of propositional proof systems based on the set of axioms they start with and the kind of inference rules they are allowed to use. In this work, we will focus on algebraic proof systems. In algebraic proof systems, propositional tautologies are expressed as an unsatisfiable set of polynomial equations and the inference rules are algebraic, i.e. they involve reasoning based on polynomial arithmetic.

The study of algebraic proof systems originates from the work of Beame, Impagliazzo, Kraj\'{i}\v{c}ek, Pitassi, and Pudl\'{a}k~\cite{BIKPP94}  who introduced the Nullstellensatz proof system (based on Hilbert's Nullstellensatz). Their work was followed by the work of Clegg, Edmonds, and Impagliazzo~\cite{CEI96} who introduced Polynomial Calculus as a \emph{dynamic} variant of the Nullstellensatz proof system. Over the years, substantial work on these proof systems has helped us get a good understanding of their power in terms of complexity measures such as sparsity and degree~\cite{BIKPP94,BIPRS97,Razborov98,Grigoriev98,IPS99,BGIP01,AR2001}. 

However, as noted in~\cite{FSTW21}, sparsity and degree only roughly capture the complexity of algebraic proofs. More recently, Grochow and Pitassi \cite{GP14} proposed the Ideal Proof System (IPS) as a natural generalization of these well-studied algebraic proof systems such as Polynomial Calculus and Nullstellensatz proof systems. In the last decade, several papers studied this proof system. (See for instance~\cite{GP14,PitassiTzameretSurvey,FSTW21,GHT,HLT24}.) This has allowed us to understand many other aspects of algebraic proofs, such as proof size and proof depth. 

In this paper, we extend this line of work. Specifically, we revisit some of the known upper and lower bounds for Ideal Proof Systems over characteristic $0$ and show similar bounds over fields of any characteristic\footnote{In all the results mentioned here, when we say that a result holds over characteristic $0$, it in fact holds over large enough characteristic as well.}.

\subsection{Ideal Proof Systems}
\label{sec:intro-ips}

We start by describing the general setup for an algebraic (static\footnote{In the literature, the following type of proof system is often referred to as a static proof system. There are other algebraic proof systems, where the proof is presented line-by-line and those are known as dynamic proof systems. Here, we will only discuss static proof systems.}) proof system. Let $\xbar$ denote the set of variables $\{x_1, x_2, \ldots, x_n\}$. We are given a set of polynomial axioms $f_1(\mathbf{x}), f_2(\mathbf{x}), \ldots, f_m(\mathbf{x}) \in$ $\mathbb{F}[\mathbf{x}]$ and the goal is to show that there is no $0$-$1$ assignment to the variables such that it simultaneously satisfies $\{f_1(\mathbf{x}) = 0, f_2(\mathbf{x}) = 0, \ldots, f_m(\mathbf{x}) = 0\}$ over $\mathbb{F}$. To force a common Boolean solution, the set of axioms is appended with additional axioms, $\{x_i^2 - x_i =0\}_{i \in [n]}$ for $i \in [n]$. These are called the \emph{Boolean axioms.} 

Based on Hilbert's Nullstellensatz, we know that if $\{f_1(\mathbf{x}) = 0, f_2(\mathbf{x}) = 0, \ldots, f_m(\mathbf{x}) = 0\}$ $\cup$ $\{x_i^2 - x_i=0\}_{i \in [n]}$ are simultaneously not satisfiable, then such a refutation\footnote{The words `proofs' and `refutations' are treated interchangeably in this paper. What we will be `proving' is a statement that `refutes' the existence of a common solution to a system of equations.} can be given by polynomials $A_1(\mathbf{x}), A_2(\mathbf{x}), \ldots, A_m(\mathbf{x})$ and $B_1(\mathbf{x}), B_2(\mathbf{x}), \ldots, B_n(\mathbf{x})$ such that
\begin{equation}
    \sum_{i \in [m]} A_i(\mathbf{x}) \cdot f_i(\mathbf{x}) +  \sum_{i \in [n]} B_i(\mathbf{x}) \cdot (x_i^2 - x_i) = 1.
\end{equation}
The complexity of such a proof can be defined using complexity parameters of the polynomials $\{A_i(\xbar)\}$ and $\{B_i(\xbar)\}$. 
In the case of the Ideal Proof System, Grochow, and Pitassi proposed that we assume that $A_i(\xbar), B_i(\xbar) \in \mathbb{F}[\xbar]$ are computed by algebraic circuits. (See \Cref{sec:results} for the formal definition.) Based on this, they defined complexity measures such as circuit size and circuit depth of IPS. 

This proof system in its full generality is known to be quite strong. Specifically, it can polynomially simulate Extended Frege~\cite{GP14}, which is one of the most powerful among well-studied propositional proof systems. Additionally, the same work also showed that proving lower bounds for this proof system would also imply strong algebraic circuit lower bounds, which is also a very challenging problem. 

 In light of this (and other reasons explained below), many restricted variants of the IPS have been studied. Let $\mathcal{C}$ be a class of polynomials. Then, a $\mathcal{C}$-IPS refutation is an IPS-refutation wherein $\{A_i(\xbar)\}_{i \in [m]}$ and $\{B_i(\xbar)\}_{i \in [n]}$ belong to the class $\mathcal{C}$. Forbes, Shpilka, Tzameret, and Wigderson~\cite{FSTW21}, as well as Govindasamy, Hakoniemi, and Tzameret~\cite{GHT}, considered different classes of polynomials, for example, the class of polynomials computed by read-once oblivious algebraic branching programs (roABPs), by multilinear formulas, or by constant-depth algebraic formulas. They proved upper and lower bounds on the size of (some variants of) $\mathcal{C}$-IPS refutations over characteristic $0$. 
\subsection{Motivation}
\label{sec:motivation}
We extend these works and prove similar bounds in arbitrary characteristic. 
Our work is motivated by the following important strands of research in proof complexity. 

\paragraph{IPS-refutations and $\ACz[p]$-Frege.}
A long-standing open question in proof complexity, open for almost three decades~\cite{kra2015}, is to prove superpolynomial lower bounds against $\ACz[p]$-Frege proof systems, i.e., a proof system in which the lines of the proof are constant-depth Boolean circuits that use modular gates. In the late 80s, Razborov \cite{Razborov1987} and Smolensky \cite{Smolensky1987, Smolensky1993} resolved the Boolean circuit lower bound question for $\ACz[p]$, but the corresponding proof complexity question has proved to be elusive. 

Over the years, several attempts have been made to resolve this question. The most relevant to our work is the result by Grochow and Pitassi \cite[Theorem 3.5]{GP14} which showed that constant-depth-IPS over characteristic $p$ can efficiently simulate $\ACz[p]$-Frege proofs. This means that proving superpolynomial lower bounds against constant-depth-IPS refutations will give superpolynomial lower bounds against $\ACz[p]$-Frege. This gives a strong motivation to prove IPS lower bounds over small characteristics.

\paragraph{Functional lower bounds over any characteristic.}
Building on the work of \cite{GP14}, \cite{FSTW21} further explored the power of IPS refutations. They proposed a concrete approach towards proving size lower bounds for IPS refutations via 
\emph{functional lower bounds} (further explained in \Cref{sec:techniques}). Their method was inspired by the notion of functional lower bounds in Boolean circuit complexity \cite{Grigoriev-Razborov, FKS16}. They demonstrated the promise of their method by proving several lower bounds for different fragments of IPS.


For example, the strong algebraic complexity lower bounds known for $\roABP$s \cite{Nisan} and multilinear formulas \cite{Raz-2009} follow from understanding the \emph{evaluation dimension} complexity measure in these models. Since this measure is essentially functional in nature, \cite{FSTW21} used it to successfully prove lower bounds for $\mathcal{C}$-IPS when $\mathcal{C}$ is a class of read-once branching programs or multilinear formulas. Their bounds are over characteristic $0$.  



This approach of~\cite{FSTW21} was further adapted by Govindasamy, Hakoniemi, and Tzameret~\cite{GHT} to prove superpolynomial lower bounds against (multilinear) constant-depth-IPS refutations. Their proof builds on some of the key components of the superpolynomial lower bound against constant-depth algebraic circuits by Limaye, Srinivasan, and Tavenas. The latter lower bound of \cite{LST} only worked over characteristic $0$; for this and other reasons, the result of \cite{GHT} was also limited to characteristic $0$. In a recent paper, however, Forbes \cite{Forbes-LST-CCC} improved the circuit lower bound result of \cite{LST} and proved the same\footnote{Some parameters in the lower bound by \cite{LST} were subsequently improved by \cite{BDS24} and \cite{Forbes-LST-CCC} achieves those improved parameters.} lower bound over any characteristic. 

In light of these results, the next obvious step is to prove the lower bounds of~\cite{FSTW21,GHT} over any characteristic. We achieve that in this work.\footnote{The subset-sum instances from \cite{FSTW21,GHT} are not always unsatisfiable over fields of positive characteristic; this requires that we tweak their instances to ensure unsatisfiability. Barring these changes, we qualitatively match their lower bounds over fields of positive characteristic.}




\subsection{Our Results}
\label{sec:results}
To describe our results, we start with the formal definitions of IPS refutations and its variants. 

\begin{definition}[$\IPS$ proof systems \cite{GP14, FSTW21}]
\label{def:proof-systems}
Let $f_1,\ldots, f_m \in \F[x_1,\ldots, x_n]$ be a system of unsatisfiable polynomials over the Boolean cube $\Boo^{n}$. In other words, there is no Boolean assignment $\mathbf{a} \in \{0,1\}^n$ to the variables $x_1,\ldots, x_n$ so that $f_i(\mathbf{a}) = 0$ for all $i\in [m].$

Given a class of algebraic circuits $\mathcal{C}$, a \emph{$\mathcal{C}$-$\IPS$ refutation} of the system of equations defined by $f_1,\ldots,f_m$ is an algebraic circuit $C\in \mathcal{C}$ in variables $x_1,\ldots,x_n,y_1,\ldots, y_m, z_1,\ldots, z_n$ such that
\begin{itemize}
    \item $C(\mathbf{x},\mathbf{0},\mathbf{0}) = 0$, and
    \item $C(\mathbf{x},f_1,\ldots,f_m,x_1^2-x_1,\ldots,x_n^2-x_n) = 1.$
\end{itemize}
The size of the refutation is the size of the circuit $C.$

Further, if the circuit $C$ has individual degree at most $1$ in the variables $\mathbf{y}$ and $\mathbf{z}$, then we say that $C$ is a \emph{$\mathcal{C}$-$\IPSLIN$ refutation}. If the circuit $C$ has individual degree at most $1$ in the variables $\mathbf{y}$ (but not necessarily in $\mathbf{z}$), then $C$ is said to be a \emph{$\mathcal{C}$-$\IPSLINp$ refutation.}

Finally, we say that a circuit $C\in \mathcal{C}$ is a \emph{multilinear $\mathcal{C}$-$\IPSLINp$ refutation} if additionally $C(\mathbf{x},\mathbf{y},\mathbf{0})$ is a multilinear polynomial in the variables $\mathbf{x}\cup \mathbf{y}$.
\end{definition}

\begin{remark}
    \label{rem:def-pf-systems}
    We mostly employ the above definition in the case that $m=1$, i.e. the case when we have a single polynomial equation that is unsatisfiable over the Boolean cube. Further, while our upper bound results are proved in the more restrictive $\mathcal{C}$-$\IPSLIN$ proof system, our lower bounds results hold in the setting of the stronger $\mathcal{C}$-$\IPSLINp$ proof systems. 
\end{remark}

We also recall some standard notions about polynomials and algebraic models of computation, which will be useful below. 

\paragraph{Multilinear and symmetric polynomials.}
A polynomial $f(\mathbf{x}) \in \F[x_{1},\ldots,x_{n}]$ is a \emph{multilinear} if the individual degree is at most $1$. For a polynomial $f(\mathbf{x})$, the \emph{multilinearization} operator, denoted by $\ml[~\cdot~]$, changes for each variable $x_{j}$ and any $k$, every occurrence of $x_{j}^{k}$ in $f(\mathbf{x})$ to $x_{j}$.

A polynomial $f(\mathbf{x}) \in \F[x_{1},\ldots,x_{n}]$ is said to be a \emph{symmetric polynomial} if the polynomial remains invariant under any permutation of the input variables. For a degree parameter $0 \leq d \leq n$, the $d^{th}$ elementary symmetric polynomial $e_{n,d}(x_{1},\ldots,x_{n})$ is defined to be the following multilinear polynomial $e_{n,d}(x_{1},\ldots,x_{n}) \; = \; \sum_{\substack{S \subseteq [n] \\ |S| = d}} \, \prod_{i \in S} x_{i}$. Whenever $n$ is clear from the context, we will denote the $d^{th}$ elementary symmetric polynomial by $e_{d}(\mathbf{x})$.

\paragraph{Algebraic models of computation.}
We recall the definitions of some of the standard models of computation relevant to our results.

\textit{Algebraic circuits and formulas.} An \emph{algebraic circuit} is a directed acyclic graph in which each node either computes a sum (or a linear combination) of its inputs, or a product of its inputs. The leaf nodes are either variables or constants. The size of an algebraic circuit is the number of edges in the circuit, and the depth of an algebraic circuit is the longest path from the output node (a sink) to a leaf node (a source). An \emph{algebraic formula} is an algebraic circuit where the output of each node feeds into at most another node; in other words, the underlying graph of an algebraic formula is a tree. An algebraic formula is a \emph{multilinear formula} if every gate of the formula computes a multilinear formula.

\textit{Sparse polynomials and constant-depth circuits.} The class $\sum\prod$ consists of depth-2 formulas with an addition gate in the top layer and multiplication gates in the bottom (second) layer. All the gates have unbounded fan-in. $\sum\prod$ formulas essentially compute polynomials in the \emph{sparse} representation i.e. as a sum of monomials. In general, a constant-depth algebraic circuit has $O(1)$ alternating layers of additional and multiplication gates.


\textit{Read-Once Oblivious Algebraic Branching Programs.}
    A read-once oblivious algebraic branching program in the variable-order $\pi \in \mathcal{S}_n$\footnote{$\mathcal{S}_n$ denotes the set of all permutation of $[n]$.} is a directed acyclic graph whose vertices are partitioned into $n$ layers $V_0 = \{s\}, V_1, V_2, \dots, V_n = \{t\}$. For each $i \in \{1, 2, \dots, n\}$, there are edges directed from layer $V_{i-1}$ to $V_{i}$ that are labelled by univariate polynomials in the variable $x_{\pi(i)}$. For each $s$-to-$t$ path $p$, the polynomial computed by $p$ is defined to be product of the edge labels on $p$. The polynomial computed by the $\roABP$ is defined as the sum of polynomials computed by all $s$-to-$t$ paths. The \emph{width} of an $\roABP$ is $\max_{0\leq i \leq n} |V_i|$ i.e. the size of the largest layer of vertices.

For more background on these models of computation, please refer to one of the standard surveys in algebraic complexity (\cite{SY10},\cite{saptharishisurvey}).

\subsubsection{Lower Bounds Over Positive Characteristic}
\label{sec:intro-lower}

We start by stating our lower bound results. 

\begin{theorem}[Lower bounds for sparse-$\IPSLINp$ in positive characteristic]
    \label{thm:lbd-sparse}
    The following holds for any large enough $n$. Let $p$ be any prime number. Let $k \in \mathbb{N}$ such that $p^k > 2^{\Omega(n)}$. There exist $\alpha_{i}\in \F_{p^{k}}$ and $\beta \in \F_{p^{{{2k}}}}\setminus \F_{p^{k}}$ 
    such that 
    \begin{itemize}
        \item The polynomial $f = \sum_{i \in [n]}\alpha_{i} x_i - \beta$ has no Boolean satisfying assignment.
        \item Any sparse-$\IPSLINp$ refutation\footnote{Note that sparse-$\IPSLIN$ (a weaker system than sparse-$\IPSLINp$) is equivalent to the Nullstellensatz proof system of \cite{BIKPP94}.} of $f$ must have size at least $2^{\Omega(n)}$ 
    \end{itemize}
\end{theorem}
Note that the hard instance above is a sparse polynomial. We show that it has no small sparse refutation over positive characteristic. 

\begin{theorem}[Lower bounds for fixed-order $\roABP$ in positive characteristic]
    \label{thm:lbd-fixed-order-roabp}
    The following holds for any large enough $n$. Let $p$ be any prime number. Let $k \in \mathbb{N}$ such that $p^k > 2^{\Omega(n)}$.  There exist $\alpha_{i}\in \F_{p^k}$ and $\beta \in \F_{p^{2k}}\setminus \F_{p^k}$ such that 
    \begin{itemize}
        \item The polynomial $f = \sum_{i \in [n]}\alpha_{i} x_i y_i - \beta$ has no Boolean satisfying assignment.
        \item Any $\roABP$-$\IPSLINp$ refutation of $f$ in any order of variables where $\mathbf{x}$ variables come before $\mathbf{y}$ variables, must have width $2^{\Omega(n)}$. 
    \end{itemize}
\end{theorem}

To obtain lower bounds against more powerful models such as $\roABP$-$\IPSLINp$ with respect to \emph{any} order, or multilinear formulas, \cite{FSTW21} used a slightly modified hard instance. We also use an instance the same as theirs up to the choice of coefficients. 

\begin{theorem}
[Lower bounds for any order roABP-$\IPSLINp$ and multilinear-formula-$\IPSLINp$]
    \label{thm:lbd-roabp}
    The following holds for any large enough $n$. Let $p$ be any prime number. Let $k \in \mathbb{N}$ such that $p^k > 2^{\Omega(n)}$. There exist $\alpha_{i,j}\in \F_{p^k}$ and $\beta \in \F_{p^{2k}}\setminus \F_{p^k}$ such that 
    \begin{itemize}
        \item The polynomial $f = \sum_{1\leq i < j\leq n}\alpha_{i,j} z_{i,j} x_i x_j - \beta$ has no Boolean satisfying assignment.
        \item Any $\roABP$-$\IPSLINp$ refutation of $f$ must have size at least $2^{\Omega(n)}$.
        \item Moreover, any multilinear-formula-$\IPSLINp$ refutation of $f$ must have size at least $n^{\Omega(\log n)}$ and for $\Delta = o(\log n / \log \log n)$, any product-depth\footnote{The product-depth of a circuit is the maximum number of product gates appearing in any leaf-to-root path.}-$\Delta$ multilinear-formula-$\IPS$ refutation requires size $\geq n^{\Omega\left(\frac{1}{\Delta^2}(\frac{n}{\log n})^{1/\Delta}\right )}$. 
    \end{itemize}
\end{theorem}

Again notice that, $f$  is a sparse polynomial and hence has a polynomial size $\roABP$. It is also efficiently computable by a multilinear formula. 

In general, in Boolean proof complexity, it is typical that the hard-to-refute instances are themselves easy to compute. 
In algebraic proof complexity, there are some lower bound results that do not have this property. That is, the instances that are hard to refute are also hard to compute. For example, the set of results obtained by the approach of multiples in \cite[Theorem 1.18, Theorem 1.19, Theorem 1.20]{FSTW21} and in a paper by Andrews and Forbes \cite{AF22}. 
%
%
 Additionally, in a recent work Hakoniemi, Limaye, and Tzameret \cite{HLT24} presented instances that were hard to refute for $\roABP$-$\IPSLINp$ and for multilinear-formula-$\IPSLINp$ over any characteristics, i.e., similar to what we prove here. However, unfortunately, their instances were hard to compute and specifically, they could not be computed by $\roABP$ or by multilinear formulas. Hence, our result here have the best of both the worlds; the lower bounds hold over any characteristic and the hard instances are easy to compute.

\begin{theorem}[Lower bounds for multilinear constant-depth-$\IPSLINp$ in positive characteristic]
    \label{thm:lbd-const-depth}
    The following holds for any large enough $n$. Let $p$ be any prime and let $k \in \mathbb{N}$ be large enough so that $p^k > 2^{\Omega((\log n)^2)}$. There exist $\alpha_{i,j,k,\ell}\in \F_{p^k}$ and $\beta \in \F_{p^{2k}}\setminus \F_{p^k}$ such that 
    \begin{itemize}
        \item The polynomial $f = \sum_{1\leq i < j < k < \ell \leq n}\alpha_{i,j,k,\ell} z_{i,j,k,\ell} x_i x_j x_k x_\ell - \beta$ has no Boolean satisfying assignment.
        \item Any multilinear constant-depth-$\IPSLINp$ refutation of $f$ must have size $n^{\omega(1)}$.
    \end{itemize}
\end{theorem}

The characteristic $0$ (or large characteristic) version of the above theorem was presented in~\cite{GHT}. Their lower bound is a step towards constant-depth-$\IPS$ lower bounds. Our result above can thus be thought of as another step forward in the right direction.
Moreover, our input instance is the same as the input instance in Theorem~1~\cite{GHT} up to the choice of coefficients, and it is easy to compute (while being hard to refute). More specifically, it is computable by polynomial-sized constant-depth multilinear formulas. 

\begin{remark}
    In all our results, the field characteristic is arbitrary, but the field size is quite large, i.e., $p^k$ is either exponential or superpolynomial.  This setting is non-trivial because the field elements have polynomial bit complexity. Other results in the area, such as the work of Alekseev, Grigoriev, Hirsch, and Tzameret~\cite{AGHT} similarly use polynomial constraints with coefficients from exponentially large domains. Specifically~\cite{AGHT} study a variant of the subset sum instance, called the Binary Value Principle, $\sum_{i \in [n]} 2^{i-1} x_i + 1 = 0$ in the context of IPS proof systems in fields of characteristic zero. 
    
    It is an interesting open question to prove similar $\IPS$ lower bounds over finite fields of small size. Unfortunately, as we show below, this forces the polynomial instances to become more complicated. See \Cref{sec:related} for recent independent work that makes progress in this direction.
\end{remark}

\subsubsection{Upper Bounds Over Positive Characteristic}
\label{sec:intro-upper}

A natural question for hard instances above is: what is the weakest proof system in which they are efficiently refutable? In personal communication, Tzameret observed that the above instances were refutable by constant-depth-$\IPSLIN$ hence showing that these proof systems can be exponentially more succinct than their multilinear counterpart. The theorem below shows that the above polynomials have efficient constant-depth-$\IPSLIN$ refutations, even in the setting of positive characteristic.


\begin{restatable}[Upper bounds for (non-multilinear) constant-depth-$\IPSLIN$]{theorem}{ubdnonmlposchar}\label{thm:ubd-const-depth}
Fix a prime number $p$. The following holds for any natural numbers $n$ and $k.$\newline
Let $f\in \F_{p^{k}}[x_1,\ldots,x_n]$ be any polynomial with sparsity $s$ and degree $D$ with coefficients from the field $\F_{p^k}$ and let $\beta$ be any element of $\F\setminus\F_{p^k}$ where $\F$ is a field extension of $\F_{p^k}.$\newline
Then,
\begin{itemize}
    \item The polynomial $f(\mathbf{x}) - \beta$ has no satisfying assignment over the Boolean cube $\Boo^{n}$
    \item  There is a constant-depth-$\IPSLIN$ refutation of degree $O(k\cdot p\cdot D)$ and size $\poly(s,p).$ 
\end{itemize}
\end{restatable}
Note that since $\beta \notin \F_{p^{k}}$, the polynomial $f(\mathbf{x}) - \beta$ does not have a zero over $\Boo^{n}$ (in fact it does not have a solution over $\F_{p^{k}}^{n}$). So the first item of above follows immediately. We also give non-trivial constant-depth-$\IPSLIN$ refutations for degree-$1$ polynomials that are unsatisfiable over $\Boo^{n}$ with all the coefficients in the same field.\\

\begin{restatable}[Upper bound on degree of Nullstellensatz certificate]{theorem}{degreeupperbound}\label{thm:ub-degree}
Fix a prime $p$. The following holds for any natural numbers $n$ and $k$ with $n > kp$.\newline
The following holds for every $\alpha_{1},\ldots,\alpha_{n},\beta \in \F_{p^{k}}$. Suppose the degree-$1$ polynomial $\sum_{i=1}^{n} \alpha_{i} x_{i} - \beta \in \F_{p^{k}}[x_{1},\ldots,x_{n}]$ is unsatisfiable over the Boolean cube $\Boo^{n}$ (i.e. there does not exist a Boolean point $\mathbf{a} \in \Boo^{n}$ such that $\sum_{i=1}^{n} \alpha_{i} a_{i} - \beta = 0$).\newline
Then, there is a constant-depth-$\IPSLIN$ refutation of degree $\bigO(k \cdot p)$ and size $\bigO(n/kp)^{\bigO(kp)}$.\\

\noindent
In particular, if $p = \bigO(1)$ and $k = o(n)$, then there is a constant-depth-$\IPSLIN$ refutation of degree $o(n)$ and size $2^{o(n)}$.
\end{restatable}

\noindent
Note that for degree-$1$ polynomials, the difference in \Cref{thm:ubd-const-depth} and \Cref{thm:ub-degree} is in the constant-term $\beta$. If every $\alpha_{i} \in \F_{p^{k}}$ and $\beta \notin \F_{p^{k}}$, then the polynomial is always unsatisfiable over $\Boo^{n}$ (no matter the choice of $\alpha_{i}$'s and $\beta$). In fact, it is unsatisfiable over $\F_{p}^{n}$. Our proof of \Cref{thm:ubd-const-depth} leverages this and yields an efficient refutation. However, if $\beta \in \F_{p^{k}}$, then our proof of \Cref{thm:ubd-const-depth} falls apart. We handle this separately in \Cref{thm:ub-degree}, but we do not match \Cref{thm:ubd-const-depth} qualitatively. More precisely, \Cref{thm:ubd-const-depth} yields a $\mathrm{poly}(n,p)$-sized non-multilinear constant-depth refutations, but \Cref{thm:ub-degree} yields a roughly $\binom{n}{k}$-sized non-multilinear constant-depth refutations.


\begin{remark}\label{remark: exp field size and multilinearity not artifacts}
    Suppose the characteristic $p$ is a fixed prime independent of the number of variables $n$.
    \begin{itemize}
        \item \Cref{thm:ubd-const-depth} shows that the exponential field size in \Cref{thm:lbd-sparse}, \Cref{thm:lbd-fixed-order-roabp} and \Cref{thm:lbd-roabp} is not an artifact of the proofs.\footnote{Suppose the field $\F_{p^k}$ is not large enough, say, $k = o(n)$. Then there is a refutation of degree $d = O(k\cdot p\cdot D)$, which is $o(n)$ when $p$ and $D$ are constants. In particular, the sparsity of the refutation is at most $\binom{n+d}{d}$, which is $2^{o(n)}$ when $d = o(n)$.}. For fields of subexponential size, the polynomials in these theorems have refutations of degree $o(n)$ and in particular have $\mathrm{roABP}$-$\IPSLIN$ refutations of size $2^{o(n)}.$ \footnote{When the characteristic $p$ is a growing function of $n$, this argument breaks down. It might be possible to get rid of the exponential field size.}
        \item \Cref{thm:ubd-const-depth} also shows that the multilinearity assumption in \Cref{thm:lbd-const-depth} is not an artifact of the proof. Non-multilinear proofs, even over large fields, allow efficient constant-depth refutations for sparse instances. 
    \end{itemize}
\end{remark}

Our final result shows a constant-depth upper bound for multilinear and \emph{symmetric} systems of polynomials, i.e. systems defined by polynomials $f(x_1,\ldots,x_n)$ of the form 
\[
\sum_{d=1}^n \alpha_d e_{n,d} + \alpha_0
\]
where $e_{n,d}$ denotes the elementary symmetric polynomial of degree $d$ in variables $x_1,\ldots,x_n.$ Such polynomial systems have been employed in~\cite{FSTW21} to prove lower bounds against restricted systems of constant-depth-$\IPSLIN.$ Our results imply that general constant-depth circuit refutations can be exponentially more succinct than these restricted families, even for positive characteristic.

\begin{restatable}[Upper bounds for multilinear symmetric systems]{theorem}{symconstant}\label{thm:ubd-sym-const-depth}
Fix a field $\F$. Let \break $f_1,\ldots,f_m \in \F[x_{1},\ldots,x_{n}]$ be a family of multilinear and symmetric polynomials with no common Boolean solution i.e. there does not exist a $\mathbf{x} \in \Boo^{n}$ such that each $f_{i}(\mathbf{x}) = 0$. This system has a constant-depth-$\IPSLIN$ refutation of size $\bigO(m^{2} n^{5} \log n)$ and depth $8$.
\end{restatable}

\subsection{Proof Techniques}
\label{sec:techniques}
\paragraph{Lower bounds.}
Our proof uses the functional lower bound method introduced by~\cite{FSTW21}, which can be described as follows. We know that a $\mathcal{C}$-$\IPSLINp$ refutation for $f(\xbar)$ consists of $A(\xbar)$, $B_{i}(\xbar) \in \mathbb{F}[\xbar]$ such that \[f(\xbar) \cdot A(\xbar) + \sum_{i \in [n]} (x_i^2 - x_i) \cdot B_{i}(\xbar) = 1,\] where $A(\xbar), B_{1}(\xbar), \ldots, B_{n}(\xbar)$ belong to $\mathcal{C}$. As $f(\xbar)$ is unsatisfiable over the Boolean hypercube, this implies that over the Boolean hypercube, $A(\xbar)$ is a well-defined reciprocal of $f(\xbar)$. Hence, to show that $A(\xbar)$ cannot belong to $\mathcal{C}$, it is enough to show that any polynomial that agrees with $1/f(\xbar)$ cannot be computed by $\mathcal{C}$. That is, the problem of proving a lower bound on the size of $\mathcal{C}$-$\IPSLINp$ is reduced to proving a functional lower bound for $1/f(\xbar)$. 

At the heart of such a functional lower bound lies a \emph{degree lower bound}, i.e., a lower bound on the degree of $\tilde{f}(\xbar)$, where $\tilde{f}(\xbar)$ and $f(\xbar)$ are related. In fact, $f(\xbar)$ is a \emph{lifted} version of $\tilde{f}(\xbar)$. Once we have such a degree lower bound for $\tilde{f}(\xbar)$, we can apply proof ideas from algebraic complexity theory such as the rank-based lower bound methods. These methods allow for the degree lower bounds for $\tilde{f}(\xbar)$ to be lifted to size lower bounds for $f(\xbar)$.

For their machinery to work over positive characteristic, 
we prove a \emph{positive characteristic} version of the degree lower bound (see \Cref{lem:basic degree lower bound} for the formal statement). In the case of the lower bound argument in \cite{FSTW21}, it was important to obtain a tight degree lower bound of exactly $n$. They needed it for the next step, i.e., \emph{lifting}, to work. In our case, we show that such a degree lower bound holds with high probability (over the choice of coefficients of the hard instance). Once we have the degree lower bound, the rest of the lower bound proof works similar to the proof by~\cite{FSTW21}.

\paragraph{Upper bounds.}
We now describe the main ingredients in our upper bounds. We start by describing the main ideas in the proof of \Cref{thm:ubd-const-depth}. 

\paragraph{Constant-depth upper bounds.} Here, we proceed in two steps. First, we observe that for any sparse polynomial of degree $d$, we can \emph{flatten} it to a linear polynomial by renaming the monomials by fresh variables. Our hard instance is indeed sparse, hence the observation can be used to rewrite the polynomial as a linear polynomial over a fresh set of variables. 

Now, consider a linear polynomial $L(\xbar) - \beta$ such that $L(\xbar) = \alpha_1 x_1 + \alpha_2 x_2 + \ldots + \alpha_n x_n$, where $\alpha_1, \ldots, \alpha_n \in \mathbb{F}_{p^k}$ for some $k$ and prime $p$ and $\beta \in \mathbb{F}\setminus \mathbb{F}_{p^k}$ such that it is not satisfiable over $0$-$1$ assignments.

To prove that the polynomial has a refutation over constant-depth circuits, we first prove that for every $j$, $L_j(\xbar) = \alpha_1^{p^j} x_1 + \alpha_2^{p^j} x_2 + \ldots + \alpha_n^{p^j} x_n -  \beta^{p^{j}}$ can be expressed as a multiple of $L(\xbar)$ modulo the ideal $\xbar^p - \xbar$, which is a shorthand for the ideal generated by $\{x_i^p - x_i\}_{i \in [n]}$. 

We then observe that for $j =k$, $L_k(\xbar) - L(\xbar)$ is a non-zero constant and use this observation to construct small depth circuits for the refutation of $L(\xbar) - \beta$. Throughout, we use some standard but useful tricks available to positive characteristic fields.\\

For the proof of \Cref{thm:ub-degree}, we observe that the multilinear part of $(f(\mathbf{x})-\beta)^{-1}$ has degree $\bigO(kp)$. This follows from Fermat's Little Theorem and using basic properties about multilinearization. See \Cref{subsec:proof-ub-degree} for complete details.

\paragraph{Upper bounds for symmetric polynomials}
Now we discuss the proof outline for \Cref{thm:ubd-sym-const-depth}. For ease of exposition, we explain the ideas for the case of $m = 1$ in \Cref{thm:ubd-sym-const-depth}, i.e. there is one multilinear symmetric polynomial $f(\mathbf{x})$ that does not have a solution over the Boolean cube $\Boo^{n}$. Suppose $\F$ has characteristic $p > 0$. Any symmetric polynomial is a polynomial of the $n$ elementary symmetric polynomials\footnote{This follows from the Fundamental Theorem of Symmetric Polynomials.} i.e. $e_{1}(\mathbf{x}), \ldots, e_{n}(\mathbf{x})$. However, if we restrict to the Boolean cube $\Boo^{n}$, then any symmetric polynomial is a polynomial of just $\bigO(\log n)$ elementary symmetric polynomials. Let $\widehat{\mathbf{e}}(\mathbf{x})$ denotes the tuple of those $\bigO(\log n)$ elementary symmetric polynomials (see \Cref{claim:sym-char-p} for an explicit description of $\widehat{\mathbf{e}}(\mathbf{x})$.)

Let $F(\mathbf{y})$ be the $\bigO(\log n)$ variate polynomial such that $F(\mathbf{y}) \circ \widehat{\mathbf{e}}(\mathbf{x})$ agrees with $f(\mathbf{x})$ on the Boolean cube $\Boo^{n}$. The Boolean cube $\Boo^{n}$ is mapped to $\F_{p}^{\bigO(\log n)}$ under the map $\widehat{\mathbf{e}}(\mathbf{x})$ because $\mathrm{char}(\F) = p$. The unsatisfiability of $f(\mathbf{x})$ over the Boolean cube $\Boo^{n}$ implies the unsatisfiability of $F(\mathbf{y})$ over $\F_{p}^{\bigO(\log n)}$. Applying Hilbert's Nullstellensatz Theorem (see \Cref{thm:nullstellensatz}) on the unsatisfiability\footnote{To capture the restriction of $\F_{p}^{n}$, we add $n$ univariate polynomials, each of which vanishes on one coordinate of $\F_{p}^{n}$.} of $F(\mathbf{y})$ over $\F_{p}^{\bigO(\log n)}$, we get a \emph{low-variate} Nullstellensatz certificate (it is a Nullstellensatz certificate in just $\bigO(\log n)$ variables)\footnote{Loosely speaking, one can imagine this as a ``dimension reduction'' of our problem. The symmetric structure of $f(\mathbf{x})$ led us to convert a problem in $n$ variables to a problem in just $\bigO(\log n)$ variables.}. The coefficients of this low-variate Nullstellensatz certificate can be computed via $\mathrm{poly}(n)$-sized constant-depth circuits. This follows from the fact that we are working over constant characteristic. Refer to the diagram below for a schematic representation of what we discussed so far.

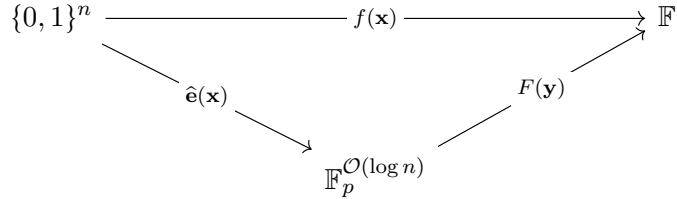
\[\begin{tikzcd}
	{\{0,1\}^{n}} &&&&&& {\mathbb{F}} \\
	\\
	&&& {\mathbb{F}_{p}^{\mathcal{O}(\log n)}}
	\arrow["{f(\mathbf{x})}"{description}, from=1-1, to=1-7]
	\arrow["{\widehat{\mathbf{e}}(\mathbf{x})}"{description}, from=1-1, to=3-4]
	\arrow["{F(\mathbf{y})}"{description}, from=3-4, to=1-7]
\end{tikzcd}\]

\paragraph{}Next we ``lift'' the Nullstellensatz back to the $n$ variables $(x_{1},\ldots,x_{n})$. To do so, we plug-in $\widehat{\mathbf{e}}(\mathbf{x})$ in place of $\mathbf{y}$. Observe that this substitution by $\widehat{\mathbf{e}}(\mathbf{x})$ preserves the size and the depth of the coefficients of the low-variate Nullstellensatz certificate because of the Ben-Or's construction (see \Cref{thm:ben-or}).\newline
It remains to \emph{prove} via constant-depth circuits that $F(\widehat{\mathbf{e}}(\mathbf{x}))$ agrees with $f(\mathbf{x})$ on the Boolean cube, i.e. $F(\widehat{\mathbf{e}}(\mathbf{x})) - f(\mathbf{x})$ lie in the ideal $(\mathbf{x}^{2} - \mathbf{x})$. Here ``to prove in constant-depth circuits'' refers to giving a certificate for the ideal membership whose coefficients can be computed by constant-depth circuits. More precisely, we want to prove that there exists polynomials $B_{j}(\mathbf{x})$'s which have $\mathrm{poly}(n)$-sized constant-depth circuits such that
\begin{align*}
    F(\widehat{\mathbf{e}}(\mathbf{x})) \; = \; f(\mathbf{x}) + \sum_{j=1}^{n} B_{j}(\mathbf{x}) \cdot (x_{j}^{2} - x_{j}).
\end{align*}
This is the key step in our proof. To prove this, it suffices to prove the following special case, which we prove in \Cref{lemma:multilinearize-char-p}.


\begin{lemma}
Let $\ell = \bigO(\log n)$ and fix an arbitrary sequence $(\alpha_{1},\ldots,\alpha_{\ell})$ where each $\alpha_{i} \in [n]$. There exist polynomials $B_{j}(\mathbf{x})$'s such that
\begin{align*}
    \prod_{i = 1}^{\ell} e_{\alpha_{i}}(\mathbf{x}) \; = \; \ml\brac{\prod_{i = 1}^{\ell} e_{\alpha_{i}}(\mathbf{x})} + \sum_{j=1}^{n} B_{j}(\mathbf{x}) \cdot (x_{j}^{2} - x_{j}),
\end{align*}
and each polynomial $B_{j}(\mathbf{x})$ can be computed by a $\mathrm{poly}(n)$-sized constant-depth circuit. 
\end{lemma}


\subsection{Related Work}
\label{sec:related}
In an independent work, Elbaz, Govindasamy, Lu, and Tzameret \cite{EGLT-25} consider related questions. Using the recent lower bound of Forbes~\cite{Forbes-LST-CCC}, which proves the positive characteristic version of the constant-depth formula lower bound of~\cite{LST}, they obtain lower bounds for fragments of the IPS over finite fields of \emph{any} size.

\subsection{Preliminaries}
In this subsection, we present a few more definitions and standard facts on polynomials which will be used in our proofs later on.\\

\noindent
For a polynomial $f(x_{1},\ldots,x_{n})$, the individual degree of $f$ is an integer $D$ such that for all $i \in [n]$, the degree of $f$ when viewed as a univariate polynomial in the variable $x_{i}$ is at most $D$.\newline
We next mention some useful properties about multilinear polynomials.\\

\begin{fact}[Standard facts on multilinear polynomials]\label{fact:multilinear}
Let $f(\mathbf{x}), g(\mathbf{x}) \in \F[\mathbf{x}]$.
\begin{itemize}
    \item $f(\mathbf{x})$ and $\ml[f(\mathbf{x})]$ agree on the Boolean cube $\Boo^{n}$.
    \item $f(\mathbf{x})$ and $g(\mathbf{x})$ agree on the Boolean cube $\Boo^{n}$ if and only if $\ml[f(\mathbf{x})]$ is equal to the $\ml[g(\mathbf{x})]$.
    \item $\ml[f(\mathbf{x}) g(\mathbf{x})] = \ml[\ml[f(\mathbf{x})] \ml[g(\mathbf{x})]]$.
\end{itemize}
\end{fact}
\noindent

\begin{theorem}[Fundamental Theorem of Symmetric Polynomials]\label{thm:fundamental-sym}
Fix any arbitrary field $\F$. If $f \in \F[x_{1},\ldots,x_{n}]$ is a symmetric polynomial of degree $d$, then there exists a unique polynomial $F \in \F[y_{1},\ldots,y_{d}]$ such that $f(\mathbf{x}) = F(e_{1}(\mathbf{x}),\ldots,e_{d}(\mathbf{x}))$.
\end{theorem}

\noindent
A classical and beautiful construction of Ben-Or shows that every elementary symmetric polynomial can be computed by $\mathrm{poly}(n)$-sized constant-depth circuits.\\

\begin{theorem}[Ben-Or's construction for elementary symmetric polynomials]\label{thm:ben-or}
(See \cite[Theorem 5.1]{Shpilka-Wigderson}). Let $\F$ be a field with $|\F| > n$. Then for every $d \in [n]$, the $d^{th}$ elementary symmetric polynomial $e_{d}(x_{1},\ldots,x_{n})$ has a circuit of size $\bigO(n^{2})$ and depth $3$ (a $\Sigma \Pi \Sigma$ circuit).\newline
More particularly, for any choice of $(n+1)$ distinct elements $\gamma_{1},\ldots,\gamma_{n+1} \in \F$ and for every $k \in [n]$, there exists coefficients $c_{k,i}$'s such that
\begin{align*}
    e_{k}(\mathbf{x}) \; = \; \sum_{i=1}^{n+1} c_{k,i} \prod_{j=1}^{n} (1+\gamma_{i} x_{j})
\end{align*}
\end{theorem}
The following recursive definition of elementary symmetric polynomials will be used in the proofs.
\begin{equation}\label{eqn:recursive-elem-sym}
    e_{d}(x_{1},\ldots,x_{n}) \; = \; x_{1} \cdot e_{d-1}(x_{2},\ldots,x_{n}) + e_{d}(x_{2},\ldots,x_{n}), \quad \quad \text{for all } d \in [n]
\end{equation}

\begin{theorem}[Polynomial Identity Lemma]\label{lem:ODLSZ}
(See \cite[Lemma 9.2.2]{GuruswamiRudraSudanCodingTheory}).
    Let $\F$ be an arbitrary field. Let $f(\mathbf{x})$ be a nonzero polynomial of degree at most $d$ and let $S \subseteq \F$. If we choose $\mathbf{a} \sim S^n$ uniformly at random, then: $$\Pr_{\mathbf{a}\sim S^n} [f(\mathbf{a}) = 0] \leq \frac{d}{|S|}$$ 
\end{theorem}

\paragraph{}For a natural number $k$ and variables $(z_{1},\ldots,z_{n})$, we will use $(\mathbf{z}^{k} - \mathbf{z})$ to denote the following ideal $(\mathbf{z}^{k} - \mathbf{z}) \; := \; (z_{1}^{k} - z_{1}, \ldots, z_{n}^{k} - z_{k}) \; \subseteq \; \F[z_{1},\ldots,z_{n}]$. We recall the following lemma which holds for fields with positive characteristic.\\

\begin{lemma}[Freshman's Dream]\label{lemma:freshman}
Fix a prime number $p$ and a field $\F$ of $\mathrm{char}(\F) = p$. Then for any $a, b \in \F$, we have, $(a+b)^{p} \; = \; a^{p} + b^{p}$. 
More generally, for any $a_{1},\ldots,a_{m} \in \F$, we get,
$(a_{1}+\ldots+a_{m})^{p} \; = \; a_{1}^{p} + \ldots + a_{m}^{p}$. 
\end{lemma}


\paragraph{}Next we recall the definition of an ideal and a variety, and then we state Hilbert's Nullstellensatz.

\begin{definition}[Ideal and Variety]
Fix any field $\F$ and consider the commutative ring $\F[x_{1},\ldots,x_{n}]$. For a set of polynomials $f_{1},\ldots,f_{m} \in \F[\mathbf{x}]$, the ideal generated by $f_{i}$'s, denoted by $(f_{1},\ldots,f_{m})$ is defined as:
\begin{align*}
    (f_{1},\ldots,f_{m}) \; = \; \setcond{h \in \F[\mathbf{x}]}{\exists g_{1},\ldots,g_{m} \in \F \text{ such that } \; h = \sum_{i=1}^{m} g_{i}f_{i} }.
\end{align*}
For a set of polynomials $f_{1},\ldots,f_{m} \in \F$, their variety, denoted by $\mathbb{V}(f_{1},\ldots,f_{m})$ is a subset of the algebraic closure of $\overline{\F}^{n}$, defined as:
\begin{align*}
    \mathbb{V}(f_{1},\ldots,f_{m}) \; = \; \setcond{\mathbf{a} \in \Bar{\F}^{n}}{f_{1}(\mathbf{a}) = \cdots = f_{m}(\mathbf{a}) = 0}.
\end{align*}
\end{definition}

\noindent
Now we state Hilbert's Nullstellensatz which essentially says that if a set of polynomials do not have a common zero, then there exists ``witness'' for this, i.e. one can express $1$ as a polynomial combination of $f_{i}$'s.\\

\begin{theorem}[Hilbert's Nullstellensatz]\label{thm:nullstellensatz}
Fix any field $\F$. Let $f_{1}, \ldots, f_{m} \in \F[x_{1},\ldots,x_{n}]$ be a set of multivariate polynomials such that they do not have any common zeros over the algebraic closure of $\F$. Then the constant $1$ lies in the ideal $(f_{1}(\mathbf{x}), \ldots, f_{m}(\mathbf{x}))$. In other words, there exists polynomials $A_{1},\ldots,A_{m} \in \F[x_{1},\ldots,x_{n}]$ such that
\begin{align*}
    A_{1}(\mathbf{x}) \cdot f_{1}(\mathbf{x}) + \cdots + A_{m}(\mathbf{x}) \cdot f_{m}(\mathbf{x}) \; = \; 1.
\end{align*}
\end{theorem}

\noindent
Strictly speaking, Hilbert's Nullstellensatz guarantees that the polynomials $A_{i}'s$ are in $\overline{\F}[\mathbf{x}]$ ($\overline{\F}$ is the algebraic closure of $\F$). However, the above statement also follows easily by observing that we can solve for $A_{i}$'s by solving a system of linear equations over $\F$. Throughout this article, we will refer to $(A_{1}(\mathbf{x}), \ldots, A_{m}(\mathbf{x}))$ as a \underline{\emph{Nullstellensatz certificate}}\footnote{There are infinitely many Nullstellensatz certificates for a system $\set{f_{1},\ldots,f_{m}}$. To see this, suppose $m = 2$ and let $(A_{1},A_{2})$ be a Nullstellensatz certificate. Then for any polynomial $g \in \F[\mathbf{x}]$, $(A_{1} + gf_{2}, A_{2} - gf_{1})$ is also a Nullstellensatz certificate.} for the system $\set{f_{1}(\mathbf{x}), \ldots, f_{m}(\mathbf{x})}$. We will also refer to $A_{i}$'s as \emph{coefficients} because if we take a polynomial combination of $f_{i}$'s with $A_{i}$'s being the coefficients, then we can generate $1$.

\begin{lemma}[Nullstellensatz certificate implies refutations]\label{lemma:nullstellensatz-refutations}
Fix any field $\F$. Let $P_{1},\ldots,P_{m} \in \F[x_{1},\ldots,x_{n}]$ be polynomials that have no common Boolean solution. Let the polynomials $A_{i}(\mathbf{x})$'s and $B_{j}(\mathbf{x})$'s be coefficients of the Nullstellensatz certificate, i.e.
\begin{align*}
    \sum_{i=1}^{m} A_{i}(\mathbf{x}) \cdot P_{i}(\mathbf{x}) + \sum_{j=1}^{n} B_{j}(\mathbf{x}) \cdot (x_{j}^{2} - x_{j}) \; = \; 1.
\end{align*}
Suppose for every $i \in [m]$ and for every $j \in [n]$, the polynomials $A_{i}(\mathbf{x})$ and $B_{j}(\mathbf{x})$ have a circuit of size $s$ and depth $\Delta$, then there exists a $\IPS$ proof for the system $\set{P_{1},\ldots,P_{r}}$ of size $\bigO(sm)$ and depth $\Delta+2$.
\end{lemma}
\begin{proof}[Proof of \Cref{lemma:nullstellensatz-refutations}]
Define the circuit $C(\mathbf{x}, \mathbf{y}, \mathbf{z})$ as follows:
\begin{align*}
    C(\mathbf{x}, \mathbf{y}, \mathbf{z}) \; = \; \sum_{i=1}^{m} A_{i}(\mathbf{x}) \cdot y_{i} \, + \, \sum_{j=1}^{n} B_{j}(\mathbf{x}) \cdot z_{j}
\end{align*}
Clearly $C(\mathbf{x}, \mathbf{0}, \mathbf{0}) = 0$ and $C(\mathbf{x}, f_{1},\ldots,f_{m}, x_{1}^{2}-x_{1},\ldots,x_{n}^{2}-x_{n}) = 1$. It is easy to verify the size and depth parameters of $C(\mathbf{x}, \mathbf{y}, \mathbf{z})$.
\end{proof} 
\Cref{lemma:nullstellensatz-refutations} allows us to restrict our attention to finding an efficient (in terms of algebraic complexity) Nullstellensatz certificate, which yields a short $\IPS$-proof.

\newcommand{\Coeff}{\ensuremath{\operatorname{Coeff}}}
\newcommand{\bfCoeff}{\ensuremath{\operatorname{\mathbf{Coeff}}}}
\newcommand{\Eval}{\ensuremath{\operatorname{Eval}}}
\newcommand{\bfEval}{\ensuremath{\operatorname{\mathbf{Eval}}}}
\newcommand{\bfalpha}{\ensuremath{\boldsymbol{\alpha}}}
\newcommand{\bfw}{\ensuremath{\mathbf{w}}}
\newcommand{\Aw}{\ensuremath{A_\bfw}}
\newcommand{\Bw}{\ensuremath{B_\bfw}}
\renewcommand{\char}{\ensuremath{\operatorname{char}}}
\newcommand{\bbF}{\ensuremath{\F}}
\newcommand{\ksw}{\ensuremath{\operatorname{ks}_{\bfw}}}
\newcommand{\kswnew}[1]{\ensuremath{\operatorname{ks}_{\bfw,#1}}}
\section{Lower Bounds in Large Fields of Positive Characteristic}

In this section, we will prove size lower bounds for several fragments of IPS over positive characteristic. As explained in \Cref{sec:intro-lower}, we start by proving a tight degree lower bound (\Cref{lem:basic degree lower bound}) over positive characteristic. Using our positive characteristic variant of the degree lower bound, we then recover the lower bound results from~\cite{FSTW21} and~\cite{GHT} over positive characteristic. 

\subsection{Degree Lower Bound for Arbitrary Characteristic}

For any $\mathbf{a} \in \Boo^n$, we use $|\mathbf{a}|$ to denote its Hamming weight.
For any $\mathbf{a} = (a_1, \dots, a_n) \in \Boo^n$ and any subset of indices $S\subseteq [n]$, we use $\mathbf{a}_S$ to denote $\prod_{i\in S}a_i$.
All the statements in this section work over fields of arbitrary characteristic. 

First, we state a standard fact about multilinear polynomials, which will be useful in the main lemma.
\begin{fact}\label{lem:top coeff of multilinear poly}
    Let $f(\mathbf{x}) = \sum_{S\subseteq [n]} \lambda_S \mathbf{x}_S$ be a multilinear polynomial on $n$ variables. Then, $$\lambda_{[n]} = \sum_{\mathbf{a}\in\Boo^n} (-1)^{|\mathbf{a}|}f(\mathbf{a})$$
\end{fact}

The next lemma is our main degree lower bound which shows that a multilinear polynomial for the inverse of a random linear form will have maximal degree. While similar statements have been observed in the literature (e.g. \cite[Proposition 2]{Grigoriev98}), we give an explicit proof for the sake of completeness.

\begin{lemma} \label{lem:basic degree lower bound}
    Let $\F$ and $\F'$ be fields such that $\F$ is a strict subfield of $\F'$. Let $n\in \mathbb{N}$ be a natural number and let $\mathbf{x}$ denote the tuple of variables $(x_1, \dots, x_n)$. Fix any $\beta \in \F'\setminus \F$. For any $\bfalpha = (\alpha_1, \dots, \alpha_n) \in \F^n$, let $f_{\bfalpha}(\mathbf{x})$ be the unique multilinear polynomial that agrees with the function $$\frac{1}{\sum_{i=1}^n{\alpha_ix_i}-\beta}$$ on the Boolean cube $\Boo^n$.
    Let $S \subseteq \F$ be any finite subset of the field.
    Then, for a uniformly random $\bfalpha \sim S^n$: $$\Pr_{\bfalpha \sim S^n}[\deg{f_{\bfalpha}(\mathbf{x})} = n] \geq 1 - \frac{2^n - 1}{|S|}$$
\end{lemma}
\begin{proof}
    By \Cref{lem:top coeff of multilinear poly}, the coefficient of $\mathbf{x}_{[n]}$ in $f_{\bfalpha}(\mathbf{x})$ is $\sum_{\mathbf{a}\in\Boo^n} (-1)^{|\mathbf{a}|}f_{\bfalpha}(\mathbf{a})$, or equivalently, $$\sum_{V \subseteq [n]} (-1)^{|V|}{\frac{1}{(\sum_{i \in V}\alpha_i) - \beta}}$$
    Based on the above expression, we define the rational function $\lambda_{[n]}(\mathbf{z})$ as follows. 
    $$\lambda_{[n]}(\mathbf{z}) := \sum_{V \subseteq [n]} (-1)^{|V|}{\frac{1}{(\sum_{i \in V}z_i) - \beta}}$$
    We will use $N(\mathbf{z})$ and $D(\mathbf{z})$ to denote the numerator and denominator of $\lambda_{[n]}(\mathbf{z})$. For any $S\subseteq [n]$, we will use $L_S(\mathbf{z})$ to denote $\sum_{i \in S}z_i$. It follows that 
    \begin{align*}
        N(\mathbf{z}) &= \sum_{V\subseteq [n]}{(-1)^{|V|}\prod_{T\subseteq [n]: T \neq V}({L_T(\mathbf{z})-\beta}}) \\
        D(\mathbf{z}) &= \prod_{V \subseteq [n]}{(L_V(\mathbf{z})-\beta)}
    \end{align*}
    Since $\beta \in \F'\setminus \F$, $D(\bfalpha) \neq 0$ for any $\bfalpha \in \F$. If we prove that $N(\mathbf{z})$ is a non-zero polynomial, then by the Polynomial Identity Lemma (\Cref{lem:ODLSZ}), for any finite subset $S\subseteq \F$, $\Pr_{\bfalpha \sim S^n}[N(\bfalpha) \neq 0] \geq 1 - \frac{2^n - 1}{|S|}$, which implies that $\Pr_{\bfalpha \sim S^n}[\lambda_{[n]}(\bfalpha) \neq 0] \geq 1 - \frac{2^n - 1}{|S|}$, and thus proves the theorem. Thus, it is enough to prove that some monomial in $N(\mathbf{z})$ has non-zero coefficient.
    
    For $V \neq \emptyset$, $\prod_{T\subseteq [n]: T \neq V}({L_T(\mathbf{z})-\beta})$ has degree at most $2^n - 2$ since $L_{\emptyset}(\mathbf{z})-\beta$ will not increase the degree. The term $\prod_{T \neq \emptyset}({L_T(\mathbf{z})-\beta})$ syntactically contributes monomials of degree $2^n - 1$ from $\prod_{T \neq \emptyset}L_T(\mathbf{z})$ , but is possible that these coefficients vanish if the field $\F$ is of positive characteristic. We will show that there is a monomial of degree $2^n - 1$ with coefficient 1, and thus this monomial will survive over any field.
    \begin{claim}
        The coefficient of the monomial\footnote{The same proof works for any monomial $\prod_{i=1}^n {z_{\sigma(i)}^{2^{i-1}}}$, where $\sigma$ is an arbitrary permutation on $[n]$.} $\prod_{i=1}^n {z_i^{2^{i-1}}}$ in $\prod_{T \neq \emptyset}({L_T(\mathbf{z})-\beta})$ is 1. 
    \end{claim}
    \begin{proof}[Proof sketch.] We would like to count the number of ways of collecting variables from each $L_T(\mathbf{z})$ to construct the required monomial.
    We first observe (via a simple counting argument) that for every $i\in [n]$, the number of subsets $T\subseteq [n]$ such that 
        $\{j\in[n]: j>i\} \cap T = \emptyset$, and
         $i \in T$,
    is $2^{i-1}$. Moreover, for each $i\in[n]$, if $\mathcal{T}_i$ is the collection of subsets with the above properties, then we observe that $\mathcal{T}_i \cap \mathcal{T}_j = \emptyset$ for all $i \neq j$, $i\in[n]$, $j\in [n]$.
        
    With these observations, it inductively follows that for each $i \in [n]$, conditioned on the degree of variables $z_n, \dots, z_{i+1}$ being correct (i.e. $z_j^{2^{j-1}}$), there is exactly one way of ensuring that the degree of $z_{i}$ is $2^{i-1}$: for each $T$ that is one of the $2^{i-1}$ subsets satisfying the properties of the above observation, select the $z_i$'s from $L_T(\mathbf{z})$.
    \end{proof}
\end{proof}

Note that \Cref{lem:basic degree lower bound} is interesting only when the field size is large (at least $2^n$), and that will be the case for subsequent lemmas as well. The next lemma proves a stronger version of the previous lemma: for a random linear form, the inverse of \emph{every} restriction of the linear form (by setting some variables to 0) will have maximal degree.
\begin{lemma}\label{lem: deg lower bound union bound}
    Let $\F$ and $\F'$ be fields such that $\F$ is a strict subfield of $\F'$. Let $n\in \mathbb{N}$ be a natural number and let $\mathbf{x}$ denote the tuple of variables $(x_1, \dots, x_n)$. Fix any $\beta \in \F'\setminus \F$. For any $\emptyset \neq U \subseteq [n]$, let $f_{\bfalpha,U}(\mathbf{x})$ be the unique multilinear polynomial that agrees with the function $$\frac{1}{\sum_{i\in U}{\alpha_ix_i}-\beta}$$ on the Boolean cube $\Boo^n$.
    Let $S \subseteq \F$ be a finite subset of the field.
    Then, for an $\bfalpha \sim S^n$ chosen uniformly at random: $$\Pr_{\bfalpha \sim S^n}[\exists \text{ a non-empty } U\subseteq [n]:\deg{f_{\bfalpha,U}(\mathbf{x})} < |U|] \leq \sum_{\emptyset \neq U\subseteq [n]}\frac{2^{|U|}-1}{|S|} < \frac{2^{2n}}{|S|}$$
    In particular, with probability at least $1-(2^{2n}/|S|)$ over the choice of $\bfalpha \sim S^n$, for every $U \subseteq [n]$, the leading monomial of $f_{\bfalpha,U}(\mathbf{x})$ is $c\cdot\prod_{i\in U}x_i$ for some $c\in\F\setminus\{0\}$. 
\end{lemma}
\begin{proof}
    This lemma is a simple application of union bound with the previous lemma. The previous lemma tells us that for a uniformly random $\bfalpha \sim S^n$ and any $U\subseteq [n]$, $$\Pr_{\bfalpha \sim S^n}[\deg{f_{\bfalpha,U}(\mathbf{x})} < |U|] \leq \frac{2^{|U|} - 1}{|S|}$$
    Union bound over all $U\subseteq [n]$ gives us the required statement.
\end{proof}

\subsection{Sparse-$\IPSLINp$ Lower Bound}
The following claim from \cite{FSTW21} proves a lower bound against sparse-$\IPSLINp$ over fields of large characteristic. 
\begin{proposition}[Sparsity lower bound (Proposition 5.6 \cite{FSTW21})]
    Let $n \geq 8$. Let $\F$ be a field of characteristic $> n$. Let $\beta \in \F\setminus\{0,\dots, n\}$. Suppose $f(\mathbf{x})$ be a polynomial such that $$f(\mathbf{x}) \cdot \left (\sum_{i=1}^{n}{x_i} - \beta \right) \equiv 1 \left({\mathbf{x}^2 - \mathbf{x}} \right )$$ where $\left({\mathbf{x}^2 - \mathbf{x}} \right)$ denotes the ideal $(x_1^2-x_1, \dots, x_n^2 - x_n)$.
    Then, the sparsity of $f(\mathbf{x})$ is at least $2^{\frac{n}{4} -  1}$.
\end{proposition}
The proof uses two observations.

\begin{enumerate}
    \item (\cite[Lemma 5.5]{FSTW21}) If $f(\mathbf{x})$ has sparsity $s$, then a random restriction $\rho$ will ensure that $\deg(\rho(f)) \leq \log(s) + 1$ with reasonable probability.
    \item (Chernoff bound) A random restriction $\rho$ will keep at least $n/4$ variables alive with reasonable probability.
\end{enumerate}

By a union bound, we can find a random restriction $\rho$ that ensures that the degree of $\rho(f)$ is at most $\log(s) + 1$ but at least $n/4$ variables survive $\rho$. In particular, $\rho(\sum_{i\in [n]}x_i - \beta) = \sum_{i \in S} x_i - \beta$ for some $S\subseteq [n]$ with $|S|\geq n/4$. But the degree lower bound in \cite{FSTW21} tells us that the inverse of $\sum_{i \in S} x_i - \beta$ on the Boolean cube must have degree $\geq |S|$. Combining the above observations with the degree lower bound, we get that $n/4 \leq \log(s) + 1$ or $s \geq 2^{n/4 - 1}$.

The only part of the proof that requires $\char \F > n$ is the degree lower bound; the two observations work over all fields. Thus, we can replace their degree lower bound with \Cref{lem: deg lower bound union bound} to recover the sparsity lower bound over large enough fields of arbitrary characteristic.

\begin{theorem}
    Let $n \geq 8$.  Let $p\in \mathbb{N}$ be any prime. Let $\tilde{\F}$ be a field of characteristic $p$ and size $p^{2k}$, where $k$ is the smallest integer that satisfies $p^k > 2^{2n}$. Let $\beta$ be an arbitrary element in $\tilde{\F}\setminus\F$, where $\F$ denotes the subfield of size $p^k$. For any $\bfalpha \in \F^{n}$, let $f_{\bfalpha}(\mathbf{x})$ be a polynomial which agrees with $$\frac{1}{\sum_{i\in [n]}\alpha_{i}x_i - \beta}$$ on the Boolean cube. Then there exists an $\bfalpha \in \F^{n}$ such that $f_{\bfalpha}$ has sparsity $\geq 2^{\Omega(n)}$
\end{theorem}

\subsection{$\roABP-\IPSLINp$ Lower Bound}

 \Cref{lem: deg lower bound union bound} tells us that for a random choice of coefficients $\bfalpha$ and any $U\subseteq [n]$, the inverse of $\sum_{i\in U}^n \alpha_ix_i - \beta$ has degree $|U|$ over the Boolean cube. The authors of \cite{FSTW21} ``lift'' such maximal degree lower bounds to construct a polynomial $P(\mathbf{x})$ such that any $\roABP$ that computes (in \emph{any} order of variables) the inverse of $P(\mathbf{x})$ over the Boolean cube requires exponential size. A high-level overview of their proof is as follows. 
 \begin{enumerate}
     \item The optimal width of an $\roABP$ computing a polynomial $g$ is captured exactly by the \emph{coefficient dimension}\footnote{These notions are defined with respect to a certain partition of the variables and any order of variables that is consistent with the specified partition.} of $g$. 
     \item The coefficient dimension of a polynomial $g$ is at least as large as the \emph{evaluation dimension} of $g$. 
     \item For $f(\mathbf{x},\mathbf{y}) := \sum_{i\in [n]}x_iy_i - \beta$, evaluations of $f$ on $y\in \Boo^n$ will be $f_S(\mathbf{x}) = \sum_{i\in S} x_i - \beta$ for various $S\subseteq [n]$. 
     \item By the degree lower bound in \cite{FSTW21}, any multilinear polynomial computing the inverse of $f_S$ over the Boolean cube must have degree $|S|$. This eventually implies that the evaluation space of $g(\mathbf{x},\mathbf{y}):= \frac{1}{f(\mathbf{x},\mathbf{y})}$ over $y \in \Boo^n$ will contain all the multilinear monomials on $\mathbf{x}$ variables. In particular, the evaluation dimension\footnote{Again, the order of variables will be important here, but one can also construct a polynomial which works against $\roABP$s in \emph{any} order of variables.} of $g$ is at least $2^n$, and thus, any $\roABP$ computing $g$ must have width $\geq 2^n$.
 \end{enumerate}
 The only part of their proof that requires a restriction on the characteristic of the underlying field is the degree lower bound. The rest of their proof works with the degree lower bound in \Cref{lem: deg lower bound union bound}. In the rest of this section, we state the final theorems that follow using our degree lower bound in the proofs of \cite{FSTW21}. For more details, we recommend the reader to refer to the appendix as well as \cite{FSTW21}. 

\begin{theorem}[Functional lower bound against $\roABP$ in a fixed order of variables]\label{thm:roabp-lbd-fixed-order-positive-char}
    Let $n\in \mathbb{N}$. Let $p\in \mathbb{N}$ be any prime. Let $\tilde{\F}$ be a field of characteristic $p$ and size $p^{2k}$, where $k$ is the smallest integer that satisfies $p^k > 2^{2n}$. Let $\beta$ be an arbitrary element in $\tilde{\F}\setminus\F$, where $\F$ denotes the subfield of size $p^k$. For any $\bfalpha \in \F^{n}$, let $f_{\bfalpha}(\mathbf{x},\mathbf{y})$ be a polynomial which agrees with $$\frac{1}{\sum_{i \in [n]}\alpha_{i}x_iy_i - \beta}$$ on the Boolean cube. Then there exists an $\bfalpha \in \F^{n}$ such that any $\roABP$ that computes $f_{\bfalpha}$ in any order of variables where $\mathbf{x}$ precedes $\mathbf{y}$ requires width $\geq 2^n$. 
\end{theorem}
\begin{restatable}[Functional lower bound against $\roABP$ in any order of variables]{theorem}{roabplbdanyorder}\label{thm:roabp-lbd-any-order-positive-char}
    Let $n\in \mathbb{N}$. Let $p\in \mathbb{N}$ be any prime. Let $\tilde{\F}$ be a field of characteristic $p$ and size $p^{2k}$, where $k$ is the smallest integer that satisfies $p^k > \binom{2n}{n}2^{2n}$. Let $\beta$ be an arbitrary element in $\tilde{\F}\setminus\F$, where $\F$ denotes the subfield of size $p^k$. For any $\bfalpha \in \F^{\binom{2n}{2}}$, let $f_{\bfalpha}(\mathbf{x}=(x_i)_{i=1}^{2n},\mathbf{z}=(z_{i,j})_{i,j\in[n]})$ be a polynomial which agrees with $$\frac{1}{\sum_{i<j}\alpha_{i,j}z_{i,j}x_ix_j - \beta}$$ on the Boolean cube. Then there exists an $\bfalpha \in \F^{\binom{2n}{2}}$ such that any $\roABP$ that computes $f_{\bfalpha}$ in any order of variables requires size $\geq 2^n$.
\end{restatable}


\subsection{Multilinear-formula-$\IPS$ Lower Bound}

Lower bounds against multilinear-formula-$\IPS$ follow from a coefficient dimension lower bound (see \Cref{lem: coeff dim any partition}) and the following theorem of Raz and Yehudayoff that connects multilinear formula size to coefficient dimension. Here, we present the version from \cite[Theorem 3.13]{FSTW21}.

\begin{theorem}[Raz-Yehudayoff \cite{RY09}\cite{Raz-2009}]
    Let $f \in \F[x_1, \dots, x_{2n}, \mathbf{z}]$ be a multilinear polynomial and let $f_{\mathbf{z}}$ denote the polynomial $f$ over the ring $\F[\mathbf{z}]$. Suppose for any balanced partition $(\mathbf{u}, \mathbf{v})$ of $\mathbf{x} = (x_1, \dots, x_{2n})$: $$\dim_{\F(\mathbf{z})} \bfCoeff_{\mathbf{u} | \mathbf{v}}(f_\mathbf{z}) \geq 2^n$$
    Then any multilinear formula for $f$ requires size $\geq n^{\Omega(\log n)}$, and for $\Delta = o(\log n / \log \log n)$, any product-depth-$\Delta$ multilinear formula computing $f$ will require size $\geq n^{\Omega\left(\frac{1}{\Delta^2}(\frac{n}{\log n})^{1/\Delta}\right )}$. 
\end{theorem}
\begin{theorem}[Functional lower bounds against multilinear formula]\label{thm:multilinear-formula-lbds-via-raz-yehudayoff}
    Let $n\in \mathbb{N}$. Let $p\in \mathbb{N}$ be any prime. Let $\F_{p^{2k}}$ be a field of characteristic $p$ and size $p^{2k}$, where $k$ is the smallest integer that satisfies $p^k > \binom{2n}{n}2^{2n}$. Let $\beta$ be an arbitrary element in $\F_{p^{2k}}\setminus\F$, where $\F$ denotes the subfield of size $p^k$. For any $\bfalpha \in \F^{\binom{2n}{n}}$, let $f_{\bfalpha}(\mathbf{x}=(x_i)_{i=1}^{2n},\mathbf{z}=(z_{i,j})_{i,j\in[n]})$ be a polynomial which agrees with $$\frac{1}{\sum_{i<j}\alpha_{i,j}z_{i,j}x_ix_j - \beta}$$ on the Boolean cube. There exists an $\bfalpha \in \F^{\binom{2n}{n}}$ such that any multilinear-formula computing $f_{\bfalpha}$ requires size $\geq n^{\Omega(\log n)}$ and for $\Delta = o(\log n / \log \log n)$, any product-depth-$\Delta$ multilinear-formula computing $f_{\bfalpha}$ requires size $\geq n^{\Omega\left(\frac{1}{\Delta^2}(\frac{n}{\log n})^{1/\Delta}\right )}$. 
\end{theorem}

While this immediately implies multilinear-formula-$\IPSLINp$ lower bounds, one can observe (as noted in Lemma 5.2 of \cite{FSTW21}) that any multilinear-formula-$\IPS$ refutation, by multilinearity, is a multilinear-formula-$\IPSLINp$ refutation. Thus, the lower bounds work against multilinear-formula-$\IPS$.

\subsection{Constant-depth Multilinear $\IPSLINp$ Lower Bound}

In \cite{GHT}, Govindasamy, Hakoniemi, and Tzameret prove super polynomial lower bounds against constant-depth multilinear $\IPSLINp$ refutations of the subset sum variant $$\sum_{i,j,k,l\in[n]}z_{i,j,k,l}x_ix_jx_kx_l - \beta$$
In particular, they prove the following theorem.

\begin{theorem}[Constant-depth functional lower bounds \cite{GHT}]\label{thm: ght}
    Let $n, \Delta \in \mathbb{N}_{+}$ with $\Delta \leq \mathcal{O}(\log\log\log n)$ and assume that $\char(\bbF)=0$. Let $f$ be the multilinear polynomial such that $$ f = \frac{1}{\sum_{i,j,k,l\in[n]}z_{i,j,k,l}x_ix_jx_kx_l - \beta}$$
    over the Boolean cube. Then, any circuit of product-depth $\Delta$ computing $f$ has size at least $$n^{(\log n)^{\exp(-\mathcal{O}(\Delta))}}$$
\end{theorem}

We prove the same statement for large fields of arbitrary characteristic. Our proof exactly follows the structure of \cite{GHT}. Their proof requires the $\char\bbF = 0$ condition for two reasons:
\begin{enumerate}
    \item They use the results of Limaye, Srinivasan, and Tavenas \cite{LST}, which gave superpolynomial lower bounds against constant-depth circuits over any field $\F$ with $\char(\F) = 0$ or greater than the degree $d$ of the hard polynomial. In particular, they use the result that over fields with $\char(\F)=0$ or greater than $d$, any low-degree set-multilinear polynomial computed by a constant-depth circuit can also be computed by a set-multilinear constant-depth circuit.\footnote{They also use other ideas from \cite{LST} such as relative rank, word polynomial, etc., but those ideas do not require any restrictions on the characteristic of the underlying field.}
    \item They use the degree lower bound for the multilinear representation of $1/(\sum_{i\in[n]}x_i - \beta)$, proved by Forbes, Shpilka, Tzameret, and Wigderson \cite{FSTW21}.
\end{enumerate}

To deal with the first requirement, we use the recent beautiful result of Forbes \cite{Forbes-LST-CCC}, which extends the results of \cite{LST} to arbitrary fields. In particular, we will use the following statement from \cite{Forbes-LST-CCC}, which says that the set-multilinear projection of a constant-depth circuit can be efficiently computed by a constant-depth circuit over arbitrary fields. 
\begin{theorem}
   \cite[Corollary 27]{Forbes-LST-CCC}. Let $\F$ be an arbitrary field. Let $\mathbf{x} = \mathbf{x}_1 \sqcup \mathbf{x}_2 \sqcup \dots \sqcup \mathbf{x}_d$ be a partition of the variables $\mathbf{x}$. Suppose $f$ can be computed by a size $s$ product-depth $\Delta$ arithmetic circuit. Then the set-multilinear projection of $f$ (the restriction of $f$ to monomials that are set-multilinear with respect to the specified partition) can be computed by a size $\poly(s,\Theta(\frac{d}{\log d})^d)$-size circuit of product-depth $2\Delta$. 
\end{theorem}

To deal with the second requirement, we use our degree lower bound from \Cref{lem: deg lower bound union bound}, which works for arbitrary fields of exponential size i.e. there is no restriction on the characteristic of the field. 

\paragraph{Overview of \cite{GHT}}

\begin{enumerate}
    \item Using the \emph{word polynomials} framework of \cite{LST}, construct a \emph{knapsack polynomial} $\ksw$ (for a partition given by a word $w \in \mathbb{Z}^d$) with the property that the set-multilinear projection of $\frac{1}{\ksw}$ over the Boolean cube requires superpolynomially large set-multilinear constant-depth circuits.
    \item Consider a degree-4 subset-sum variant $f(\mathbf{z},\mathbf{x}) := \sum_{i,j,k,l}z_{i,j,k,l}x_ix_jx_kx_l - \beta$ so that for the word $w\in \mathbb{Z}^d$ that will be used to instantiate the previous point, there exists an assignment of some of the variables in $\mathbf{z}$, $\mathbf{x}$ that maps $f(\mathbf{z},\mathbf{x})$ to $\ksw$ (upto a renaming of variables).
    \item If there is a multilinear polynomial computing $1/f(\mathbf{z},\mathbf{x})$ over $\Boo^n$ that has a small constant-depth circuit, then there is a multilinear polynomial computing $1/\ksw$ over $\Boo^n$ that has a small constant-depth circuit. Moreover by the set-multilinearization of \cite{LST}, there is a small set-multilinear constant-depth circuit computing the set-multilinear projection of $1/\ksw$. 
    \item Combining the first point with the contrapositive of the third point, conclude that any multilinear polynomial computing $1/f(\mathbf{z},\mathbf{x})$ over $\Boo^n$ requires superpolynomially large constant-depth circuits. The multilinear constant-depth $\IPSLINp$ lower bound follows.
\end{enumerate}

In \cite{GHT}, the proof for the hardness of $\frac{1}{\ksw}$ requires the underlying field to be of large characteristic, essentially because it requires the degree lower bound from \cite{FSTW21}, which requires large characteristic. To make \Cref{thm: ght} work over fields of positive characteristic, we will employ our degree lower bound from \Cref{lem: deg lower bound union bound} with a variant of the knapsack polynomial; the rest of the proof remains the same as that of \Cref{thm: ght}. To provide the necessary details, we first describe the construction of the knapsack polynomial. Then, we state the particular claim from \cite{GHT} that uses the degree lower bound from \cite{FSTW21}. Finally, we show how our degree lower bound \Cref{lem: deg lower bound union bound} fits into the rest of the proof. 

\paragraph{Constructing the knapsack polynomial}
We shall now recall the definitions required for defining the hard polynomial in \cite{GHT} via the word polynomials template of \cite{LST}. 

Let $\mathbf{w} \in \mathbb{Z}^d$ be an arbitrary word. For any $S\subseteq [d]$, let $w|_{S}$ denote the subword of $w$ indexed by the set $S$. Consider the sequence $\overline{X}(w) = (X(w_1),\dots, X(w_d))$ of sets of variables. Define the \emph{positive indices} and \emph{negative indices} of $\mathbf{w}$ as:
$$P_{\mathbf{w}} := \{i\in [d]: w_i \geq 0\}$$
$$N_{\mathbf{w}} := \{i\in [d]: w_i < 0\}$$
Let any $i\in P_\bfw$, the variables of $X(w_i)$ will be of the form $x_{\sigma}^{(i)}$, where $\sigma$ is a binary string indexed by the set:
$$A_\bfw^{(i)} := \left[ \sum_{\substack{i' \in P_\bfw\\ i'<i}} w_{i'} + 1,\sum_{\substack{i' \in P_\bfw\\ i'\leq i}}w_{i'}\right ]$$
We will call these sets \emph{positive indexing sets}. The size of each $\Aw^{(i)}$ is $|w_i|$. The number of strings in $\Aw^{(i)}$ is $2^{|w_i|}$. \\
For $i\in N_\bfw$, we similarly define the \emph{negative indexing sets} $\Bw^{(i)}$ that will be used to index the variables of $X(w_i)$ for $i\in N_\bfw$. \\
A word $w\in \mathbb{Z}^d$ is \emph{balanced} if:
\begin{itemize}
    \item $\forall i\in P_\bfw \; \exists j \in N_\bfw$ such that $\Aw^{(i)} \cap \Bw^{(j)} \neq \emptyset$ (i.e. $j\in N_\bfw$ is a \emph{witness} that $\bfw$ is balanced at $i\in P_\bfw$)
    \item $\forall j\in N_\bfw \; \exists i \in P_\bfw$ such that $\Aw^{(i)} \cap \Bw^{(j)} \neq \emptyset$ (i.e. $i\in P_\bfw$ is a \emph{witness} that $\bfw$ is balanced at $j\in N_\bfw$)
\end{itemize}

For any $i\in P_\bfw, \sigma \in \Boo^{\Aw^{(i)}}$, define: 
\begin{equation}
    f_\sigma^{(i)} := \prod_{\substack{j \in N_\bfw \\ \Aw^{(i)} \cap \Bw^{(j)} \neq \emptyset}} \sum_{\substack{\sigma_j \in \Boo^{\Bw^{(j)}}\\ \sigma_j(k) = \sigma(k) \forall k \in  \Aw^{(i)} \cap \Bw^{(j)} }} y_{\sigma_j}^{(j)}
\end{equation}
The product ranges over each $j \in N_{\bfw}$ that witnesses the fact that $\bfw$ is balanced at $i$. The sum ranges over each $\sigma_j$ that is consistent with $\sigma$ on $\Aw^{(i)} \cap \Bw^{(j)}$. Now, we define the knapsack polynomial as
\begin{equation}
    \ksw := \left ( \sum_{i\in P_\bfw}\sum_{\sigma \in \Boo^{\Aw^{(i)}}} x_{\sigma}^{(i)}f_\sigma^{(i)} \right ) - \beta
\end{equation}
where $\beta \in \F$ is any field element such that $\ksw$ has no Boolean roots.

To make the proof work over fields of positive characteristic, we define a variant of $\ksw$ as:
\begin{equation}
    \kswnew{\bfalpha}:= \left ( \sum_{i\in P_\bfw}\alpha_i\sum_{\sigma \in \Boo^{\Aw^{(i)}}} x_{\sigma}^{(i)}f_\sigma^{(i)} \right ) - \beta
\end{equation}    
where $\bfalpha = (\alpha_i)_{i\in P_\bfw} \in \F^{|P_\bfw|}$, and $\beta$ will be chosen from an extension field $\tilde{\F} \supset \F$ so that $\kswnew{\bfalpha}$ has no Boolean roots. 

For any word $\bfw \in \mathbb{Z}^d$, $M_\bfw(f)$ denotes the matrix with rows indexed by all monomials $m$ that are set-multilinear over $\bfw|_{P_\bfw}$, and columns indexed by all monomials $m'$ that are set-multilinear over $\bfw|_{N_\bfw}$. For each such pair of monomials $(m,m')$, the corresponding entry in $M_\bfw(f)$ carries the coefficient of $mm'$ in $f$. To show that the set-multilinear projection of any multilinear polynomial $f$ computing $1/\ksw$ over $\Boo^n$ requires superpolynomially large set-multilinear constant-depth circuits, \cite{GHT} shows that $M_\bfw(f)$ is full-rank.
\begin{lemma}[Rank lower bound lemma (Lemma 6 \cite{GHT})]\label{lem: ght rank lower bound}
    Let $\bfw \in \mathbb{Z}^d$ be a balanced word, and let $f$ be the multilinear polynomial such that $$f = \frac{1}{\ksw}$$
    over $\Boo^n$. Then, $M_\bfw(f)$ is full-rank.
\end{lemma}
With this lemma, the lower bound follows via the arguments from \cite{LST}. Importantly for us, this lemma uses the degree lower bound from \cite{FSTW21}; we describe a sketch of the same. 

\paragraph{The use of degree lower bound in \cite{GHT}} Suppose $f = \sum_{m}g_m(\mathbf{x})m$, where the sum runs over all multilinear monomials $m$ in the $\mathbf{y}$ variables, and $g_m(\mathbf{x})$ is some multilinear polynomial in the $\mathbf{x}$ variables. They show that for any $m$ which is set-multilinear on $\bfw|_{N_\bfw}$, the leading monomial of $g_m(\mathbf{x})$ is the set-multilinear monomial $m'$ on positive variables such that $\sigma(m')$ is consistent with $\sigma(m)$ (\cite{GHT} describes this formally). For each monomial $m$ that is set-multilinear on $\bfw|_{N_\bfw}$, the leading monomial of $g_m(\mathbf{x})$ turns out to be a different set-multilinear monomial on the positive variables, and together, these leading monomials span the space of all set-multilinear monomials on the positive variables. This makes $M_\bfw(f)$ full-rank. To get a handle on $g_m(\mathbf{x})$ (for $m$ being a monomial on $\bfw|_{N_\bfw}$, consisting only of $\mathbf{y}$ variables), \cite{GHT} sets all the variables in $m$ to 1 and all the $\mathbf{y}$ variables outside $m$ to 0. They call this transformation $\tau_m$. For the proof of \Cref{lem: ght rank lower bound}, an important requirement is that:
\begin{displayquote}For every $T\subseteq N_\bfw$ and for every set-multilinear monomial $m$ on $\bfw|_T$, the leading monomial of $\tau_m(f)$ is $\prod_{i\in U_T}x_{\sigma_i}^{(i)}$, which is the product of all the variables that show up in the denominator of $$\frac{1}{\tau_m(\ksw)} = \frac{1}{\sum_{i\in U_T}x_{\sigma(i)}^{(i)} - \beta}$$ where $U_T = \{i \in P_\bfw: \Aw^{(i)} \subseteq \Bw^T\}$, and for each $i \in P_\bfw$, $\sigma(i)$ is the unique indexing string that agrees with $\sigma(m)$ on $\Aw^{(i)}$, the $i^{th}$ positive indexing set.
\end{displayquote}
This requirement is satisfied due to the degree lower bound from \cite{FSTW21}, which requires the field to be of characteristic 0. The proof in \cite{GHT} includes helpful figures and the reader is encouraged to refer to the paper.\\
Let us recall our variant of $\ksw$:
\begin{equation}
    \kswnew{\bfalpha}:= \left ( \sum_{i\in P_\bfw}\alpha_i\sum_{\sigma \in \Boo^{\Aw^{(i)}}} x_{\sigma}^{(i)}f_{\sigma_i} \right ) - \beta
\end{equation}
where $\bfalpha = (\alpha_i)_{i\in P_\bfw} \in \F^{|P_\bfw|}$. To prove \Cref{thm: ght} in positive characteristic, we use the following lemma that follows by a union bound over all $T\subseteq N_\bfw$ and all set-multilinear monomials on $\bfw|_{T}$, on top of \Cref{lem:basic degree lower bound}.

\begin{lemma}\label{lem: ght degree lower bound positive char}
Let $d\in \mathbb{N}$ be a natural number and $\bfw \in \mathbb{Z}^d$ be a balanced word. Let $m = |P_\bfw|$. For any $\bfalpha = (\alpha_1, \dots, \alpha_m) \in \F^m$, $T \subseteq N_\bfw$ and any $m_T$ that is a set-multilinear monomial on $\bfw|_T$, let $ f_{\bfalpha,T,m_T}(\mathbf{x})$ be the unique multilinear polynomial that agrees with the function $$\tau_{m_T} \left (\frac{1}{\kswnew{\bfalpha}} \right ) = \frac{1}{\sum_{i\in U_T}{\alpha_ix_{\sigma(i)}^{(i)}}-\beta}$$ on the Boolean cube, where $\beta \in \F$ is chosen so that $\kswnew{\bfalpha}$ has no Boolean roots, and $U_T = \{i \in P_\bfw: \Aw^{(i)} \subseteq \Bw^T\}$. 
Let $S \subseteq \F$ be a finite subset of the field. Let $\gamma := |N_\bfw|+\sum_{i\in N_\bfw}|w_i|$.
Then, for an $\bfalpha \in S^m$ chosen uniformly at random: $$\Pr_{\bfalpha \sim S^m}[\exists T\subseteq N_\bfw, m_T:\deg{f_{\bfalpha,T,m_T}(\mathbf{x})} < |U_T|] < \frac{2^{\gamma+m}}{|S|}$$
In particular, with probability at least $1-(2^{\gamma+m}/|S|)$ over the choice of $\bfalpha \in S^m$, for every choice of $T\subseteq N_\bfw$ and set-multilinear monomial $m_T$ over $\bfw|_T$, the leading monomial of $f_{\bfalpha,T,m_T}(\mathbf{x})$ is $c\cdot\prod_{i\in U_T}x_{\sigma_i}^{(i)}$ for some $c\in\F\setminus\{0\}$.     
\end{lemma}
\begin{proof}
    The number of $T\subseteq N_\bfw$ is $2^{|N_\bfw|}$. The number of set-multilinear monomials on $\bfw|_T$ for any $T\subseteq N_\bfw$ is $2^{\sum_{i\in T}|w_i|}$, which is at most $2^{\sum_{i\in N_\bfw}|w_i|}$ . For any fixed $T \subseteq N_\bfw$ and $m_T$ that is a set-multilinear monomial on $\bfw|_T$, \Cref{lem:basic degree lower bound} implies that for an $\bfalpha \in S^m$ chosen uniformly at random: $$\Pr_{\bfalpha \in S^m}[\deg{f_{\bfalpha,T,m_T}(\mathbf{x})} < |U_T|] < \frac{2^{m}}{|S|}$$
    Applying a union bound over all $T \subseteq N_\bfw$ and $m_T$ implies that for an $\bfalpha \in S^m$ chosen uniformly at random: $$\Pr_{\bfalpha \in S^m}[\exists T\subseteq N_\bfw, m_T:\deg{f_{\bfalpha,T,m_T}(\mathbf{x})} < |U_T|] < \sum_{T\subseteq N_\bfw,m_T}\frac{2^{m}}{|S|} \leq \frac{2^{\gamma+m}}{|S|} $$    
\end{proof}

With this lemma, the rest of the proof of \cite{GHT} works out verbatim. We state the final theorem, which is a version of \Cref{thm: ght} for finite fields of positive characteristic.

\begin{theorem}[\cite{GHT} over positive characteristic]\label{thm: ght positive char}
Let $n, \Delta \in \mathbb{N}_{+}$ with $\Delta \leq \mathcal{O}(\log\log\log n)$. Let $p\in \mathbb{N}$ be any prime. Let $\tilde{\F}$ be a field of characteristic $p$ and size $p^{2k}$, where $k$ is the smallest integer that satisfies $p^k > 2^{C(\log n)^2}$ for an absolute constant\footnote{This $C$ is a fixed constant that depends on the exact choice of parameters in the proof of \cite{GHT}} $C\geq1$. Let $\beta$ be an arbitrary element in $\tilde{\F}\setminus\F$, where $\F$ denotes the subfield of size $p^k$. For any $\bfalpha \in \F^{n^4}$, Let $f_{\bfalpha}$ be the multilinear polynomial such that $$ f = \frac{1}{\sum_{i,j,k,l\in[n]}\alpha_{i,j,k,l}z_{i,j,k,l}x_ix_jx_kx_l - \beta}$$
over the Boolean cube. Then, there exists an $\bfalpha \in \F^{n^4}$ such that any circuit of product-depth $\Delta$ computing $f_{\bfalpha}$ has size at least $$n^{(\log n)^{\exp(-\mathcal{O}(\Delta))}}$$    
\end{theorem}
\paragraph{The reason for $|\F|>2^{\Omega((\log n)^2)}$ in \Cref{thm: ght positive char}}: When we instantiate \Cref{lem: ght degree lower bound positive char} inside the proof of \Cref{thm: ght positive char}, the parameter $d$, which is the number of variable sets, will be $O(\log n)$, and the word $\bfw \in \mathbb{Z}^d$ will also be chosen so that for each $i\in [d]$, $|w_i|\leq O(\log n)$. Thus, $\sum_{i\in N_\bfw}|w_i| = O((\log n)^2)$, and fighting the union bound in \Cref{lem: ght degree lower bound positive char} will require the field to be larger than $2^{O((\log n)^2)}$.

\section{Non-multilinear Upper Bounds}

\subsection{Proof of \Cref{thm:ubd-const-depth}}
In this section, we prove \Cref{thm:ubd-const-depth}. We start by proving it for a restricted setting when the polynomial $f(\mathbf{x})$ is a degree-$1$ polynomial. In particular, we prove \Cref{thm:ubd-const-depth-linear}, stated below.

\begin{theorem}[Upper bounds for (non-multilinear) constant-depth-$\IPSLIN$ in positive characteristic]\label{thm:ubd-const-depth-linear}
Fix a prime number $p$. The following holds for any natural numbers $n$ and $k.$\newline
Let $L \in \F[x_1,\ldots,x_n]$ be a degree-$1$ polynomial with coefficients from the $\F_{p^k}$ and let $\beta$ be any element of $\F\setminus\F_{p^k}$ where $\F$ is a field extension of $\F_{p^k}.$\newline
Then, 
\begin{itemize}
    \item The polynomial $L(\mathbf{x})-\beta$ has no satisfying assignment over the Boolean cube $\Boo^{n}$
    \item There is a constant-depth-$\IPSLIN$ refutation of degree $\bigO(k \cdot p)$ and size $\bigO(k \cdot np).$
\end{itemize}
\end{theorem}

\noindent
Over fields of large enough characteristic, \cite[Proposition 4.15]{FSTW21} showed that $L(\mathbf{x})-\beta$ has a constant-depth \emph{multilinear}-$\IPSLIN$ refutation of size that depends on the number of possible values $L(\mathbf{x})$ could take over $\Boo^{n}$. \Cref{thm:ubd-const-depth-linear} shows that if we allow non-multilinear $\IPSLIN$ refutation, then the circuit size is small.

\begin{proof}[Proof of \Cref{thm:ubd-const-depth-linear}]
Firstly, since the coefficients of the polynomial $L(\mathbf{x})$ are in the field $\F_{p^{k}}$, $L(\mathbf{x})$ cannot be equal to $\beta \notin \F_{p^{k}}$ for any $\mathbf{x} \in \Boo^{n}$. In other words, $L(\mathbf{x}) - \beta$ has no satisfying assignment over the Boolean cube $\Boo^{n}$.\newline
To show the existence of a low-degree constant-depth-$\IPSLIN$ refutation, we will use \Cref{lemma:nullstellensatz-refutations}. In particular, \Cref{lemma:nullstellensatz-refutations} says that it is sufficient to prove that there exists polynomials $A(\mathbf{x}), B_{1}(\mathbf{x}), \ldots, B_{n}(\mathbf{x})$ such that
\begin{align*}
    A(\mathbf{x}) \cdot (L(\mathbf{x}) - \beta) + \sum_{j = 1}^{n} B_{j}(\mathbf{x}) \cdot (x_{j}^{2} - x_{j}) \; = \; 1,
\end{align*}
where $A(\mathbf{x}), B_{1}(\mathbf{x}), \ldots, B_{n}(\mathbf{x})$ are low-degree polynomials and have constant-depth circuits of size $\mathrm{poly}(n)$.

\paragraph{}Without loss of generality, we can assume that $L(\mathbf{x})$ is a homogeneous degree-$1$ polynomial\footnote{For the sake of less cumbersome notation} because of the following reason. If $L(\mathbf{x})$ has a non-zero constant term $\alpha_{0} \in \F_{p^{k}}$, then we can work with $(\alpha_{0} + \beta) \in \F \setminus \F_{p^{k}}$, instead of $\beta \in \F \setminus \F_{p^{k}}$.\newline

\noindent
Suppose $L(\mathbf{x}) \; = \; \alpha_{1} x_{1} + \cdots + \alpha_{n} x_{n}$, where for each $i \in [n]$, the coefficient $\alpha_{i} \in \F_{p^{k}}$. For any natural number $0 \leq j \leq k$, we define $L_{j}(\mathbf{x})$ to be the following degree-$1$ polynomial:
\begin{align*}
    L_{j}(\mathbf{x}) \; := \; \alpha_{1}^{p^{j}} x_{1} + \cdots + \alpha_{n}^{p^{j}} x_{n} - \beta^{p^{j}}
\end{align*}
In the above notation, $L_{0}(\mathbf{x}) = L(\mathbf{x}) - \beta$. The next claim shows that we can express $L_{j}(\mathbf{x})$ as a multiple of $L_{0}(\mathbf{x})$ modulo the ideal\footnote{Recall that $(\mathbf{x}^{p} - \mathbf{x}) = (x_{1}^{p} - x_{1}, \ldots, x_{n}^{p} - x_{n})$} $( \mathbf{x}^{p} - \mathbf{x})$.

\begin{claim}\label{claim:Lj-L0-mod-const-depth}
For every $j \in [k]$, there exists polynomials $A_{j}(\mathbf{x}), B_{j,1}(\mathbf{x}), \ldots, B_{j,n}(\mathbf{x})$ such that:
\begin{align*}
    L_{j}(\mathbf{x}) \; = \; A_{j}(\mathbf{x}) \cdot L_{0}(\mathbf{x}) + \sum_{i = 1}^{n} B_{j,i}(\mathbf{x}) \cdot (x_{i}^{p} - x_{i}),
\end{align*}
where each polynomial $A_{j}(\mathbf{x}), B_{j,1}(\mathbf{x}), \ldots, B_{j,n}(\mathbf{x})$
\begin{itemize}
    \item The polynomial $A_{j}(\mathbf{x})$ is a degree-$\bigO(j \cdot p)$ polynomial and has a circuit of size $\bigO(j \cdot (n+p))$ and depth $2$.
    \item For each $j \in [n]$, the polynomial $B_{j}(\mathbf{x})$ is a degree-$\bigO(j \cdot p)$ polynomial and has a circuit of size $\bigO(j \cdot np + j^{2})$ and depth $3$.
\end{itemize}
\end{claim}
\begin{proof}[Proof of \Cref{claim:Lj-L0-mod-const-depth}]
The proof is via induction on $j$.

\paragraph{Base case ($j = 1$):}As we are working over a field $\F$ of characteristic $p$, we have:
\begin{gather*}
    L_{0}(\mathbf{x})^{p} \; = \; \paren{\sum_{i=1}^{n} \alpha_{i} x_{i} - \beta}^{p} \; = \; \sum_{i=1}^{n} \alpha^{p} x_{i}^{p} - \beta^{p} \quad \quad (\text{Using \Cref{lemma:freshman}}) \\ \\
   \Rightarrow L_{0}(\mathbf{x})^{p} \; = \; \underbrace{\paren{\sum_{i=1}^{n} \alpha_{i}^{p} x_{i} - \beta^{p}}}_{= L_{1}(\mathbf{x})} + \sum_{i=1}^{n} \alpha_{i}^{p} \cdot (x_{i}^{p} - x_{i}) \quad \quad (\text{Adding and subtracting terms}) \\ \\
   \Rightarrow L_{1}(\mathbf{x}) \; = \; \underbrace{L_{0}(\mathbf{x})^{p-1}}_{:= A_{1}(\mathbf{x})} \cdot L_{0}(\mathbf{x}) + \sum_{i = 1}^{n} B_{1,i}(\mathbf{x}) \cdot (x_{i}^{p} - x_{i}),
\end{gather*}
where
\begin{itemize}
    \item $A_{1}(\mathbf{x}) = L_{0}(\mathbf{x})^{p-1}$ has a circuit of size $\bigO(np)$ and depth $2$ (a $\Pi \Sigma$ circuit). Additionally, $A_{1}(\mathbf{x})$ is a degree-$\bigO(p)$ polynomial.
    \item For each $i \in [n]$, $B_{1,i}(\mathbf{x}) = -\alpha_{i}^{p}$ has a circuit of size $\bigO(1)$ and depth $1$. Additionally, $B_{1,i}(\mathbf{x})$ is a constant, so has degree-$0$.
\end{itemize}

\paragraph{Induction step:}Now assume the induction hypothesis is true for some $1 \leq j < k$. Proceeding similarly to the base case, we have,
\begin{gather*}
    L_{j}(\mathbf{x})^{p} = \paren{\sum_{i=1}^{n} \alpha_{i}^{p^{j}} x_{i} - \beta^{p^{j}}}^{p} \; = \; \sum_{i=1}^{n} \alpha^{p^{j+1}} x_{i}^{p} - \beta^{p^{j+1}}\quad \quad (\text{Using \Cref{lemma:freshman}}) \\ \\
   \Rightarrow L_{j}(\mathbf{x})^{p} \; = \; \underbrace{\paren{\sum_{i=1}^{n} \alpha_{i}^{p^{j+1}} x_{i} - \beta^{p^{j+1}}}}_{:= L_{j+1}(\mathbf{x})} \; + \;\sum_{i=1}^{n} \alpha_{i}^{p^{j+1}} \cdot (x_{i}^{p} - x_{i})  \quad \quad (\text{Adding and subtracting terms})
\end{gather*}
\begin{equation}\label{eqn:ubd-constant-induction}
    \Rightarrow L_{j+1}(\mathbf{x}) \; = \; L_{j}(\mathbf{x}) \cdot L_{j}(\mathbf{x})^{p-1} +\sum_{i=1}^{n} (-\alpha_{i}^{p^{j+1}}) \cdot (x_{i}^{p} - x_{i})
\end{equation}
Using the induction hypothesis, we know there exists polynomials $A_{j}(\mathbf{x}), B_{j,1}(\mathbf{x}), \ldots, B_{j,n}(\mathbf{x})$ such that
\begin{align*}
    L_{j}(\mathbf{x}) \; = \; A_{j}(\mathbf{x}) \cdot L_{0}(\mathbf{x}) + \sum_{i=1}^{n} B_{j,i}(\mathbf{x}) \cdot (x_{i}^{p} - x_{i}),
\end{align*}
where the polynomials satisfy the size constraints as stated in \Cref{claim:Lj-L0-mod-const-depth}. Substituting this in \Cref{eqn:ubd-constant-induction}, we get,
\begin{align*}
    L_{j+1}(\mathbf{x}) \; = \; \paren{A_{j}(\mathbf{x}) \cdot L_{0}(\mathbf{x}) + \sum_{i=1}^{n} B_{j,i}(\mathbf{x}) \cdot (x_{i}^{p} - x_{i})} \cdot L_{j}(\mathbf{x})^{p-1} + \sum_{i=1}^{n} (-\alpha_{i}^{p^{j+1}}) \cdot (x_{i}^{p} - x_{i}) \\
    \Rightarrow L_{j+1}(\mathbf{x}) \; = \; \underbrace{(A_{j}(\mathbf{x}) \cdot L_{j}(\mathbf{x})^{p-1})}_{:= A_{j+1}(\mathbf{x})} \, \cdot \, L_{0}(\mathbf{x}) \; + \; \sum_{i=1}^{n} \, \underbrace{(B_{j,i}(\mathbf{x}) L_{j}(\mathbf{x})^{p-1} - \alpha_{i}^{p^{j+1}})}_{:= B_{j+1,i}(\mathbf{x})} \, \cdot \, (x_{i}^{p} - x_{i})
\end{align*}
Now,
\begin{itemize}
    \item The polynomial $A_{j+1}(\mathbf{x})$ has a circuit of size $\bigO((j+1) \cdot np)$ and depth $2$ (note that $A_{j}(\mathbf{x})$ is a product of powers of linear polynomials). Additionally, $A_{j+1}(\mathbf{x})$ is a degree-$\big((j+1) \cdot p)$ polynomial.
    \item For every $i \in [n]$, the polynomial $B_{j+1,i}(\mathbf{x})$ has a circuit of size $\bigO((j+1) \cdot np + (j+1)^{2})$ and depth $3$ (note that $B_{j}(\mathbf{x})$ is a $\Sigma \Pi \Sigma$ circuit). Additionally, $B_{j+1,i}(\mathbf{x})$ is a degree-$\bigO((j+1) \cdot p)$ polynomial.
\end{itemize}
This finishes the induction and also the proof of \Cref{claim:Lj-L0-mod-const-depth}.
\end{proof}

\paragraph{}So far in \Cref{claim:Lj-L0-mod-const-depth}, we have shown that the linear polynomial $L_{k}(\mathbf{x})$ is a multiple of the linear polynomial $L_{0}(\mathbf{x})$ modulo the ideal $(\mathbf{x}^{p} - \mathbf{x})$. Next we use the fact that $\beta \notin \F_{p^{k}}$ to show that $L_{k}(\mathbf{x})$ and $L_{0}(\mathbf{x})$ differ by a non-zero constant.
\begin{observation}
The polynomial $L_{k}(\mathbf{x}) - L_{0}(\mathbf{x})$ is a \emph{non-zero constant} polynomial. This is because $\alpha_{i}^{p^{k}} = \alpha_{i}$ (since $\alpha_{i} \in \F_{p^{k}}$) and on the other hand, $\beta^{p^{k}} \neq \beta$. 
\end{observation}

\noindent
\Cref{claim:Lj-L0-mod-const-depth} gives us that there exists polynomials $A_{k}(\mathbf{x}), B_{k,1}(\mathbf{x}), \ldots, B_{k,n}(\mathbf{x})$ satisfying:
\begin{align*}
    L_{k}(\mathbf{x}) \; = \; A_{k}(\mathbf{x}) \cdot L_{0}(\mathbf{x}) + \sum_{i=1}^{n} B_{k,i}(\mathbf{x}) \cdot (x_{i}^{p} - x_{i}) \\
    \Rightarrow L_{k}(\mathbf{x}) - L_{0}(\mathbf{x}) \; = \; (A_{k}(\mathbf{x}) - 1) \cdot L_{0}(\mathbf{x}) + \sum_{i = 1}^{n} B_{k,i}(\mathbf{x}) \cdot (x_{i}^{p} - x_{i})
\end{align*}
\begin{equation}\label{eqn:ubd-non-ml-final-linear}
    \Rightarrow \dfrac{(A_{k}(\mathbf{x}) - 1)}{\beta^{p^{k}} - \beta} \cdot L_{0}(\mathbf{x}) \; + \; \sum_{i = 1}^{n} \; \dfrac{B_{k,i}(\mathbf{x})}{\beta^{p^{k}} - \beta} \cdot (x_{i}^{p} - x_{i}) \; = \; 1,
\end{equation}
where in the final implication we used that $L_{k}(\mathbf{x}) - L_{0}(\mathbf{x}) = \beta^{k} - \beta$ is a non-zero constant. For each $i \in [n]$, the polynomial $(x_{i}^{p} - x_{i})$ is a multiple of $(x_{i}^{2} - x_{i})$ because:
\begin{align*}
    x_{i}^{p} - x_{i} \; = \; (x_{i}^{p-2} + \cdots + x_{i} + 1) \cdot (x_{i}^{2} - x_{i})
\end{align*}
Substituting it back in \Cref{eqn:ubd-non-ml-final-linear}, we get,
\begin{align*}
    \underbrace{\dfrac{(A_{k}(\mathbf{x}) - 1)}{\beta^{p^{k}} - \beta}}_{:= A(\mathbf{x})} \cdot L_{0}(\mathbf{x}) \; + \; \sum_{i = 1}^{n} \; \underbrace{\paren{\dfrac{B_{k,i}(\mathbf{x})}{\beta^{p^{k}} - \beta} \cdot (x_{i}^{p-2} + \cdots + x_{i} + 1) }}_{:= B_{i}(\mathbf{x})} \cdot (x_{i}^{2} - x_{i}) \; = \; 1
\end{align*}

\paragraph{Degree and Size Analysis}We define $A(\mathbf{x})$ and $B_{i}(\mathbf{x})$ as follows:
\begin{itemize}
    \item \Cref{claim:Lj-L0-mod-const-depth} says that $A_{k}(\mathbf{x})$ is a degree-$\bigO(kp)$ polynomial and is computable by circuit of size $\bigO(k  (n+p))$ and depth $2$ (a $\Pi \Sigma$ circuit). Hence $A(\mathbf{x})$ is a degree-$\bigO(k p)$ polynomial and is computable by a circuit of size $\bigO(k (n+p))$ and depth $3$ (a $\Sigma \Pi \Sigma$ circuit).
    
    \item \Cref{claim:Lj-L0-mod-const-depth} says that $B_{k,i}(\mathbf{x})$ is a degree-$\bigO(k p)$ polynomial and is computable by a circuit of size $\bigO(k np + k^{2})$ and depth $3$ (a $\Sigma \Pi \Sigma$ circuit). Hence $B_{i}(\mathbf{x})$ is a degree-$\bigO(k p)$ polynomial and is computable by a circuit of size $\bigO(k np + k^{2})$ and depth $3$.
\end{itemize}
Thus we have shown that there is a low-degree constant-depth-$\IPSLIN$ refutation of $L(\mathbf{x}) - \beta$ and this finishes the proof of \Cref{thm:ubd-const-depth-linear}.
\end{proof}

\paragraph{}Now we ready to prove \Cref{thm:ubd-const-depth} using \Cref{thm:ubd-const-depth-linear}. The idea is to replace each monomial in the sparse polynomial by a new variable, resulting in a linear polynomial in the new variables. A refutation of the resulting linear polynomial can be ``lifted'' to a refutation of the sparse polynomial in the original variables. We use the refutation of linear polynomials from \Cref{thm:ubd-const-depth-linear}, and to lift this refutation, we need to show that \emph{monomial axioms} are in the ideal of the Boolean axioms. Before proceeding, we will prove the following claim on monomial axioms. It follows from a straightforward induction on the number of variables. We will omit the proof here, and it can be found in \Cref{app:proof-monomial-axioms}.\\

\begin{restatable}{claim}{monomialaxioms}\label{claim:monomial-axioms}
For any exponent vector $\bm{\mu} = (\mu_{1},\ldots,\mu_{n})$ with $|\bm{\mu}| \leq D$, there exists polynomials $E_{\bm{\mu},1}(\mathbf{x}), \ldots, E_{\bm{\mu},n}(\mathbf{x})$ such that the following holds:
\begin{align*}
    ((\mathbf{x}^{\bm{\mu}})^{2} - \mathbf{x}^{\bm{\mu}}) \; = \; \sum_{\substack{j \in [n] \\ \mu_{j} > 0}} E_{\bm{\mu},j}(\mathbf{x}) \cdot (x_{j}^{2} - x_{j}),
\end{align*}
and for each $j \in [n]$ with $\mu_{j} > 0$, the polynomial $E_{\bm{\mu},j}(\mathbf{x})$ has a circuit of size $\bigO(nD^{2})$ and depth $2$ (a $\Pi \Sigma$ circuit).
\end{restatable}

\noindent
Below we recall \Cref{thm:ubd-const-depth} and proceed to prove it.

\ubdnonmlposchar*

\begin{proof}[Proof of \Cref{thm:ubd-const-depth}]
From \Cref{lemma:nullstellensatz-refutations}, it suffices to show that there exists coefficients $A(\mathbf{x})$ and $B_{j}(\mathbf{x})$'s in the ring $\F[x_{1},\ldots,x_{n}]$ such that
\begin{align*}
    A(\mathbf{x}) \cdot (f(\mathbf{x})-\beta) \, + \, \sum_{j=1}^{n} B_{j}(\mathbf{x}) \cdot (x_{j}^{2} - x_{j}) \; = \; 1,
\end{align*}
where the $A(\mathbf{x})$ and $B_{j}(\mathbf{x})$'s have constant-depth circuits of $\mathrm{poly}(n)$-size and degree $\bigO(kpD)$.\\

\noindent
Let $f(\mathbf{x}) = \sum_{\bm{\mu} : |\bm{\mu}| \leq D} \alpha_{\bm{\mu}} \mathbf{x}^{\bm{\mu}}$, where $\bm{\mu}$ denotes an exponent vector. Define the support of $f(\mathbf{x})$:
\begin{align*}
    \mathrm{Supp}(f) \; = \; \setcond{\bm{\mu} \subseteq [n]}{\alpha_{\bm{\mu}} \neq 0}
\end{align*}
The cardinality of $\mathrm{Supp}(f)$ is equal to the sparsity of $f(\mathbf{x})$ which is $s$.

\paragraph{\underline{Reducing to linear polynomial}}For every $\bm{\mu} \in \mathrm{Supp}(f)$, define a new variable $y_{\bm{\mu}}$, i.e. $s$ new $y$ variables. Let $F(\mathbf{y})$ denote the polynomial when we replace the monomials in $f(\mathbf{x})$ with the new $y$-variables, i.e.
\begin{align*}
    F(\mathbf{y}) \; = \; \sum_{\bm{\mu} \in \mathrm{Supp}(f)} \alpha_{\bm{\mu}} y_{\bm{\mu}}
\end{align*}
Thus $F(\mathbf{y})$ is a degree-$1$ polynomial in $s$ variables.

\paragraph{}Observe that $F(\mathbf{y}) - \beta$ does not have a solution over the Boolean hypercube $\Boo^{s}$ since $\beta \notin \F_{p^{k}}$. From the proof of \Cref{thm:ubd-const-depth-linear} on the degree-$1$ polynomial $F(\mathbf{y}) - \beta$, we get that there exists polynomials $\Tilde{A}(\mathbf{y})$ and $\Tilde{B}_{1}(\mathbf{y}), \ldots, \Tilde{B}_{s}(\mathbf{y})$ such that the following holds:
\begin{equation}\label{eqn:ubd-sparse-constant-depth}
    \Tilde{A}(\mathbf{y}) \cdot (F(\mathbf{y}) - \beta) \; + \; \sum_{\bm{\mu} \in \mathrm{Supp}(f)} \Tilde{B}_{\bm{\mu}}(\mathbf{y}) \cdot (y_{\bm{\mu}}^{2} - y_{\bm{\mu}}) \; = \; 1,
\end{equation}
where
\begin{itemize}
    \item The polynomial $\Tilde{A}(\mathbf{y})$ is a degree-$\bigO(kp)$ polynomial and is computable by a circuit of size $\bigO(k(s+p))$ and depth $3$ (a $\Sigma \Pi \Sigma$ circuit).
    \item For each $\bm{\mu} \in \mathrm{Supp}(f)$, the polynomial $\Tilde{B}_{\bm{\mu}}(\mathbf{y})$ is a degree-$\bigO(kp)$ polynomial and is computable by a circuit of size $\bigO(ksp + k^{2})$ and depth $3$ (a $\Sigma \Pi \Sigma$ circuit).
\end{itemize}

\paragraph{\underline{Lifting the Nullstellensatz certificate}}Plugging in $y_{S} = \mathbf{x}^{S}$ in the \Cref{eqn:ubd-sparse-constant-depth}, we get,
\begin{equation}\label{eqn:ubd-sparse-lifted-cert}
    \underbrace{\Tilde{A}(\mathbf{y}) \circ \mathbf{x}}_{:= A(\mathbf{x})} \cdot (f(\mathbf{x}) - \beta) \; + \; \sum_{\bm{\mu} \in \mathrm{Supp}(f)} \, \underbrace{\Tilde{B}_{\bm{\mu}}(\mathbf{y}) \circ \mathbf{x}}_{B_{\bm{\mu}}'(\mathbf{x})} \cdot ((\mathbf{x}^{\bm{\mu}})^{2} - \mathbf{x}^{\bm{\mu}}) \; = \; 1,
\end{equation}
where
\begin{itemize}
    \item The polynomial $A(\mathbf{x})$ is a degree-$\bigO(kp \cdot D)$ polynomial and is computable by a circuit of size $\bigO(k (s+p) + sD)$ and depth $4$ (a $\Sigma \Pi \Sigma \Pi$ circuit)
    \item For each $\bm{\mu} \in \mathrm{Supp}(f)$, the polynomial $B_{\bm{\mu}}'(\mathbf{x})$ is a degree-$\bigO(k p \cdot D)$ polynomial and is computable by a circuit of size $\bigO(k sp + k^{2} + s D)$ and depth $4$ (a $\Sigma \Pi \Sigma \Pi$ circuit).
\end{itemize}

\noindent
Now applying \Cref{claim:monomial-axioms} for each subset $\bm{\mu} \in \mathrm{Supp}(f)$ in the ``lifted'' Nullstellensatz certificate \Cref{eqn:ubd-sparse-lifted-cert},
\begin{align*}
    A(\mathbf{x}) \cdot (f(\mathbf{x}) - \beta) \; + \; \sum_{\bm{\mu} \in \mathrm{Supp}(f)} B_{\bm{\mu}}'(\mathbf{x}) \cdot \paren{\sum_{j=1}^{n} E_{\bm{\mu},j}(\mathbf{x})} \cdot (x_{j}^{2} - x_{j}) \; = \; 1 \\
    \Rightarrow A(\mathbf{x}) \cdot (f(\mathbf{x}) - \beta) \; + \; \sum_{j=1}^{n} \; \underbrace{\paren{\sum_{\bm{\mu} \in \mathrm{Supp}(f)} B_{\bm{\mu}}'(\mathbf{x}) \cdot E_{\bm{\mu},j}(\mathbf{x}) }}_{:= B_{j}(\mathbf{x})} \cdot (x_{j}^{2} - x_{j}) \; = \; 1,    
\end{align*}
where for each $j \in [n]$, the polynomial $B_{j}(\mathbf{x})$ is a degree-$\bigO(kpD)$ polynomial and is computable by a circuit of size $\bigO(k sp + sD)$ and depth $5$. This finishes the proof of \Cref{thm:ubd-const-depth}.
\end{proof}

\subsection{Proof of \Cref{thm:ub-degree}}\label{subsec:proof-ub-degree}
In this section, we are going to show \Cref{thm:ub-degree}, which we recall below.\\

\degreeupperbound*

\paragraph{}Observe that the size bound is the ``trivial'' one, i.e. a $n$-variate multilinear polynomial with degree $D$ has at most $\binom{n}{\leq D}$ monomials. Letting $D = \bigO(kp)$, we get the stated size bound in \Cref{thm:ub-degree}. So in our proof of \Cref{thm:ub-degree}, it will be enough to prove that there is a Nullstellensatz certificate of degree $\bigO(kp)$. As we will show, it will be sufficient to show that the multilinear polynomial equivalent to $1/(\sum \alpha_{i} x_{i} - \beta)$ on $\Boo^{n}$ has degree $\bigO(kp)$. This will be our main technical lemma in the proof of \Cref{thm:ub-degree}, which we state and prove next.\\

\begin{lemma}[Degree of the ``inverse'' polynomial]\label{lemma:deg-inverse-poly}
Fix a prime $p$, a parameter $k \in \mathbb{N}$ and finite field $\F_{p^{k}}$. The following holds for every $\alpha_{1},\ldots,\alpha_{n},\beta \in \F_{p^{k}}$ for which the equation $\sum_{i=1}^{n} \alpha_{i} x_{i} - \beta = 0$ is unsatisfiable over the Boolean cube $\Boo^{n}$.\newline

If $f \in \F[x_{1},\ldots,x_{n}]$ is a multilinear polynomial that agrees with $1/(\sum_{i=1}^{n} \alpha_{i}x_{i} - \beta)$ on $\Boo^{n}$, i.e.
\begin{align*}
    f \; \equiv \; \dfrac{1}{\sum_{i=1}^{n} \alpha_{i} x_{i} - \beta} \mod{(\mathbf{x}^{2} - \mathbf{x})},
\end{align*}
then $\deg(f) \leq k \cdot (p-1)$.
\end{lemma}
\begin{proof}[Proof of \Cref{lemma:deg-inverse-poly}]
Let $L(\mathbf{x}) := \sum_{i=1}^{n} \alpha_{i} x_{i} - \beta$, $q = p^{k}$ and $(m_{0},\ldots,m_{k-1})$ denote the $p$-ary representation of $(q-2)$ i.e.
\begin{align*}
    q-2 \; = \; \sum_{j=0}^{k-1} m_{j} p^{j}, && \text{for all } \; i, \; 0\leq m_{i} \leq (p-1).
\end{align*}
The hypothesis says that for every $\mathbf{a} \in \Boo^{n}$, $L(\mathbf{a}) \neq 0$. As we are working over the field $\F_{q}$, we get that for every $\mathbf{a} \in \Boo^{n}$, $L(\mathbf{a}) \cdot (L(\mathbf{a}))^{q-2} = 1$. In other words,
\begin{align*}
 \ml[(L(\mathbf{x}))^{q-2}] \; \equiv \; \dfrac{1}{L(\mathbf{x})} \, \mod{(\mathbf{x}^{2} - \mathbf{x})}.
\end{align*}
Since multilinear extension of a boolean function is unique, we get that $f = \ml[L^{q-2}]$, where $f \in \F_{q}[\mathbf{x}]$ is as defined in the statement of \Cref{lemma:deg-inverse-poly}. So we will now show that $\deg(\ml[(L(\mathbf{x}))^{q-2}])$ is $k (p-1)$.\\

\noindent
For every non-negative integer $j \geq 0$, repeated applications of \Cref{lemma:freshman} gives us:
\begin{align*}
    (L(\mathbf{x}))^{p^{j}} \; = \; \sum_{i=1}^{n} \alpha_{i}^{p^{j}} x_{i}^{p^{j}} - \beta^{p^{j}} \quad
    \Rightarrow \quad \ml[(L(\mathbf{x}))^{p^{j}}] \; = \; \sum_{i=1}^{n} \alpha_{i}^{p^{j}} x_{i} - \beta^{p^{j}}.
\end{align*}
For simplicity in notation, for every $j$, let $L_{j}(\mathbf{x}) := \ml[(L(\mathbf{x}))^{p^{j}}]$, and as we just showed, $\deg(L_{j}) = 1$. Using the $p$-ary expansion of $(q-2)$ and the third item of \Cref{fact:multilinear}, we have,
\begin{align*}
    L(\mathbf{x})^{q-2} \; = \; \prod_{j=0}^{k-1} \; (L(\mathbf{x})^{p^{j}})^{m_{j}} \quad \Rightarrow \ml[(L(\mathbf{x}))^{q-2}] \; = \; \ml\brac{  \prod_{j=0}^{k-1} \; \ml[L_{j}(\mathbf{x})  ^{m_{j}}] } \\
    \Rightarrow \deg(\ml[(L(\mathbf{x}))^{q-2}]) \; \leq \; \sum_{j=0}^{k-1} \deg(\ml[L_{j}(\mathbf{x})  ^{m_{j}}]) \; \leq \; \sum_{j=0}^{k-1} m_{j} \leq k \cdot (p-1).
\end{align*}
Hence we have showed that the degree of $f$ is $\leq k(p-1)$ and this finishes the proof of \Cref{lemma:deg-inverse-poly}.
\end{proof}



\noindent
We now prove \Cref{thm:ub-degree} using an almost straightforward application of \Cref{lemma:deg-inverse-poly}.\\

\begin{proof}[Proof of \Cref{thm:ub-degree}]
From \Cref{lemma:nullstellensatz-refutations}, it suffices to show that there exists coefficients $A(\mathbf{x})$ and $B_{j}(\mathbf{x})$'s in the ring $\F_{p^{k}}[\mathbf{x}]$ such that
\begin{align*}
    A(\mathbf{x}) \cdot \paren{\sum_{i=1}^{n} \alpha_{i} x_{i} - \beta} \, + \, \sum_{j=1}^{n} B_{j}(\mathbf{x}) \cdot (x_{j}^{2} - x_{j}) \; = \; 1,
\end{align*}
where $A(\mathbf{x})$ and $B_{j}(\mathbf{x})$'s have constant-depth circuits of size $\bigO(n/kp)^{\bigO(kp)}$ and degree $\bigO(kp)$.\\

Let $L(\mathbf{x}) := \sum_{i=1}^{n} \alpha_{i} x_{i} - \beta$. Let $A \in \F[\mathbf{x}]$ be the multilinear polynomial such that for every $\mathbf{x} \in \Boo^{n}$, $A(\mathbf{x})$ equals $1/L(\mathbf{x})$ (note that $L(\mathbf{x}) \neq 0$ for every $\mathbf{x} \in \Boo^{n}$ because $L$ is unsatisfiable over the Boolean cube). Applying \Cref{lemma:deg-inverse-poly}, we get that $\deg(A) \leq k\cdot(p-1)$.\newline
Since $A(\mathbf{x})$ is a $n$-variate multilinear polynomial of degree $\leq k(p-1)$, it has at most $\binom{n}{\leq D}$ monomials. Using Stirling's approximation, we get that the number of monomials is $\bigO(n/kp)^{\bigO(kp)}$, which implies a $\Sigma \Pi$ circuit for $A(\mathbf{x})$ of size $\bigO(n/kp)^{\bigO(kp)}$ and degree $\bigO(kp)$.\\

\noindent
Now it remains to argue for $B_{j}(\mathbf{x})$'s. Let $B_{1}(\mathbf{x})$ be the quotient and $R_{1}(\mathbf{x})$ be the remainder when $A(\mathbf{x}) \cdot L(\mathbf{x})$ is divided by $(x_{1}^{2}-x_{1})$,
\begin{align*}
    A(\mathbf{x}) \cdot L(\mathbf{x}) \; = \; B_{1}(\mathbf{x}) \cdot (x_{1}^{2} - x_{1}) + R_{1}(\mathbf{x}).
\end{align*}
Clearly $\deg(B_{1}), \deg(R_{1}) \leq \deg(A) + 1$. Next, let $B_{2}(\mathbf{x})$ denote the quotient and $R_{2}(\mathbf{x})$ denote the remainder when $R_{1}(\mathbf{x})$ is divided by $(x_{2}^{2} - x_{2})$, and so on. Since $A(\mathbf{x}) \cdot L(\mathbf{x}) - 1 \in (\mathbf{x}^{2} - \mathbf{x})$, we know that $R_{n}(\mathbf{x}) = 1$. In other words,
\begin{align*}
    A(\mathbf{x}) \cdot L(\mathbf{x}) \; = \; \sum_{j=1}^{n} B_{j}(\mathbf{x}) \cdot (x_{j}^{2} - x_{j}) + 1  \\
    \Rightarrow A(\mathbf{x}) \cdot \paren{\sum_{i=1}^{n} \alpha_{i} x_{i} - \beta} + \sum_{j=1}^{n} (-B_{j}(\mathbf{x})) \cdot (x_{j}^{2} - x_{j}) \; = \; 1.
\end{align*}
Here, for each $j \in [n]$, $\deg(B_{j})\leq \deg(A) + 1$. Similar to $A(\mathbf{x})$, each $B_{j}(\mathbf{x})$ has a $\Sigma \Pi$ circuit of size $\bigO(n/kp)^{\bigO(kp)}$ and degree $\bigO(kp)$. This finishes the proof of \Cref{thm:ub-degree}.
\end{proof}

\section{Symmetric Refutations in Constant Depth}
In this section, we will prove \Cref{thm:ubd-sym-const-depth}, which we recall below.
\symconstant*

 \paragraph{}One of the steps in our proof of \Cref{thm:ubd-sym-const-depth} is a \emph{multilinearization} step, i.e. given a polynomial $f(\mathbf{x})$, we want to find a certificate in constant-depth circuits certifying that $f(\mathbf{x})$ and $\ml[f(\mathbf{x})]$ agree on the Boolean cube $\Boo^{n}$. More formally, we are interested in finding polynomials $B_{j}(\mathbf{x})$'s such that
 \begin{align*}
     f(\mathbf{x}) \; = \; \ml[f(\mathbf{x})] + \sum_{j=1}^{n} B_{j}(\mathbf{x}) \cdot (x_{j}^{2} - x_{j}),
 \end{align*}
 and the polynomials $B_{j}(\mathbf{x})$ have a $\mathrm{poly}(n)$-sized constant-depth circuit.

\paragraph{}We also need a few standard facts about elementary symmetric polynomials in fields of positive characteristic. A standard fact that is useful in our proof is that a symmetric function over the Boolean cube in constant positive characteristic only depends on $\bigO(\log n)$ elementary symmetric polynomials (instead of $n$ elementary symmetric polynomials for symmetric polynomials over arbitrary domains). We now give a proof below for the sake of completeness.

\begin{lemma}[Lucas's Theorem \cite{Lucas}]\label{lemma:lucas}
Fix a prime number $p$ and any two natural numbers $a$ and $b$. Denote $a$ and $b$ in their unique $p$-ary representations as:
\begin{align*}
    a = \sum_{i = 0}^{\ell - 1} a_{i} p^{i}, \quad \quad b = \sum_{i = 0}^{\ell - 1} b_{i} p^{i}, \quad \quad  a_{i}, b_{i} \in \{0,1,\ldots, p-1\}
\end{align*}
Then,
\begin{align*}
    \displaystyle\binom{a}{b} \; \equiv \; \prod_{i = 0}^{\ell - 1} \displaystyle\binom{a_{i}}{b_{i}} \; \mod{p},
\end{align*}
where we define $\binom{x}{y}$ to be $0$ if $x < y$.
\end{lemma}

Next, we show that a symmetric function over the Boolean cube $\Boo^{n}$ in characteristic $p$ depends on $\bigO(\log n)$ elementary symmetric polynomials.

\begin{claim}[Symmetric functions over $\Boo^{n}$ in positive char]\label{claim:sym-char-p}
Fix a prime number $p$ and a field $\F$ with $\mathrm{char}(\F) = p$. Fix a variable parameter $n \in \mathbb{N}$.\newline
Let $f(\mathbf{x}) \in \F[x_{1},\ldots, x_{n}]$ be a multilinear and symmetric polynomial. Then $f(\mathbf{x})$ is a function of $\bigO(\log_{p} n)$ elementary symmetric polynomials on $n$ variables.
\end{claim}
\begin{proof}[Proof of \Cref{claim:sym-char-p}]
Fix any natural number $0 \leq d \leq n$ and consider the $d^{th}$ elementary symmetric polynomial, i.e., $e_{d}(x_{1},\ldots,x_{n})$. Fix an arbitrary point $\mathbf{a} \in \Boo^{n}$ and let $k = |\mathbf{a}|$, where $|\mathbf{a}|$ denotes the Hamming weight of $\mathbf{a}$. We know that $e_{d}(\mathbf{a}) = \binom{k}{d}$.
Denote $k$ and $d$ in their unique $p$-ary representation, i.e.
\begin{align*}
    k \; = \; \sum_{i = 0}^{\ell} k_{i} p^{i}, \quad \text{ and } \quad d \; = \; \sum_{i = 0}^{\ell} d_{i} p^{i}, \quad \quad k_{i}, d_{i} \in \set{0,1,\ldots,p-1}
\end{align*}
By Lucas's Theorem (\Cref{lemma:lucas}), we have,
\begin{align*}
    \binom{k}{d} \; = \; \binom{k_{\ell}}{d_{\ell}} \cdots \binom{k_{0}}{d_{0}} \mod{p}
\end{align*}
Note that $e_{d_{i} p^{i}}(\mathbf{a}) = \binom{k}{d_{i} p^{i}}$. For every $0 \leq i \leq \ell$, using Lucas's Theorem (\Cref{lemma:lucas}), we have,
\begin{equation}\label{eqn:sym-p-powers-1}
    \binom{k}{d_{i}p^{i}} \; = \; \binom{k_{i}}{d_{i}} \mod{p} \quad \Rightarrow \quad \binom{k}{d_{i}p^{i}} \; = \; \dfrac{1}{d_{i}!} \cdot \prod_{j = 0}^{d-1} \paren{\binom{k_{i}}{1} - j} \quad \quad (d_{i}! \neq 0 \mod{p})
\end{equation}
Using Lucas's Theorem (\Cref{lemma:lucas}), we have
\begin{equation}\label{eqn:sym-p-powers-2}
    \binom{k_{i}}{1} \; = \; \binom{k}{p^{i}} \mod{p}
\end{equation}
Define the polynomial $S_{d,i}(z) := \frac{1}{d_i!}\prod_{j = 0}^{d_i-1} \paren{z - j}$. Note that $e_{p^{i}}(\mathbf{a}) = \binom{k}{p^{i}} \mod{p}$.
Plugging in \Cref{eqn:sym-p-powers-2} in \Cref{eqn:sym-p-powers-1}, we get
\begin{align*}
    e_{d}(\mathbf{a}) \; = \; \binom{k}{d} \; = \; \prod_{i = 1}^{\ell} \, S_{d,i}(e_{p^{i}}(\mathbf{a}))
\end{align*}
We have shown that $e_{d}(\mathbf{a})$ is a polynomial of $e_{p^{i}}(\mathbf{a})$ for $i \in \set{0,1,\ldots,\ell}$. Since $\mathbf{a}$ was an arbitrarily chosen point in $\Boo^{n}$, we just argued that on the Boolean hypercube, $e_{d}(\mathbf{x})$ is a polynomial of $e_{p^{0}}(\mathbf{x}),\ldots,e_{p^{\ell}}(\mathbf{x})$. This holds for every $0 \leq d \leq n$. Hence, every symmetric function on $\Boo^{n}$ in characteristic $p$ is a polynomial of $e_{p^{0}}(\mathbf{x}),\ldots,e_{p^{\ell}}(\mathbf{x})$, i.e. of $\bigO(\log_{p} n)$ elementary symmetric polynomials.
\end{proof}

A key lemma in our proof is the multilinearization lemma \Cref{lemma:multilinearize-char-p}, which shows that multilinearization of a sparse polynomial in $\widehat{\mathbf{e}}(\mathbf{x})$ has a small constant-depth circuit.\\

\begin{restatable}[Multilinearization of polynomial of elementary symmetric polynomials]{lemma}{multsymposchar}\label{lemma:multilinearize-char-p}
Fix a prime number $p$ and a field $\F$ with $\mathrm{char}(\F) = p$. Fix a variable parameter $r \in \mathbb{N}$.\newline
Let $F(\mathbf{y}) \in \F[y_{1},\ldots,y_{r}]$ be a polynomial with individual degree strictly less than $p$. Then \break $\ml[F(e_{1}(\mathbf{y}), e_{p}(\mathbf{x}), \ldots, e_{p^{r-1}}(\mathbf{x}))]$ has a circuit of size $\bigO( n^{5} \log n)$ and depth $5$.    
\end{restatable}

\noindent
We will prove \Cref{lemma:multilinearize-char-p} later. For now, we show how it is useful in proving \Cref{thm:ubd-sym-const-depth}.

\begin{proof}[Proof of \Cref{thm:ubd-sym-const-depth}]
We will first prove in the setting when the underlying field $\F$ has a small positive char, i.e. $\mathrm{char}(\F) = p$ for a constant prime $p$. The proof of characteristic $0$ or $> n$ is similar and simpler too. We will come back to the setting $\mathrm{char}(\F) = 0$ or $> n$ towards the end of the proof.

\paragraph{}From \Cref{lemma:nullstellensatz-refutations}, we know that it suffices to prove there exists polynomials $A_{i}(\mathbf{x})$'s and $B_{j}(\mathbf{x})$'s in the ring $\F[x_{1},\ldots,x_{n}]$ such that:
\begin{align*}
    \sum_{i = 1}^{m} A_{i}(\mathbf{x}) \cdot f_{i}(\mathbf{x})  + \sum_{j=1}^{n} B_{j}(\mathbf{x}) \cdot (x_{j}^{2} - x_{j}) \; = \; 1,
\end{align*}
where $A_{i}(\mathbf{x})$'s and $B_{j}(\mathbf{x})$'s have $\mathrm{poly}(n)$-sized constant-depth circuits.

\paragraph{\underline{Reducing to few variables}}Let $r$ denote the number of digits when $n$ is expressed in $p$-ary representation. We have $r = \lfloor \log_{p} n \rfloor + 1 \leq 2 \log_{p} n$. \Cref{claim:sym-char-p} tells us that there exists polynomials $F_{1}(\mathbf{y}),\ldots,F_{m}(\mathbf{y}) \in \F[y_{1},\ldots,y_{r}]$ such that:
\begin{align*}
    f_{i}(\mathbf{x}) \; = \; F_{i}(e_{1}(\mathbf{x}), e_{p}(\mathbf{x}), \ldots, e_{p^{r-1}}(\mathbf{x})) \mod{\langle \mathbf{x}^{2} - \mathbf{x} \rangle}
\end{align*}
We will denote the tuple of polynomials $(e_{1}(\mathbf{x}), e_{p}(\mathbf{x}), \ldots, e_{p^{r-1}}(\mathbf{x}))$ by $\widehat{\mathbf{e}}(\mathbf{x})$.\newline
Since $f_{i}(\mathbf{x})$ and $F_{i}(\widehat{\mathbf{e}}(\mathbf{x}))$ agree on the Boolean cube $\Boo^{n}$, their multilinear components are equal, i.e. $\ml[f_{i}(\mathbf{x})] = \ml[F_{i}(\widehat{\mathbf{e}}(\mathbf{x}))]$ (see \Cref{fact:multilinear}). Since $f_{i}(\mathbf{x})$ is a \textbf{multilinear} polynomial, we have $f_{i}(\mathbf{x}) = \ml[F_{i}(\widehat{\mathbf{e}}(\mathbf{x}))]$.\\

\noindent
For every $1 \leq i \leq r$, the polynomial $e_{p^{i-1}}(\mathbf{x})$ take values in $\F_{p}$ over the Boolean cube. For every $j \in [r]$, let $p_{j}(t)$ be a univariate polynomial that vanishes on the set $\F_{p}$, i.e. $p_{j}(t) = \prod_{\alpha \in \F_{p}} (t - \alpha)$.\newline
For every $n < t < p^{r}-1$, define the polynomial $Q_{t} \in \F[\mathbf{y}]$ as follows:
\begin{align*}
    Q_{t}(\mathbf{y}) := \prod_{i=0}^{r-1} \; \dfrac{1}{t_i!}\prod_{j = 0}^{t_i-1} \paren{y_{i} - j}, \quad \quad \text{ where } t = \sum_{i=0}^{r-1} t_{i} p^{i-1}
\end{align*}

\noindent
Our first claim shows that if $f_{i}(\mathbf{x})$'s do not have a common Boolean solution, then $F_{i}(\mathbf{y})$'s along with some additional constraints do not have a common solution, even over the algebraic closure of $\F$.

\begin{claim}\label{claim:ubd-sym-lifting-char-p}
The system consisting of $F_{i}(\mathbf{y})$'s, $Q_{t}(\mathbf{y})$'s, and $p_{j}(y_{j})$'s have no common solution in the closure $\overline{\F}^{n}$, i.e.
\begin{align*}
    \mathbb{V}(F_{1}(\mathbf{y}), \ldots, F_{m}(\mathbf{y}),Q_{n+1}(\mathbf{y}),\ldots,Q_{p^{r}-1}(\mathbf{y}), p_{1}(y_{1}), \ldots, p_{r}(y_{r})) \; = \; \emptyset
\end{align*}
\end{claim}
\begin{proof}
We will prove this by contradiction. Assume for the sake of contradiction that there exists a common solution $\mathbf{b}$ to the above system of polynomials. Since for every $j \in [r]$, $p_{j}(b_{j}) = 0$, this implies that $\mathbf{b} \in \set{0,\ldots,p-1}^{r}$. We will now show the existence of a point $\mathbf{a} \in \Boo^{n}$ such that
\begin{equation}\label{eqn:ubd-const-char-p-surjective}
    \widehat{\mathbf{e}}(\mathbf{a}) = \mathbf{b}
\end{equation}
Observe that an $\mathbf{a} \in \Boo^{n}$ which satisfies \Cref{eqn:ubd-const-char-p-surjective} is a common Boolean solution to the system $\set{f_{1}(\mathbf{x}), \ldots, f_{m}(\mathbf{x})}$, which is a contradiction to our hypothesis of \Cref{thm:ubd-sym-const-depth}. So to finish the contradiction, all that remains is to show the existence of such a Boolean point $\mathbf{a}$.\\

\noindent
We are only interested in showing the existence of a Boolean point satisfying \Cref{eqn:ubd-const-char-p-surjective}, i.e. we are only interested in the evaluation of $\mathbf{a}$ on symmetric polynomials. Thus we can focus on showing the existence of an appropriate Hamming weight $0 \leq k \leq n$ such that
\begin{align*}
    \binom{k}{p^{i}} = \binom{k_{i}}{1} = b_{i} \mod{p}, \quad \quad \text{ for all } \; 0 \leq i \leq r-1 
\end{align*}
where we used Lucas's Theorem (\Cref{lemma:lucas}) for the first equality. Choose $k$ to be $\sum_{i = 0}^{\ell} b_{i} p^{i}$. If $k \leq n$, then $k$ satisfies all the $r$ constraints in \Cref{eqn:ubd-const-char-p-surjective}. In particular, we can set $\mathbf{a} = 1^{k} 0^{n-k}$ and it will satisfy \Cref{eqn:ubd-const-char-p-surjective}. Thus to complete the proof, we need to show that $k \leq n$.\\

\noindent
By assumption, $Q_{t}(\mathbf{b}) = 0$ for all $n < t \leq p^{r}-1$. Fix any $t = \sum_{i=0}^{r-1} t_{i} p^{i-1}$. By the definition of the polynomial $Q_{t}(\mathbf{y})$ (see the proof of \Cref{claim:sym-char-p}):
\begin{align*}
    Q_{t}(\mathbf{b}) \; = \; \prod_{i=0}^{r-1} \binom{b_{i}}{t_{i}} \; = \; 0.
\end{align*}
This means there exists an $i \in \set{0,1,\ldots,r-1}$ such that $b_{i} < t_{i}$. Since this holds for every $n < t \leq p^{r}-1$, this implies that for each coordinate $i$, $b_{i} \leq n_{i}$ where $n = \sum_{i=0}^{r-1} n_{i} p^{i}$. Thus $k \leq n$.\\

\noindent
Hence we have found a Boolean point $\mathbf{a} \in \Boo^{n}$ which satisfies \Cref{eqn:ubd-const-char-p-surjective}. Since for every $i \in [m]$, $f_{i}(\mathbf{x})$ and $F_{i}(\widehat{\mathbf{e}}(\mathbf{x}))$ agree on the Boolean cube $\Boo^{n}$, $\mathbf{a}$ is a common Boolean solution to $f_{i}(\mathbf{x})$'s. This is a contradiction to our assumption in \Cref{thm:ubd-sym-const-depth}.
\end{proof}

\paragraph{\underline{Low-variate Nullstellensatz}}We have shown that the unsatisfiability of the $n$-variate polynomials $f_{i}$'s over the Boolean cube implies the unsatisfiability of $\bigO(\log n)$-variate polynomials $F_{i}$'s (with some additional polynomials to reflect the Boolean cube restriction). Now we use Hilbert's Nullstellensatz to get a Nullstellensatz certificate for the $\bigO(\log n)$-variate polynomials and ``lift'' it to get a Nullstellensatz certificate for the original system of polynomial equations.\\

\noindent
\Cref{claim:ubd-sym-lifting-char-p} says that the system consisting of $F_{i}$'s, $Q_{t}$'s, and $p_{j}$'s do not have a common zero over the algebraic closure $\Bar{\F}$. Applying Hilbert's Nullstellensatz (\Cref{thm:nullstellensatz}) on this system, we know that there exist polynomials $\Tilde{A}_{i}(\mathbf{y})$'s, $\Tilde{S}_{t}(\mathbf{y})$'s, and $\Tilde{B}_{j}(\mathbf{y})$'s such that:
\begin{equation}\label{eqn:ubd-sym-constant-null}
    \sum_{i=1}^{m} \Tilde{A}_{i}(\mathbf{y}) \cdot F_{i}(\mathbf{y}) + \sum_{t=n+1}^{p^{r}-1} \Tilde{S}_{t}(\mathbf{y}) \cdot Q_{t}(\mathbf{y}) \, + \, \sum_{j=1}^{r} \Tilde{B}_{j}(\mathbf{y}) \cdot p_{j}(y_{j}) \; = \; 1
\end{equation}

\paragraph{\underline{Size analysis of low-variate certificate}}Next, we will show that the coefficients $\Tilde{A}_{i}$'s, $\Tilde{S}_{t}$'s, and $\Tilde{B}_{j}$'s in the Nullstellensatz certificate (\Cref{eqn:ubd-sym-constant-null}) have small constant-depth circuits. More precisely, we will show that $\Tilde{A}_{i}$'s, $\Tilde{S}_{t}$'s, and $\Tilde{B}_{j}$'s are polynomials with sparsity $\mathrm{poly}(n)$, which in turn implies that they have $\mathrm{poly}(n)$-sized depth $2$ circuits. Since these polynomials are $\bigO(\log n)$-variate polynomials, it will suffice to show that they have constant individual degrees. We will argue as follows:
\begin{itemize}
    \item \textbf{Step (a)} The polynomials $F_{i}$'s and $Q_{t}$'s have constant individual degree.
    \item \textbf{Step (b)} Use the low-variate Nullstellensatz certificate \Cref{eqn:ubd-sym-constant-null} to argue that the polynomials $\Tilde{A}_{i}$'s and $\Tilde{S}_{t}$'s can be assumed to have constant individual degree.
    \item \textbf{Step (c)} Use the previous two items to argue that the polynomials $\Tilde{B}_{j}$'s have sparsity at most $\mathrm{poly}(n)$.
\end{itemize}

\paragraph{Step (a)}We first argue that the individual degree of each $F_i(\mathbf{y})$ can be assumed to be $\leq (p-1)$. Recall that $f_{i}(\mathbf{x}) = F_{i}(\widehat{\mathbf{e}}(\mathbf{x}))$ over $\Boo^{n}$, i.e.: $f_{i}(\mathbf{x}) \equiv F_{i}(\widehat{\mathbf{e}}(\mathbf{x}))   \bmod  (\mathbf{x}^{2} - \mathbf{x})$.\newline
As we are working over a field of characteristic $p$, for any $i \in [n]$, $e_{i}(\mathbf{x}) \in \set{0,1,\ldots,p-1}$ for every $\mathbf{x} \in \Boo^{n}$. This implies that $\widehat{\mathbf{e}}(\mathbf{x})^{p} \equiv \widehat{\mathbf{e}}(\mathbf{x}) \bmod (\mathbf{x}^{2} - \mathbf{x})$ via Fermat's Little Theorem/Frobenius automorphism.\newline
Let $F_i'(\mathbf{y}) := F_i(\mathbf{y})/(\mathbf{y}^{p} - \mathbf{y})$ be a ``minimum individual-degree representative'' of $F_i(\mathbf{y})$ modulo the ideal $(\mathbf{y}^{p} - \mathbf{y})$.

Thus, $F_i(\mathbf{y})$ is a polynomial of individual-degree $\leq p-1$ such that $F'(\mathbf{y}) \equiv F(\mathbf{y}) \bmod (\mathbf{y}^{p} - \mathbf{y})$. Combining these together, we get,
\begin{equation}
    f_{i}(\mathbf{x}) \equiv F'_{i}(\widehat{\mathbf{e}}(\mathbf{x}))   \bmod  (\mathbf{x}^{2} - \mathbf{x})
\end{equation}
where $F'_i$ has individual-degree $\leq p-1$. With a slight abuse of notation, we will now use ``$F_i$'' to denote $F'_i$. 

\noindent
By the definition of the polynomial $Q_{t}(\mathbf{y})$, for each $t$, the individual degree of $Q_{t}(\mathbf{y})$ is $\leq p$

\paragraph{Step (b)}The polynomial $p_{j}(t)$, defined as $\prod_{\alpha \in \F_{p}} (t - \alpha)$, is equal to $(t^{p} - t)$ by Fermat's Little Theorem. From the Nullstellensatz certificate \Cref{eqn:ubd-sym-constant-null},
\begin{align*}
    \sum_{i = 1}^{m} \Tilde{A}_{i}(\mathbf{y}) \cdot F_{i}(\mathbf{y}) \, + \, \sum_{t=n+1}^{p^{r}-1} \Tilde{S}_{t}(\mathbf{y}) \cdot Q_{t}(\mathbf{y}) = 1 \mod{( \mathbf{y}^{p} - \mathbf{y} )}
\end{align*}
We would now argue that the polynomials $\Tilde{A}_{i}(\mathbf{y})$'s and $\Tilde{S}_{t}(\mathbf{y})$'s has individual degree $\leq (p-1)$. Suppose there exists an $i \in [m]$ for which $\Tilde{A}_{i}(\mathbf{y})$ has individual degree $> (p-1)$, then define $\Tilde{A}_{i}'(\mathbf{y}) := A_{i}(\mathbf{y})/(\mathbf{y}^{p} - \mathbf{y})$ to be a ``minimum individual-degree representative''. Observe that replacing the polynomial $\Tilde{A}_{i}$ with the polynomial $\Tilde{A}_{i}'$, the low-variate Nullstellensatz certificate \Cref{eqn:ubd-sym-constant-null} continues to holds.

Thus the sparsity of $\Tilde{A}_{i}(\mathbf{y})$ is at most $p^{r}$.
An analogous argument shows that for each $n < t \leq p^{r}-1$, the polynomial $\Tilde{S}_{t}(\mathbf{y})$ has individual degree $\leq (p-1)$ and thus has a $\bigO(n^{2})$-sized circuit of depth $2$ (a $\Sigma \Pi$ circuit).\\

\noindent
Using $r \leq 2 \log_{p} n$, we have that for every $i \in [m]$ and for every $n < t \leq p^{r}-1$, the polynomials $\Tilde{A}_{i}(\mathbf{y})$ and $\Tilde{S}_{t}(\mathbf{y})$ have sparsity at most $\bigO(n^{2})$. This also implies that the polynomials $\Tilde{A}_{i}(\mathbf{y})$'s and $\Tilde{S}_{t}(\mathbf{y})$'s have $\bigO(n^{2})$-sized circuits of depth $2$ (a $\Sigma \Pi$ circuit).

\paragraph{Step (c)}Now it remains to show that the polynomials $\Tilde{B}_{j}(\mathbf{y})$ have small constant-depth circuits. We will show that for each $j \in [n]$, the polynomial $\Tilde{B}_{j}(\mathbf{y})$ has sparsity at most $\mathrm{poly}(n)$. To show this, we will use the fact that $\sum_{i = 1}^{m} \Tilde{A}_{i}(\mathbf{y}) \cdot F_{i}(\mathbf{y}) + \sum_{t=n+1}^{p^{r}-1} \Tilde{S}_{t}(\mathbf{y}) \cdot Q_{t}(\mathbf{y})$ is a polynomial of constant individual degree.\\

\noindent
For a polynomial $H(\mathbf{y})$, the \emph{individual-degree-$p$} operator, denoted by $\mathsf{inddeg}_{p}$ outputs the following polynomial: For each variable $y_{j}$, every occurrence of $y_{j}^{p}$ in $H(\mathbf{y})$ is replaced by $y_{j}$ until the individual degree of the polynomial is $< p$. For $p = 2$, $\mathsf{inddeg}_{2}$ corresponds to multilinearization $\ml$.\newline
Our next claim shows that if there is a polynomial of low individual degree, then its individual-degree-$p$ component can be extracted using polynomials of small constant-depth circuits. The proof is via a simple induction. We omit the proof here and it can be found in \Cref{app:proof-ind-deg}.

\begin{restatable}{claim}{boundedinddegml}\label{claim:inddeg-p}
Let $H(\mathbf{y}) \in \F[y_{1},\ldots,y_{r}]$ denote a polynomial whose individual degree is at most $D$. Then there exists polynomials $G_{1}(\mathbf{y}), \ldots, G_{r}(\mathbf{y})$ such that the following holds:
\begin{align*}
    H(\mathbf{y}) \; = \; \mathsf{inddeg}_{p}[H(\mathbf{y})] \; + \; \sum_{j=1}^{r} G_{j}(\mathbf{y}) \cdot (y_{j}^{p} - y_{j}),
\end{align*}
and for each $j \in [r]$, the polynomial $G_{j}(\mathbf{y})$ has sparsity at most $D^{r+1}/(p-1)$.
\end{restatable}

\paragraph{}Define the polynomial $H(\mathbf{y}) = \sum_{i = 1}^{m} \Tilde{A}_{i}(\mathbf{y}) \cdot F_{i}(\mathbf{y}) + \sum_{t=n+1}^{p^{r}-1} \Tilde{S}_{t}(\mathbf{y}) \cdot Q_{t}(\mathbf{y})$. Observe that
\begin{align*}
    H(\mathbf{y}) = 1 + \sum_{j = 1}^{r} \Tilde{B}_{j}(\mathbf{y}) \cdot (y_{j}^{p} - y_{j})
\end{align*}
So the polynomials $G_{j}(\mathbf{y})$ from \Cref{claim:inddeg-p} correspond to the polynomials $\Tilde{B}_{j}(\mathbf{y})$ from the Nullstellensatz certificate.
Since both $\Tilde{A}_{i}(\mathbf{y})$ and $F_{i}(\mathbf{y})$ have individual degree at most $(p-1)$, $H(\mathbf{y})$ has individual degree at most $2(p-1)$ in each of the $r$ variables. Applying \Cref{claim:inddeg-p} on $H(\mathbf{y})$, we get that the polynomials $\Tilde{B}_{j}(\mathbf{y})$ has sparsity at most $(2(p-1))^{r} \cdot 2(p-1)/(p-1) = \mathrm{poly(n)}$ since $p \leq 2 \log_{p} n$. This implies that the polynomials $\Tilde{B}_{j}(\mathbf{y})$ has a circuit of size $\bigO(n^{2})$ and depth $2$ (a $\Sigma \Pi$ circuit).

\paragraph{\underline{Lifting the Nullstellensatz certificate}}By the definition of $F_{i}(\mathbf{y})$, we know that for every $i \in [m]$, $f_{i}(\mathbf{x}) = \ml[F_{i}(\widehat{\mathbf{e}}(\mathbf{x}))]$. Applying the multilinearization \Cref{lemma:multilinearize-char-p} on $F_{i}(\widehat{\mathbf{e}}(\mathbf{x}))$ for every $i \in [m]$, we know there exists polynomials $D_{ij}(\mathbf{x})$ for $i \in [m]$ and $j \in [n]$ such that:
\begin{align*}
    F_{i}(\widehat{\mathbf{e}}(\mathbf{x})) \; = \; f_{i}(\mathbf{x}) \; + \; \sum_{j = 1}^{n} D_{ij}(\mathbf{x}) \cdot (x_{j}^{2} - x_{j}),
\end{align*}
where each polynomial $D_{ij}(\mathbf{x})$ has a circuit of size $\bigO(n^{5} \log n)$ and depth $5$. This also implies that for each $i \in [m]$, the polynomial $f_{i}(\mathbf{x})$ has a circuit of size $\bigO(n^{5} \log n)$ and depth $5$.\newline
Similarly, applying the multilinearization lemma \Cref{lemma:multilinearize-char-p} on $Q_{t}(\widehat{\mathbf{e}}(\mathbf{x}))$ for every $n < t \leq p^{r}-1$, we know there exists polynomials $R_{t}(tj)(\mathbf{x})$ for $n < t \leq p^{r}-1$ and $j \in [n]$ such that:
\begin{align*}
    Q_{t}(\widehat{\mathbf{e}}(\mathbf{x})) \; = \; \ml[Q_{t}(\widehat{\mathbf{e}}(\mathbf{x}))] \; + \; \sum_{j=1}^{n} R_{tj}(\mathbf{x}) \cdot (x_{j}^{2} - x_{j}),
\end{align*}
where each polynomial $R_{tj}(\mathbf{x})$ has a circuit of size $\bigO(n^{5} \log n)$ and depth $5$.

\paragraph{}Next, we multilinearize the coefficients $\Tilde{A}_{i}(\widehat{\mathbf{e}}(\mathbf{x}))$'s to get a multilinear constant-depth-$\IPSLIN$ proof. Applying the multilinearization \Cref{lemma:multilinearize-char-p} on $\Tilde{A}_{i}(\widehat{\mathbf{e}}(\mathbf{x}))$ for every $i \in [m]$, we know there exists polynomials $\Tilde{D}_{ij}(\mathbf{x})$ for $i \in [m]$ and $j \in [n]$ such that:
\begin{align*}
    \Tilde{A}_{i}(\widehat{\mathbf{e}}(\mathbf{x})) \; = \; \ml[\Tilde{A}_{i}(\widehat{\mathbf{e}}(\mathbf{x}))] + \sum_{j=1}^{n} \Tilde{D}_{ij}(\mathbf{x}),
\end{align*}
where each polynomial $\Tilde{D}_{ij}(\mathbf{x})$ has a circuit of size $\bigO(n^{5} \log n)$ and depth 5. This also implies that for each $i \in [m]$, the polynomial $\ml[\Tilde{A}_{i}(\widehat{\mathbf{e}}(\mathbf{x}))]$ has a circuit of size $\bigO(n^{6})$ and depth $5$.\\
Similarly, we also multilinearize the coefficients $\Tilde{S}_{t}(\widehat{\mathbf{e}}(\mathbf{x}))$. Applying the multilinearization lemma \Cref{lemma:multilinearize-char-p} on $\Tilde{S}_{t}(\widehat{\mathbf{e}}(\mathbf{x}))$ for every $n < t \leq p^{r}-1$, we know there exists polynomials $\Tilde{R}_{tj}(\mathbf{x})$ for $n < t \leq p^{r}-1$ and $j \in [n]$ such that:
\begin{align*}
    \Tilde{S}_{t}(\widehat{\mathbf{e}}(\mathbf{x})) \; = \; \ml[\Tilde{S}_{t}(\widehat{\mathbf{e}}(\mathbf{x}))] \; + \; \sum_{j=1}^{n} \Tilde{R}_{tj}(\mathbf{x}) \cdot (x_{j}^{2} - x_{j}),
\end{align*}
where each polynomial $\Tilde{R}_{tj}(\mathbf{x})$ has a circuit of size $\bigO(n^{5} \log n)$ and depth $5$.\newline
For every $i \in [r]$, applying the multilinearization lemma \Cref{lemma:multilinearize-char-p} on $p_{i}(e_{p^{i-1}}(\mathbf{x}))$, we know that there exists polynomials $E_{i1}(\mathbf{x}), \ldots, E_{in}(\mathbf{x})$ such that:
\begin{align*}
    p_{i}(e_{p^{i-1}}(\mathbf{x})) \; = \; \sum_{j=1}^{n} E_{ij}(\mathbf{x}) \cdot (x_{j}^{2} - x_{j}),
\end{align*}
where for each $j \in [n]$, the polynomial $E_{ij}(\mathbf{x})$ has a circuit of size $\bigO(n^{5} \log n)$ and depth $5$.

\paragraph{}Substituting $\mathbf{y} = \widehat{\mathbf{e}}(\mathbf{x})$ in the low-variate Nullstellensatz certificate \Cref{eqn:ubd-sym-constant-null} and using the above polynomial relations, we get,
\begin{align*}
    \sum_{i = 1}^{m} \; \underbrace{\Tilde{A}_{i}(\widehat{\mathbf{e}}(\mathbf{x}))}_{:= A_{i}(\mathbf{x})} \; \cdot \; f_{i}(\mathbf{x}) \; + \; \sum_{j = 1}^{n} \, B_{j}(\mathbf{x}) \cdot (x_{j}^{2} - x_{j}) \; = \; 1, 
\end{align*}
where the polynomial $B_{j}(\mathbf{x})$ is:
\begin{align*}
    B_{j}(\mathbf{x}) \; = \;  \sum_{i = 1}^{m} (\ml[\Tilde{A}_{i}(\widehat{\mathbf{e}}(\mathbf{x}))] D_{ij}(\mathbf{x}) + f_{i}(\mathbf{x}) \Tilde{D}_{ij}(\mathbf{x})) \\ + \sum_{t=n+1}^{p^{r}-1} (\ml[Q_{t}(\widehat{\mathbf{e}}(\mathbf{x}))] \cdot \Tilde{S}_{t}(\widehat{\mathbf{e}}(\mathbf{x})) + Q_{t}(\widehat{\mathbf{e}})(\mathbf{x}) \cdot \ml[\Tilde{S}_{t}(\widehat{\mathbf{e}}(\mathbf{x}))]) \\ + \, \sum_{i = 1}^{m} \sum_{j' \leq j} \Tilde{D}_{ij'}(\mathbf{x}) D_{ij}(\mathbf{x}) + \sum_{i=1}^{r} \Tilde{B}_{i}(\widehat{\mathbf{e}}(\mathbf{x})) E_{ij}(\mathbf{x})
\end{align*}
We have,
\begin{itemize}
    \item For each $i \in [m]$, using Ben-Or's construction \Cref{thm:ben-or}, the polynomial $A_{i}(\mathbf{x})$ has a circuit of size $\bigO(n^{2})$ and depth $5$.
    \item For each $j \in [r]$, and again using Ben-Or's construction \Cref{thm:ben-or}, the polynomial $B_{j}(\mathbf{x})$ has a circuit of size $\bigO(m n^{5} \log n)$ and depth $7$.
\end{itemize}
This finishes the proof of \Cref{thm:ubd-sym-const-depth} in the setting when the underlying field has a positive characteristic $p$ for a constant prime $p$.

\paragraph{}We now discuss the proof of \Cref{thm:ubd-sym-const-depth} in the setting when the underlying field has characteristic $0$ or $> n$. The proof has the exact same steps as for positive characteristic. Instead of repeating the same steps again, for the sake of brevity, we only highlight the differences.

\paragraph{}Over characteristic $0$ or $> n$, every multilinear symmetric polynomial is a polynomial of $x_{1}+\ldots+x_{n}$, i.e. of $e_{1}(\mathbf{x})$. Thus $r$ in the above proof is just $1$ and $F_{i}$'s are univariate polynomials. Let $p(y) = \prod_{i=0}^{n} (y-i)$. Then it is easy to see that if $f_{i}(\mathbf{x})$ do not have a common Boolean solution, then the univariate polynomials $F_{i}(y)$'s and $p(y)$ do not have a common solution (following a similar strategy to the proof of \Cref{claim:ubd-sym-lifting-char-p}, if $b$ is a common solution, then there exists a common Boolean solution of Hamming weight $b$, which is a contradiction).\newline

\noindent
To argue about the circuit size of the coefficients of the univariate Nullstellensatz certificate, it suffices to argue about their degrees since they are all univariate polynomials. The coefficients of the univariate certificate have degree at most $\bigO(n)$ because of the polynomial $p(y)$ (it is quite similar and simpler to the degree analysis of the coefficients of the low-variate Nullstellensatz certificate in the above proof).\newline
In the end, we need to multilinearize $F_{i}(e_{1}(\mathbf{x}))$. This can again be done in a constant-depth circuit using \Cref{lemma:multilinearize-char-p}.
\end{proof}

\subsection{Multilinearization}
 To show our multilinearization lemma (\Cref{lemma:multilinearize-char-p}), it will be convenient to first define a notion of \emph{partial multilinearization}, i.e., multilinearize with respect to a subset of variables. A key lemma used in our proofs of multilinearization statements is constant-depth multilinearization when $f(\mathbf{x})$ is a \emph{product of univariate polynomials} (see \Cref{cor:ml-prod-uni}). We now define the partial multilinearization and then use it to prove \Cref{cor:ml-prod-uni}.\\

 \begin{definition}[Partial multilinearization]
Fix any field $\F$ and let $f(\mathbf{x}) \in \F[x_{1},\ldots,x_{n}]$. For any $j \in [n]$, let $f^{(\leq j)}(\mathbf{x}) \in \F[x_{j+1},\ldots,x_{n}][x_{1},\ldots,x_{j}]$ denote the polynomial $f(\mathbf{x})$ with variables $x_{1},\ldots,x_{j}$ and coefficients in $\F[x_{j+1},\ldots,x_{n}]$.\newline
The multilinearization of the polynomial $f(\mathbf{x})$ with respect to the variables $\set{x_{1},\ldots,x_{j}}$, denoted by $\ml_{\leq j}[f(\mathbf{x})]$, is defined to be:
\begin{align*}
    \ml_{\leq j}[f(\mathbf{x})] \; := \; \ml[f^{(\leq j)}(\mathbf{x})]
\end{align*}
Similarly, for any $k \in [n]$, let $f^{(k)}(\mathbf{x}) \in \F[x_{1},\ldots,x_{k-1},x_{k+1},\ldots,x_{n}][x_{k}]$ denote the polynomial $f(\mathbf{x})$ with variable $x_{k}$ only and coefficients in $\F[x_{1},\ldots,x_{k-1},x_{k+1},\ldots,x_{n}]$. The multilinearization of the polynomial $f(\mathbf{x})$ with respect to the variable $x_{k}$ only, denoted by $\ml_{k}[f(\mathbf{x})]$, is defined to be:
\begin{align*}
    \ml_{k}[f(\mathbf{x})] \; = \; \ml[f^{(k)}(\mathbf{x})]
\end{align*}
Sometimes we will denote $\ml_{k}[f(\mathbf{x})]$ by $\ml_{x_{k}}[f(\mathbf{x})]$ for sake of clarity.
\end{definition}
\noindent
\textbf{Example}: Let $f(\mathbf{x}) = x_{1}^{2} x_{2}^{3} + x_{2} x_{3}^{2}$. Then,
\begin{align*}
    \ml_{\leq 1}[f(\mathbf{x})] \; = \; x_{1}x_{2}^{3} + x_{2}x_{3}^{2}, \quad \ml_{\leq 2}[f(\mathbf{x})] \; = \; x_{1}x_{2} + x_{2}x_{3}^{2}, \quad
    \ml_{2}[f(\mathbf{x})] \; = \; x_{1}^{2}x_{2} + x_{2}x_{3}^{2}
\end{align*}
\noindent
We make one observation on partial multilinearization, which will be helpful in the proofs.
\begin{observation}\label{obs:partial-ml-sequence}
For every $j < n$, the following holds: For every polynomial $f(\mathbf{x})$,
\begin{align*}
    \ml_{\leq j+1}[f(\mathbf{x})] \; = \; \ml_{j+1}[\ml_{\leq j}[f(\mathbf{x})]]
\end{align*}
\end{observation}
\noindent
In the rest of the section, we will use the notation $\mathbf{x}_{\leq j}$ to denote $(x_{1},\ldots,x_{j})$ and $\mathbf{x}_{>j}$ to denote $(x_{j+1},\ldots,x_{n})$. 

\paragraph{}Now we show that a product of univariate polynomials can be multilinearized using constant-depth $\mathrm{poly}(n)$-sized circuits (see \Cref{cor:ml-prod-uni}). We start by showing that we can do partial multilinearization with respect to a single variable.

\begin{claim}[Multilinearize a single variable]
\label{claim:multilinearize-one-var}
Consider a univariate polynomial $h(z)$ of degree-$D$. Let $Q(\mathbf{y})$ be a polynomial with a circuit of size $s$ and depth $\Delta$. Let $\ml_{z}[h(z) \cdot Q(\mathbf{y})]$ denotes the partial multilinearization of the polynomial $h(z) \cdot Q(\mathbf{y})$ with respect to the $z$ variable.\newline
Then,
\begin{align*}
    h(z) \cdot Q(\mathbf{y}) \; = \; \ml_{z}[h(z) \cdot Q(\mathbf{y})] + B(z,\mathbf{y}) \cdot (z^{2} - z),
\end{align*}
\begin{itemize}
    \item The polynomial $\ml_{z}[h(z) \cdot Q(\mathbf{y})]$ is equal to $L(z)\cdot  Q(\mathbf{y})$, where $L(z)$ is a degree-$1$ univariate polynomial in $z$.
    \item The polynomial $B(z,\mathbf{y})$ is equal to $\Tilde{h}(z) \cdot Q(\mathbf{y})$ for a univariate polynomial $\Tilde{h}(z)$.
\end{itemize}
\end{claim}
\begin{proof}[Proof of \Cref{claim:multilinearize-one-var}]
Let $h(z) = a_{0} + a_{1} z + \cdots + a_{D} z^{D}$. Then for every $2 \leq j \leq D$, rewriting $a_{j} z^{j}$ as $a_{j}(z^{j} - z) + a_{j} z$, we get,
\begin{align*}
    h(z) \cdot Q(\mathbf{y}) \; = \; \paren{a_{0} + a_{1} z + \sum_{j=2}^{D} a_{j} (z^{j} - z + z)} \cdot Q(\mathbf{y}) \\
    =  \underbrace{\paren{a_{0} + \paren{\sum_{j=1}^{D} a_{j}} z}}_{:= L(z)} \cdot Q(\mathbf{y}) + \paren{\sum_{j=2}^{D} a_{j} (z^{j} - z)} \cdot Q(\mathbf{y})
\end{align*}
Observe that for any $j \geq 2$,
\begin{align*}
    z^{j} - z \; = \; \paren{z^{j-2} + \cdots + z + 1} \cdot (z^{2} - z)
\end{align*}
Using the above observation, we get,
\begin{align*}
    h(z) \cdot Q(\mathbf{y}) \; = \;  \underbrace{L(z) \cdot Q(\mathbf{y})}_{= \ml_{z}[h(z) \cdot Q(\mathbf{y})]} \; + \; \underbrace{Q(\mathbf{y}) \sum_{j=2}^{D} a_{j}  \paren{z^{j-1} + \cdots + z + 1} }_{= B(z, \mathbf{y})} \; \cdot \; (z^{2} - z)
\end{align*}
Let $\Tilde{h}(z) = \sum_{j=2}^{D} a_{j} (z^{j-1}+\cdots+z+1)$. Then the polynomial $B(z, \mathbf{y})$ is equal to $\Tilde{h}(z) \cdot Q(\mathbf{y})$. The partial multilinearization $\ml_{z}[h(z) \cdot Q(\mathbf{y})]$ is of the form $L(z) \cdot Q(\mathbf{y})$ for a degree-$1$ polynomial $L(z)$. This finishes the proof of \Cref{claim:multilinearize-one-var}.
\end{proof}

The next claim shows that if a product of univariate polynomials, then we can do partial multilinearization with respect to a subset of variables. It follows with a simple induction using \Cref{claim:multilinearize-one-var}. We omit the proof here and it can be found in \Cref{app:claim-partial-uni}.

\begin{restatable}[Partial multilinearization of product of univariates]{claim}{partialmlproduni}\label{claim:partial-multilinearize-prod-univariates}
Let $h_{1}(z_{1}),\ldots,h_{n}(z_{n})$ be univariate polynomials where each $h_{i}(z_{i})$ has degree at most $D$.\newline
Then there exists degree-$1$ univariate polynomials $L_{1}(z_{1}), \ldots, L_{n}(z_{n})$ and polynomials $B_{1}(\mathbf{z}), \ldots, B_{n}(\mathbf{z})$ satisfying the following: For every $k \in [n]$,
\begin{align*}
    h_{1}(z_{1}) \cdots h_{n}(z_{n}) \; = \; \ml_{\leq k}\brac{\prod_{i=1}^{n} h_{i}(z_{i})} \; + \; \sum_{j=1}^{k} B_{j}(\mathbf{z}) \cdot (z_{j}^{2} - z_{j}),
\end{align*}
where
\begin{align*}
    \ml_{\leq k}\brac{\prod_{i=1}^{n} h_{i}(z_{i})} \; = \; \prod_{i = 1}^{k} L_{i}(z_{i}) \cdot \prod_{i = k+1}^{n} h_{i}(z_{i}),
\end{align*}
and for each $j \in [n]$, the polynomial $B_{j}(\mathbf{z})$ has the following form:
\begin{align*}
    B_{j}(\mathbf{x}) \; = \; \prod_{i = 1}^{j-1} L_{i}(z_{i}) \cdot \Tilde{h}_{j}(z_{j}) \cdot \prod_{i = j+1}^{n} h_{i}(z_{i}),
\end{align*}
for some univariate polynomial $\Tilde{h}_{j}(z_{j})$.
\end{restatable}

\noindent
Setting $k = n$ in \Cref{claim:partial-multilinearize-prod-univariates} immediately gives us the following corollary.
\begin{corollary}[Multilinearization of product of univariates]\label{cor:ml-prod-uni}
Let $h_{1}(z_{1}),\ldots,h_{n}(z_{n})$ be univariate polynomials where each $h_{i}(z_{i})$ has degree at most $D$.\newline
Then there polynomials $B_{1}(\mathbf{z}), \ldots, B_{n}(\mathbf{z})$ such that,
\begin{align*}
    h_{1}(z_{1}) \cdots h_{n}(z_{n}) \; = \; \ml\brac{\prod_{i=1}^{n} h_{i}(z_{i})} \; + \; \sum_{j=1}^{k} B_{j}(\mathbf{z}) \cdot (z_{j}^{2} - z_{j}),
\end{align*}
where for each $j \in [n]$, the polynomial $B_{j}(\mathbf{z})$ has a circuit of size $\bigO(nD^{2})$ and depth $3$ (a $\Pi \Sigma \Pi$) circuit.
\end{corollary}

\noindent
In this section, we will prove the multilinearization lemma \Cref{lemma:multilinearize-char-p}. The key step in our proof of \Cref{lemma:multilinearize-char-p} is \Cref{lemma:ml-prod-two-elem} which is a special case of \Cref{lemma:multilinearize-char-p}. In particular, \Cref{lemma:ml-prod-two-elem} shows that the multilinearization of a product of two elementary symmetric polynomials has a small constant-depth circuit. Furthermore, it shows that the multlinearization of a product of two elementary symmetric polynomials has a nice structure which we use to prove \Cref{lemma:multilinearize-char-p}.\\

\begin{lemma}[Multilinearization of product of two elementary symmetric polynomials]\label{lemma:ml-prod-two-elem}
Fix any two natural numbers $\alpha$ and $\beta$. Then
\begin{itemize}
    \item There exists polynomials $R_{\alpha, \beta,j}(\mathbf{x})$'s such that
    \begin{align*}
        \ml[e_{\alpha}(\mathbf{x}) \cdot e_{\beta}(\mathbf{x})] \; = \; e_{\alpha}(\mathbf{x}) \cdot e_{\beta}(\mathbf{x}) - \sum_{j=1}^{n} R_{\alpha, \beta, j}(\mathbf{x}) \cdot (x_{j}^{2} - x_{j}),
    \end{align*}
    where each polynomial $R_{\alpha, \beta, j}(\mathbf{x})$ has a circuit of size $\bigO(n^{3})$ and depth $5$ (a $\Sigma \Pi \Sigma \Pi \Sigma$ circuit).
    \item There exists coefficients $c_{\alpha, \beta}^{(i)}$'s such that
    \begin{align*}
        \ml[e_{\alpha}(\mathbf{x}) \cdot e_{\beta}(\mathbf{x})] \; = \; \sum_{i = 1}^{n} c_{\alpha,\beta}^{(i)} \; e_{i}(\mathbf{x})
    \end{align*}
\end{itemize}
\end{lemma}

\begin{proof}[Proof of \Cref{lemma:ml-prod-two-elem}]
Using Ben-Or's construction (\Cref{thm:ben-or}) for $e_{\alpha}(\mathbf{x})$ and $e_{\beta}(\mathbf{x})$,
\begin{align*}
    e_{\alpha}(\mathbf{x}) \cdot e_{\beta}(\mathbf{x}) \; = \; \sum_{i_{1}, i_{2}} c_{\alpha,i_{1}}  c_{\beta,i_{2}} \prod_{j=1}^{n} (1+\gamma_{i_{1}} x_{j}) (1+\gamma_{i_{2}} x_{j}), \quad \quad \text{where all } \; c_{\alpha_{i_{1}}}, c_{\beta, i_{2}}, \gamma_{i_{1}}, \gamma_{i_{2}} \in \F
\end{align*}
\begin{equation}\label{eqn:ml-prod-two-elem-1}
 \Rightarrow e_{\alpha}(\mathbf{x}) \cdot e_{\beta}(\mathbf{x}) \; = \; = \sum_{i=1}^{(n+1)^{2}} \prod_{j=1}^{n} h_{i,j}^{\alpha,\beta}(x_{j}),
\end{equation}
where each polynomial $h_{i,j}^{\alpha,\beta}(x_{j})$ is a degree-$2$ \emph{univariate} polynomial. Fix any $i \in [(n+1)^{2}]$ and using \Cref{cor:ml-prod-uni} on $h_{i,1}^{\alpha,\beta}(x_{1}) \cdots h_{i,n}^{\alpha, \beta}(x_{n})$, we know that there exists polynomials $B_{i,j}^{\alpha, \beta}(\mathbf{x})$'s such that:
\begin{align*}
    h_{i,1}^{\alpha,\beta}(x_{1}) \cdots h_{i,n}^{\alpha,\beta}(x_{n}) \; = \; \ml[h_{i,1}^{\alpha,\beta}(x_{1}) \cdots h_{i,n}^{\alpha,\beta}(x_{n})] + \sum_{j=1}^{n} B_{i,j}^{\alpha, \beta}(\mathbf{x}) \cdot (x_{j}^{2} - x_{j}),
\end{align*}
where each polynomial $B_{i,j}(\mathbf{x})$ has a circuit of size $\bigO(n)$ and depth $3$ (a $\Pi \Sigma \Pi$ circuit). Now summing it over all $i \in [(n+1)^{2}]$ (see \Cref{eqn:ml-prod-two-elem-1}), we get,
\begin{align*}
    e_{\alpha}(\mathbf{x}) \cdot e_{\beta}(\mathbf{x}) \; = \; \ml[e_{\alpha}(\mathbf{x}) \cdot e_{\beta}(\mathbf{x})] + \sum_{j=1}^{n} R_{\alpha, \beta, j}(\mathbf{x}) \cdot (x_{j}^{2} - x_{j}),
\end{align*}
where each polynomial $R_{\alpha, \beta, j}(\mathbf{x})$ has a circuit of size $\bigO(n^{3})$ and depth $4$ (a $\Sigma \Pi \Sigma \Pi$ circuit). This shows the \emph{constant-depth circuit} item of \Cref{lemma:ml-prod-two-elem}.

\paragraph{}Next we argue about the \emph{structure} item of \Cref{lemma:ml-prod-two-elem}. Since $e_{\alpha}(\mathbf{x}) \cdot e_{\beta}(\mathbf{x})$ is a symmetric polynomial, its multilinearization $\ml[e_{\alpha}(\mathbf{x}) \cdot e_{\beta}(\mathbf{x})]$ is also a symmetric polynomial. The Fundamental Theorem of Symmetric Polynomials (see \Cref{thm:fundamental-sym}) implies that $\ml[e_{\alpha}(\mathbf{x}) \cdot e_{\beta}(\mathbf{x})]$ is a polynomial of $e_{k}(\mathbf{x})$'s. Note that any multilinear symmetric polynomial is a linear combination of $e_{k}(\mathbf{x})$'s. Thus $\ml[e_{\alpha}(\mathbf{x}) \cdot e_{\beta}(\mathbf{x})]$ is a linear combination of $e_{k}(\mathbf{x})$'s. This finishes the \emph{structure} item of \Cref{lemma:ml-prod-two-elem}. This finishes the proof of \Cref{lemma:ml-prod-two-elem}.
\end{proof}

\paragraph{}Now we are ready to prove \Cref{lemma:multilinearize-char-p}. The idea for the proof is as follows:
\begin{itemize}
    \item We use the fact that $F(\mathbf{y})$ has at most $\mathrm{poly}(n)$ sparsity. So for each monomial $\mathbf{y}^{\bm{\mu}}$, we multilinearize $\mathbf{y}^{\bm{\mu}} \circ \widehat{\mathbf{e}}(\mathbf{x})$ individually. 
    \item For any fixed monomial $\mathbf{y}^{\bm{\mu}} \circ \widehat{\mathbf{e}}(\mathbf{x})$, we note that it is a product of elementary symmetric polynomials. \Cref{lemma:ml-prod-two-elem} shows how to multilinearize a product of two elementary symmetric polynomials. We repeatedly apply this on $\mathbf{y}^{\bm{\mu}} \circ \widehat{\mathbf{e}}(\mathbf{x})$.
\end{itemize}
We recall the statement of \Cref{lemma:multilinearize-char-p} below and then proceed to prove it.

\multsymposchar*

\begin{proof}[Proof of \Cref{lemma:multilinearize-char-p}]
Suppose the polynomial $F(\mathbf{y}) \in \F[y_{1},\ldots,y_{r}]$ is:
\begin{align*}
    F(\mathbf{y}) \; = \; \sum_{\bm{\mu}} \, \lambda_{\bm{\mu}} \mathbf{y}^{\bm{\mu}},
\end{align*}
where $\bm{\mu} = (\mu_{1},\ldots,\mu_{r})$ denotes the exponent vector of a monomial. Recall that the individual degree of $F(\mathbf{y})$ is at most $\leq (p-1)$. Consider a monomial $\mathbf{y}^{\bm{\mu}}$ with a non-zero coefficient $\lambda_{\bm{\mu}}$ in $F(\mathbf{y})$. We will multilinearize $\mathfrak{m} = \prod_{i=1}^{r} e_{p^{i-1}}(\mathbf{x})^{\mu_{i}}$. We do it by multilinearizing two products at a time using \Cref{lemma:ml-prod-two-elem}. Defining $\mu_{0} := 1$, we have,
\begin{align*}
    \mathfrak{m} \; = \; \underbrace{e_{1}(\mathbf{x}) \cdots e_{1}(\mathbf{x})}_{\mu_{1} \text{ times}} \; \cdots \; \underbrace{e_{p^{r-1}}(\mathbf{x}) \cdots e_{p^{r-1}}(\mathbf{x})}_{\mu_{r} \text{ times}} \; = \; \prod_{\ell = 1}^{\bm{\mu}} e_{\alpha_{\ell}}(\mathbf{x}),
\end{align*}
where $\alpha_{\ell} = p^{i-1}$ if $\ell \in  \brac{\sum_{\ell=1}^{i-1} \mu_{\ell} + 1, \; \sum_{\ell=1}^{i} \mu_{i}}$.\\

\begin{claim}\label{claim:multilinearize-exponent-vector}
Let $\bm{\mu} \in \set{0,1,\ldots,p-1}^{r}$ denote an exponent vector as described above. Then for any $k \in [|\bm{\mu}|]$, the following holds:
\begin{itemize}
    \item (Constant-depth circuit). There exists polynomials $R_{\leq \alpha_{k}, j}(\mathbf{x})$'s such that
    \begin{align*}
        \ml\brac{\prod_{\ell = 1}^{k} e_{\alpha_{\ell}}(\mathbf{x})} \; = \; \prod_{\ell=1}^{k} e_{\alpha_{\ell}}(\mathbf{x}) - \sum_{j=1}^{n} R_{\leq \alpha_{k}, j}(\mathbf{x}) \cdot (x_{j}^{2} - x_{j}),
    \end{align*}
    and the polynomials $R_{\leq \alpha_{k}, j}(\mathbf{x})$'s have circuits of size $\bigO(n^{3} \cdot k)$ and depth $5$.
    \item (Structure). There exists coefficients $c_{\leq \alpha_{k}}^{(1)}, \ldots, c_{\leq \alpha_{k}}^{(n)}$ such that
    \begin{align*}
        \ml\brac{\prod_{\ell = 1}^{k} e_{\alpha_{\ell}}(\mathbf{x})} \; = \; \sum_{i=1}^{n} c_{\leq \alpha_{k}}^{(i)} \; e_{i}(\mathbf{x})
    \end{align*}
\end{itemize}
\end{claim}
\begin{proof}[Proof of \Cref{claim:multilinearize-exponent-vector}]
We will prove this using induction on $k$.
\paragraph{Base case:}For $k = 1$, we have $\ml[e_{\alpha_{1}}] = e_{\alpha_{1}}(\mathbf{x})$ and $R_{\leq \alpha_{1}, j}(\mathbf{x}) = 0$. The claim holds for the base case.

\paragraph{Induction step:}Now we assume the induction is true for $k$ and prove it for $(k+1)$. We will first prove the \emph{constant-depth circuit} item for $(k+1)$ and then prove the \emph{structure} item for $(k+1)$.\\

\noindent
Using the \textbf{structure} item of the induction hypothesis, we have,
\begin{equation}\label{eqn:ml-char-p-1}
    \ml\brac{\prod_{\ell = 1}^{k} e_{\alpha_{\ell}}(\mathbf{x})} \; = \; \sum_{i=1}^{n} c_{\leq \alpha_{k}}^{(i)} \; e_{i}(\mathbf{x})
\end{equation}
Using the third item of \Cref{fact:multilinear},
\begin{align*}
    \ml\brac{\prod_{\ell = 1}^{k+1} e_{\alpha_{\ell}}(\mathbf{x})} \; = \; \ml\brac{\ml\brac{\prod_{\ell = 1}^{k} e_{\alpha_{\ell}}(\mathbf{x})} \cdot e_{\alpha_{k+1}}(\mathbf{x})} \\ \\
    = \ml\brac{\sum_{i = 1}^{n} c_{\leq \alpha_{k}}^{(i)} \; e_{i}(\mathbf{x}) \cdot e_{\alpha_{k+1}}(\mathbf{x}) } && (\text{Using \Cref{eqn:ml-char-p-1}}) 
\end{align*}
\begin{equation}\label{eqn:partial-ml-8}
    = \sum_{i = 1}^{n} c_{\leq \alpha_{k}}^{(i)} \; \ml[e_{i}(\mathbf{x}) \cdot e_{\alpha_{k+1}}(\mathbf{x}) ]
\end{equation}
For each $i \in [n]$, we apply the \textbf{constant-depth circuit} item from \Cref{lemma:ml-prod-two-elem} on $e_{i}(\mathbf{x}) \cdot e_{\alpha_{k+1}}(\mathbf{x})$ to get:
\begin{align*}
    \ml[e_{i}(\mathbf{x}) \cdot e_{\alpha_{k+1}}(\mathbf{x})] \; = \; e_{i}(\mathbf{x}) \cdot e_{\alpha_{k+1}}(\mathbf{x}) - \sum_{j=1}^{n} D_{i,\alpha_{k+1}}(\mathbf{x}) \cdot (x_{j}^{2} - x_{j}),
\end{align*}
where the polynomials $D_{i,\alpha_{k+1}}(\mathbf{x})$ have a circuit of size $\bigO(n^{3})$ and depth $5$. Substituting it in \Cref{eqn:partial-ml-8},
\begin{align*}
    \ml\brac{\prod_{\ell = 1}^{k+1} e_{\alpha_{\ell}}(\mathbf{x})} \; = \; \sum_{i=1}^{n} c_{\leq \alpha_{k}}^{(i)} \; \paren{e_{i}(\mathbf{x}) \cdot e_{\alpha_{k+1}}(\mathbf{x}) - \sum_{j=1}^{n} D_{i,\alpha_{k+1}}(\mathbf{x}) \cdot (x_{j}^{2} - x_{j})} \\ \\
    = \paren{\sum_{i = 1}^{n} c_{\leq \alpha_{k}}^{(i)} e_{i}(\mathbf{x}) } e_{\alpha_{k+1}}(\mathbf{x}) - \sum_{j=1}^{n} \paren{\sum_{i = 1}^{n} c_{\leq \alpha_{k}}^{(i)}  D_{i,\alpha_{k+1}}(\mathbf{x})} \; \cdot \; (x_{j}^{2} - x_{j})
\end{align*}
Using \Cref{eqn:ml-char-p-1},
\begin{align*}
    \sum_{i = 1}^{n} c_{\leq \alpha_{k}}^{(i)} e_{i}(\mathbf{x}) \; = \; \ml\brac{\prod_{\ell = 1}^{k} e_{\alpha_{\ell}}(\mathbf{x})} \; = \; \prod_{\ell = 1}^{k} e_{\alpha_{\ell}}(\mathbf{x}) - \sum_{j=1}^{n} R_{\leq \alpha_{k}, j}(\mathbf{x}) \cdot (x_{j}^{2} - x_{j}),
\end{align*}
where we use the \textbf{constant-depth circuit} item from the induction hypothesis for the last equality. The polynomials $R_{\leq \alpha_{k}, j}(\mathbf{x})$'s have circuits of size $\bigO(n^{3} \cdot k)$ and depth $5$ (a $\Sigma \Pi \Sigma \Pi \Sigma$ circuit). Using this in the previous expression, we have
\begin{gather*}
    \ml\brac{\prod_{\ell = 1}^{k+1} e_{\alpha_{\ell}}(\mathbf{x})} \\ 
    = \; \paren{ \prod_{\ell = 1}^{k} e_{\alpha_{\ell}}(\mathbf{x}) - \sum_{j=1}^{n} R_{\leq \alpha_{k}, j}(\mathbf{x}) \cdot (x_{j}^{2} - x_{j}) } e_{\alpha_{k+1}}(\mathbf{x}) - \sum_{j=1}^{n} \paren{\sum_{i = 1}^{n} c_{\leq \alpha_{k}}^{(i)}  D_{i,\alpha_{k+1}}(\mathbf{x})} \; \cdot \; (x_{j}^{2} - x_{j}) \\
    = \; \prod_{\ell=1}^{k+1} e_{\alpha_{\ell}}(\mathbf{x}) - \sum_{j=1}^{n} \; \underbrace{\paren{\sum_{i = 1}^{n} c_{\leq \alpha_{k}}^{(i)}  D_{i,\alpha_{k+1}}(\mathbf{x}) + R_{\leq \alpha_{k}, j}(\mathbf{x})}}_{:= R_{\leq \alpha_{k+1}, j}(\mathbf{x})} \; \cdot \; (x_{j}^{2} - x_{j})
\end{gather*}
The polynomial $R_{\leq \alpha_{k+1},j}(\mathbf{x})$ has a circuit of size $\bigO(n^{3} \cdot (k+1))$ and depth $5$ (a $\Sigma \Pi \Sigma \Pi \Sigma$ circuit). This shows the \emph{constant-depth circuit} item of the induction.\\

\noindent
By applying the \textbf{structure item} of \Cref{lemma:ml-prod-two-elem} on $e_{i}(\mathbf{x}) \cdot e_{\alpha_{k+1}}(\mathbf{x})$, we get,
\begin{align*}
     \ml\brac{\prod_{\ell = 1}^{k} e_{\alpha_{\ell}}(\mathbf{x})} \; = \; \sum_{i=1}^{n} c_{\leq \alpha_{k}}^{(i)} \; e_{i}(\mathbf{x})
\end{align*}
Substituting it in \Cref{eqn:partial-ml-8},

\begin{align*}
    \ml\brac{\prod_{\ell = 1}^{k+1} e_{\alpha_{\ell}}(\mathbf{x})} \; = \; \sum_{i = 1}^{n} c_{\leq \alpha_{k}}^{(i)} \; \sum_{i'=1}^{n} d_{i}^{(i')} e_{i'}(\mathbf{x}) \; = \; \sum_{i = 1}^{n} c_{\leq \alpha_{k+1}}^{(i)} \; e_{i}(\mathbf{x}),
\end{align*}
where $c_{\leq \alpha_{k+1}}^{(i)} = \sum_{j=1}^{n} c_{\leq \alpha_{k}}^{(j)} d_{j}^{(i)}$. This completes the \emph{structure} item of the induction.\newline
This finishes the induction and thus we have finished the proof of \Cref{claim:multilinearize-exponent-vector}.
\end{proof}

\paragraph{}Now we employ \Cref{claim:multilinearize-exponent-vector} on each monomial $\bm{\mu}$ with non-zero coefficient and then sum them together. It is easy to verify that there exists polynomials $R_{j}(\mathbf{x})$'s such that
\begin{align*}
    F(\widehat{\mathbf{e}}(\mathbf{x})) \; = \; \ml[F(\widehat{\mathbf{e}}(\mathbf{x}))] \, + \, \sum_{j=1}^{n} R_{j}(\mathbf{x}) \cdot (x_{j}^{2} - x_{j}),
\end{align*}
where each polynomial $R_{j}(\mathbf{x})$ has a circuit of size $\bigO(p^{r} r 
 n^{3})$ and depth $5$ (a $\Sigma \Pi \Sigma \Pi \Sigma$ circuit). Using $r \leq 2 \log_{p} n$, we get that each polynomial $R_{j}(\mathbf{x})$ has a circuit of size $\bigO(n^{5} \log n)$ and depth $5$. This finishes the proof of \Cref{lemma:multilinearize-char-p}.
\end{proof}

\medskip

\printbibliography[
heading=bibintoc,
title={References}
]

\appendix


\section{Appendix}
\subsection{Details of $\roABP$-$\IPSLINp$ Lower Bound}

We recall some standard definitions and lemmas that are useful for understanding the complexity of $\roABP$s. For more details, please refer to \cite{FSTW21,Forbes-thesis}.
\begin{definition}[Coefficient matrix]
    Consider $f \in \F[\mathbf{x},\mathbf{y}]$. The coefficient matrix of $C_f$ is defined with the following entries from $\F$: 
    $$(C_f)_{\mathbf{a},\mathbf{b}} := \operatorname{Coeff}_{\mathbf{x}^{\mathbf{a}},\mathbf{y}^{\mathbf{b}}}(f)$$
    where $\operatorname{Coeff}_{\mathbf{x}^{\mathbf{a}},\mathbf{y}^{\mathbf{b}}}(f)$ denotes the coefficient of the monomial $\mathbf{x}^{\mathbf{a}}\mathbf{y}^{\mathbf{b}}$ in $f$. 
\end{definition}

\begin{definition}[Coefficient space]
    Consider $f \in \F[\mathbf{x},\mathbf{y}]$. The \emph{space of $\F[\mathbf{x}][\mathbf{y}]$ coefficients of $f$} is defined as: $$ \bfCoeff_{\mathbf{x}|\mathbf{y}}(f) := \left\{ \Coeff_{\mathbf{x}|\mathbf{y}^{\mathbf{b}}} \right\}_{\mathbf{b}\in \mathbb{N}^{n}}$$ where $\Coeff_{\mathbf{x}|\mathbf{y}^{\mathbf{b}}}$ denotes the coefficient of $\mathbf{y}^{\mathbf{b}}$ when $f$ is viewed as a polynomial in the $\mathbf{y}$-variables, with coefficients from the ring $\F[\mathbf{x}]$. The \emph{space of $\F[\mathbf{y}][\mathbf{x}]$ coefficients of $f$} is defined similarly.
\end{definition}

For any subset $S$ of polynomials over a field $\F$, we will use $\dim(S)$ to denote the dimension of the $\F$-linear span of polynomials in $S$. 

\begin{lemma}[Coefficient dimension equals rank of $C_f$ \cite{Nisan}]
    For any $f \in \F[\mathbf{x},\mathbf{y}]$: $$ \operatorname{rank}(C_f) = \dim(\bfCoeff_{\mathbf{x}|\mathbf{y}}(f)) = \dim(\bfCoeff_{\mathbf{y}|\mathbf{x}}(f))$$
\end{lemma}

\begin{lemma}[Coefficient dimension captures $\roABP$ width \cite{Nisan}\cite{Forbes-thesis}]\label{lem: nisan roabp width}
    For any $f(x_1, \dots, x_n)$, if $f$ is computable by a width-$r$ $\roABP$, then $r \geq \max_{i\in n} \dim(\bfCoeff_{\mathbf{x}_{\leq i}|\mathbf{x}_{> i}}(f))$. Further, there is a width-$r$ $\roABP$ for $f$, where $r = \max_{i\in n} \dim(\bfCoeff_{\mathbf{x}_{\leq i}|\mathbf{x}_{> i}}(f))$.
\end{lemma}

\begin{definition}[Evaluation space]
    For $f\in \F$, the \emph{space of $\F[\mathbf{x}][\mathbf{y}]$ evaluations of $f$ over a set $S\subseteq \F$} is defined as: $$ \bfEval_{\mathbf{x}|\mathbf{y},S}(f) := \left\{ f(\mathbf{x},\boldsymbol{\beta}) \right\}_{\boldsymbol{\beta}\in S^{|\mathbf{y}|}}$$ 
    Omitting the $S$ in the notation will denote that $S = \F$. The \emph{space of $\F[\mathbf{y}][\mathbf{x}]$ evaluations of $f$ over a set $S$} is defined similarly. 
\end{definition}

\begin{lemma}[Evaluation dimension $\leq$ coefficient dimension]\label{lem:coeff vs eval}
    For $f\in \F[\mathbf{x}][\mathbf{y}]$ and $S\subseteq \F$, $$\bfEval_{\mathbf{x}|\mathbf{y},S}(f) \subseteq \bfCoeff_{\mathbf{x}|\mathbf{y}}(f)$$ which implies that $\dim(\bfEval_{\mathbf{x}|\mathbf{y},S}(f)) \leq \dim(\bfCoeff_{\mathbf{x}|\mathbf{y}}(f))$. If $|S|$ is greater than the individual degree of each variable in $f$, then $\bfEval_{\mathbf{x}|\mathbf{y},S}(f) = \bfCoeff_{\mathbf{x}|\mathbf{y}}(f)$.
\end{lemma}

\begin{fact}[Dimension of polynomials = dimension of leading monomials \cite{Forbes-thesis}]\label{fact:dim of leading monomials}
    Let $S = \{f_1(\mathbf{x}), \dots, f_m(\mathbf{x})\} \subseteq \F[\mathbf{x}]$. For each $f_i$, let $\operatorname{LM}(f_i)$ denote the leading monomial of $f_i$ based on some monomial ordering. Then, $\dim \operatorname{span} S = \dim \operatorname{span}\{\operatorname{LM}(f_i): f_i \in S\}$.
\end{fact}

The following lemma proves an analog of the coefficient dimension lower bound from \cite{FSTW21} for the positive characteristic case using the degree lower bound in \Cref{lem: deg lower bound union bound}.

\begin{lemma}[Coefficient dimension lower bound from degree lower bound for fixed partition(Proposition 5.8 \cite{FSTW21})]\label{lem: fixed partition coeff dim lbd}
    Let $n\in \mathbb{N}$. For any $\bfalpha \in \F^n$ and $\beta \in B_{\bfalpha}$, let $f_{\bfalpha,\beta}(\mathbf{x},\mathbf{y})$ be a polynomial that computes $$\frac{1}{\sum_{i=1}^n \alpha_i x_iy_i - \beta}$$ on $\Boo^n$. Let $S$ be a finite subset of $\F$. Then, for a uniformly randomly chosen $\bfalpha \sim S^n$: $$ \Pr_{\bfalpha \sim S^n}[\dim (\bfCoeff_{\mathbf{x}|\mathbf{y}}(f_{\bfalpha,\beta})) \geq 2^n] \geq 1-\frac{2^{2n}}{|S|} $$
\end{lemma}

\begin{proof}
    \Cref{lem:coeff vs eval} implies that $$\dim (\bfCoeff_{\mathbf{x}|\mathbf{y}}(f_{\bfalpha,\beta})) \geq \dim \left\{ f_{\bfalpha,\beta}(\mathbf{x},\mathbf{b}) \right\}_{\mathbf{b}\in \Boo^n}$$
    For any $\mathbf{b}=(b_1, \dots, b_n) \in \Boo^n$, if $U_{\mathbf{b}}:= \{i\in[n]: b_i = 1\}$ then : $$f_{\bfalpha,\beta}(\mathbf{x},\mathbf{b}) = \frac{1}{\sum_{i\in U_{\mathbf{b}}} \alpha_i x_i - \beta}$$
    \Cref{lem: deg lower bound union bound} tells us that for a randomly chosen $\bfalpha \sim S^n$: $$\Pr_{\bfalpha \sim S^n}[\forall \mathbf{b}\in\Boo^n:\deg{f_{\bfalpha,\beta,U_{\mathbf{b}}}(\mathbf{x})} = |\mathbf{b}|] \geq 1 -\frac{2^{2n}}{|S|}$$
    In particular, for a uniformly random $\bfalpha\sim S^n$, for any $\mathbf{b}\in \Boo^n$, the leading monomial of $f_{\bfalpha,\beta,U_{\mathbf{b}}}(\mathbf{x})$ is $c_\mathbf{b}\cdot \prod_{i:b_i = 1}x_i$ for some $c_\mathbf{b} \in \F\setminus\{0\}$. Combining this with \Cref{fact:dim of leading monomials}, we get that with probability at least $1 - (2^{2n}/|S|)$: $$\dim (\bfCoeff_{\mathbf{x}|\mathbf{y}}(f_{\bfalpha,\beta}))\geq \dim \left\{ f_{\bfalpha,\beta}(\mathbf{x},\mathbf{b}) \right\}_{\mathbf{b}\in \Boo^n}\geq \dim \left\{ \operatorname{ml}(f_{\bfalpha,\beta}(\mathbf{x},\mathbf{b})) \right\}_{\mathbf{b}\in \Boo^n} \geq 2^n$$
    since each multilinear restriction $\operatorname{ml}(f_{\bfalpha,\beta}(\mathbf{x},\mathbf{b}))$ generates a different multilinear monomial as its leading monomial, and thus the space contains all $2^n$ multilinear monomials on $\mathbf{x}$. Here, we also used the fact the multilinearization operator is a linear map and does not increase the dimension.
\end{proof}


The following fact relates the coefficient dimension of a polynomial $f\in \F[\mathbf{x},\mathbf{y},\mathbf{z}]$ over $\F(\mathbf{z})$ to the coefficient dimension of $f(\mathbf{x},\mathbf{y},\mathbf{b})$ over $\F$ for any $\mathbf{b} \in \F^n$.  

\begin{fact}[Coefficient dimension over $\F(\mathbf{z})$ $\geq$ coefficient dimension over $\F$ (Lemma 5.12 \cite{FSTW21})]\label{fact: coeff dim F(z) vs F}
    Let $f\in \F[\mathbf{x},\mathbf{y},\mathbf{z}]$. Let $f_\mathbf{z}$ denote $f$ as a polynomial in $\F[\mathbf{z}][\mathbf{x},\mathbf{y}]$ so that for any $\mathbf{b} \in \F^n$, $f_\mathbf{b}(\mathbf{x},\mathbf{y}) = f(\mathbf{x},\mathbf{y},\mathbf{b}) \in \F[\mathbf{x},\mathbf{y}]$. Then for any $\mathbf{b} \in \F^n$: $$\dim_{\F(\mathbf{z})}\bfCoeff_{\mathbf{x}|\mathbf{y}}f_\mathbf{z}(\mathbf{x},\mathbf{y}) \geq \dim_{\F}\bfCoeff_{\mathbf{x}|\mathbf{y}}f_{\mathbf{\mathbf{b}}}(\mathbf{x},\mathbf{y})$$
\end{fact}
Using this fact, \cite{FSTW21} proves a coefficient dimension lower bound over $\F(\mathbf{z})$ for any partition of variables, using the coefficient dimension lower bound over $\F$ for a fixed partition of variables. We observe that their proofs work even when we replace their coefficient dimension lower bound by a suitable version over fields of positive characteristic (\Cref{lem: fixed partition coeff dim lbd}) using the degree lower bound over positive characteristic.

\begin{lemma}[Coefficient dimension lower bound for any partition of variables (Proposition 5.13 \cite{FSTW21})]\label{lem: coeff dim any partition}
    Let $n\in \mathbb{N}$. For any $\bfalpha \in \F^{\binom{2n}{2}}$ and $\beta \in B_{\bfalpha}$, let $f_{\bfalpha,\beta}(\mathbf{x}=(x_i)_{i\in[2n]},\mathbf{z}=(z_{i,j})_{i<j\leq 2n})$ be a polynomial which computes $$\frac{1}{\sum_{i<j\leq n}\alpha_{i,j}z_{i,j}x_ix_j - \beta}$$ on the Boolean cube. Let $S\subseteq \F$. Call an $\bfalpha \in S^{\binom{2n}{2}}$ \emph{good} if for any partition $\mathbf{x} = (\mathbf{u},\mathbf{v})$ with $|\mathbf{u}|=|\mathbf{v}| = n$: $$ \dim_{\F(\mathbf{z})}(\bfCoeff_{\mathbf{u}|\mathbf{v}}(f_{\bfalpha,\beta})) \geq 2^n$$
    where $f_{\bfalpha,\beta}$ is viewed as a polynomial in $\F[\mathbf{z}][\mathbf{x},\mathbf{y}]$ with coefficients in $\F[\mathbf{z}]$. \\
    Then, a uniformly randomly chosen $\bfalpha \in S^{\binom{2n}{2}}$ is good with probability $\geq 1 - \frac{\binom{2n}{n}2^{2n}}{|S|}$.
\end{lemma}
\begin{proof}
    For any balanced partition $\mathbf{x} = (\mathbf{u},\mathbf{v})$ where $|\mathbf{u}| = |\mathbf{v}| = n$, we can embed $\sum_{i\in [n]}u_iv_i - \beta$ in $\sum_{i<j\leq n}\alpha_{i,j}z_{i,j}x_ix_j - \beta$ by a natural restriction $\mathbf{z} = \mathbf{b}_{\mathbf{u},\mathbf{v}} \in \Boo^{\binom{2n}{2}}$ that sets $z_{i,j}$ to 1 if $x_i = u_k$, $x_j = v_k$, and 0 otherwise. So, for every such restriction $\mathbf{b}_{\mathbf{u},\mathbf{v}}$ that corresponds to a balanced partition: $$f(\mathbf{u},\mathbf{v},\mathbf{b}_{\mathbf{u},\mathbf{v}}) = \frac{1}{\sum_{i\in[n]}u_iv_i - \beta}$$
    For any fixed choice of balanced partition $\mathbf{b_{\mathbf{u},\mathbf{v}}}\in\Boo^{\binom{2n}{2}}$, \Cref{lem: fixed partition coeff dim lbd} tells us that for a uniformly randomly chosen $\bfalpha \in S^{\binom{2n}{n}}$: $$ \Pr_{\bfalpha \in S^{\binom{2n}{n}}}[\dim (\bfCoeff_{\mathbf{x}|\mathbf{y}}(f_{\bfalpha,\beta}(\mathbf{u},\mathbf{v},\mathbf{b}))) \geq 2^n] \geq 1-\frac{2^{2n}}{|S|} $$
    Applying a union bound over all $\binom{2n}{n}$ choices of balanced partitions $\mathbf{x} = (\mathbf{u},\mathbf{v})$ implies that for a uniformly randomly chosen $\bfalpha \in S^{\binom{2n}{2}}$: $$ \Pr_{\bfalpha \in S^{\binom{2n}{n}}}[\forall \mathbf{x} = (\mathbf{u},\mathbf{v}): \dim_{\F}(\bfCoeff_{\mathbf{x}|\mathbf{y}}(f_{\bfalpha,\beta}(\mathbf{u},\mathbf{v},\mathbf{b}_{\mathbf{u},\mathbf{v}}))) \geq 2^n] \geq 1-\frac{\binom{2n}{n}2^{2n}}{|S|} $$
    Finally, applying \Cref{fact: coeff dim F(z) vs F} implies that for a uniformly randomly chosen $\bfalpha \in S^{\binom{2n}{2}}$:
    $$ \Pr_{\bfalpha \in S^{\binom{2n}{n}}}[\forall \mathbf{x} = (\mathbf{u},\mathbf{v}): \dim_{\F(\mathbf{z})}(\bfCoeff_{\mathbf{x}|\mathbf{y}}(f_{\bfalpha,\beta}(\mathbf{u},\mathbf{v},\mathbf{b}_{\mathbf{u},\mathbf{v}}))) \geq 2^n] \geq 1-\frac{\binom{2n}{n}2^{2n}}{|S|} $$
\end{proof}

\roabplbdanyorder
\begin{proof}[Proof of \Cref{thm:roabp-lbd-any-order-positive-char}]\label{proof:thm:roabp-lbd-any-order-positive-char}
    We will instantiate \Cref{lem: coeff dim any partition} for the field $\tilde{\F}$ and the set $S = \F$. Thus, choosing $\beta \in \tilde{\F}\setminus\F$ ensures that for any choice of $\bfalpha = (\alpha_{i,j})_{1\leq i <j\leq 2n} \in S^{\binom{2n}{2}}$, $\beta$ will be in $B_{\bfalpha}$ (which we recall to be the complement of all possible subset sums of $\bfalpha$).
    With the above choices, it follows from \Cref{lem: coeff dim any partition} that for a uniformly randomly chosen $\bfalpha \in S^{\binom{2n}{2}}$, 
    $$ \Pr_{\bfalpha \in S^{m}}[\forall \mathbf{x} = (\mathbf{u},\mathbf{v}): \dim_{\F(\mathbf{z})}(\bfCoeff_{\mathbf{x}|\mathbf{y}}(f_{\bfalpha,\beta}(\mathbf{u},\mathbf{v},\mathbf{b}_{\mathbf{u},\mathbf{v}}))) \geq 2^n] \geq 1-\frac{\binom{2n}{n}2^{2n}}{|S|} > 0$$
    for $|S| = |\F|> \binom{2n}{n}2^{2n}$, where $\mathbf{x}=(\mathbf{u},\mathbf{v})$ denotes any balanced partition of $\mathbf{x}$. In particular, this implies that there exists an $\bfalpha \in S^{\binom{2n}{n}}$ such that for any balanced partition $\mathbf{x} = (\mathbf{u},\mathbf{v})$,
    \begin{equation}\label{eqn: coeff dim lbd all partitions}
    \dim_{\F(\mathbf{z})}(\bfCoeff_{\mathbf{x}|\mathbf{y}}(f_{\bfalpha,\beta}(\mathbf{u},\mathbf{v},\mathbf{b}_{\mathbf{u},\mathbf{v}}))) \geq 2^n
    \end{equation}

    Now, suppose $f(\mathbf{x},\mathbf{z})$ is computable by a width-$r$ $\roABP$ in some order of variables. Using $f_\mathbf{z}$ to denote $f$ as a polynomial in $\F[z][\mathbf{x}]$, it follows that $f_\mathbf{z}$ is also computable by a width-$r$ $\roABP$ over the fraction field $\F(\mathbf{z})$ in the induced order of variables on $\mathbf{x}$. By splitting the $\mathbf{x}$ variables in half along the induced order, using \Cref{eqn: coeff dim lbd all partitions} along with Nisan's characterization of width of $\roABP$s (\Cref{lem: nisan roabp width}), we obtain the required lower bound.
\end{proof}

\subsection{Proof of \Cref{claim:monomial-axioms}}\label{app:proof-monomial-axioms}

\monomialaxioms*

\begin{proof}[Proof of \Cref{claim:monomial-axioms}]
 We will prove it by induction on the cardinality of $\mathrm{Supp}(\bm{\mu})$, which is defined as follows:
 \begin{align*}
     \mathrm{Supp}(\bm{\mu}) \; = \; \setcond{j \in [n]}{\mu_{j} > 0}.
 \end{align*}
 
\paragraph{Base case:}Suppose $|\mathrm{Supp}(\bm{\mu})| = 1$ and $\mu_{1} > 0$. If $\mu_{1} = 1$, then we can set $E_{\bm{\mu},1}(\mathbf{x}) = 1$. Otherwise, if $\mu_{1} > 1$, then
\begin{align*}
    x_{1}^{2 \mu_{1}} - x_{1}^{\mu_{1}} \; = \; (x_{1}^{2 \mu_{1}} - x_{1}) - (x_{1}^{\mu_{1}} - x_{1})    
\end{align*}
We have the following identity for any $j \geq 2$:
\begin{align*}
    z^{j} - z \; = \; (z^{j-2} + \ldots + z + 1) \cdot (z^{2} - z)
\end{align*}
Using this we get,
\begin{gather*}
     x_{1}^{2 \mu_{1}} - x_{1}^{\mu_{1}} \; = \; (x_{1}^{2 \mu_{1} - 2} + \ldots + x_{1} + 1) \cdot (x_{1}^{2} - x_{1}) - (x_{1}^{\mu_{1} - 2} + \ldots + x_{1} + 1) \cdot (x_{1}^{2} - x_{1}) \\
     = \underbrace{(x_{1}^{2 \mu_{1} - 2} + \ldots + x_{\mu_{1} - 1})}_{:= E_{\bm{\mu},1}(\mathbf{x})} \cdot (x_{1}^{2} - x_{1})
\end{gather*}
The polynomial $E_{\bm{\mu},1}(\mathbf{x})$ has a circuit of size $\bigO(D^{2})$ and depth $2$ (a $\Sigma \Pi$ circuit).

\paragraph{Induction step:}Assume this is true for all $\bm{\mu}$ with $|\bm{\mu}| \leq D$ and $|\mathrm{Supp}(\bm{\mu})| = k$. Consider any arbitrary exponent vector $\bm{\mu}$ with $|\bm{\mu}| \leq D$ and $|\mathrm{Supp}(\bm{\mu})| = (k+1)$. Let $t$ be the largest element in $\mathrm{Supp}(\bm{\mu})$ and let $\bm{\nu}$ be the exponent vector with $\nu_{t} = 0$ and $\nu_{i} = \mu_{i}$ for all $i \neq t$. We have,
\begin{align*}
    ((\mathbf{x}^{\bm{\mu}})^{2} - \mathbf{x}^{\bm{\mu}}) \; = \;  (x_{t}^{2 \mu_{t}} - x_{t} + x_{t}) \cdot (\mathbf{x}^{\bm{\nu}})^{2} - (x_{t}^{\mu_{t}} - x_{t} + x_{t}) \cdot \mathbf{x}^{\bm{\nu}}
\end{align*}
\begin{equation}\label{eqn:monomial-axiom}
    \underbrace{((x_{t}^{2\mu_{t} - 2} + \ldots + 1) \cdot (\mathbf{x}^{\bm{\nu}})^{2} - (x_{t}^{\mu_{t}-2} + \ldots +1) \cdot \mathbf{x}^{\bm{\nu}})}_{:= E_{\bm{\mu},t}(\mathbf{x})} \, \cdot \, (x_{t}^{2} - x_{t}) \; + \; x_{t} \cdot ((\mathbf{x}^{\bm{\nu}})^{2} - \mathbf{x}^{\bm{\nu}}), 
\end{equation}
where we used the identity $(z^{j} - z) = (z^{j-2} + \ldots + z + 1) \cdot (z^{2} - z)$. Since the exponent vector $\bm{\nu}$ satisfies $|\bm{\nu}| \leq D$ and $|\mathrm{Supp}(\bm{\nu})| = k$, we can apply the induction hypothesis on $(\mathbf{x}^{\bm{\nu}})^{2} - \mathbf{x}^{\bm{\nu}}$. From induction, we know there exists polynomials $E_{\bm{\nu}, j}(\mathbf{x})$ for all $j \in \mathrm{Supp}(\bm{\nu})$ such that:
\begin{align*}
    ((\mathbf{x}^{\bm{\nu}})^{2} - \mathbf{x}^{\bm{\nu}}) \; = \; \sum_{j \in \mathrm{Supp}(\bm{\nu})} \, E_{\bm{\nu},j}(\mathbf{x}) \cdot (x_{j}^{2} - x_{j}),
\end{align*}
and the polynomials $E_{\bm{\nu},j}(\mathbf{x})$ have circuits of size $\bigO(nD^{2})$ and depth $2$ (a $\Pi \Sigma$ circuit). Substituting it in \Cref{eqn:monomial-axiom}, we get,
\begin{align*}
   ((\mathbf{x}^{\bm{\mu}})^{2} - \mathbf{x}^{\bm{\mu}}) \; = \; E_{\bm{\mu},t}(\mathbf{x}) \cdot (x_{t}^{2} - x_{t}) \; + \; \sum_{j \in \mathrm{Supp}(\bm{\nu})} \underbrace{x_{t} \cdot E_{\bm{\nu},j}(\mathbf{x})}_{:= E_{S,j}} \cdot (x_{j}^{2} - x_{j}),
\end{align*}
where the polynomials $E_{\bm{\mu},j}(\mathbf{x})$ have a circuit of size $\bigO(nD^{2})$ and depth $2$ (a $\Pi \Sigma$ circuit). Moreover, the polynomials $E_{\bm{\mu},j}(\mathbf{x})$ are of degree-$2D$ polynomials. This finishes the proof of \Cref{claim:monomial-axioms}.
\end{proof}

\subsection{Proof of \Cref{claim:inddeg-p}}\label{app:proof-ind-deg}

\boundedinddegml*

\begin{proof}[Proof of \Cref{claim:inddeg-p}]
Fix an arbitrary monomial $\mathfrak{m}$ with a non-zero coefficient in the polynomial $H(\mathbf{y})$. Say $\mathfrak{m} = y_{1}^{\mu_{1}} \cdots y_{r}^{\mu_{r}}$ where for every $j \in [r]$, $0 \leq e_{i} \leq D$. Let $S_{\mathfrak{m}} \subseteq [r]$ denote the set of variables whose exponent in $\mathfrak{m}$ is at least $p$, i.e.
\begin{align*}
    S_{\mathfrak{m}} \; = \; \setcond{j \in [r]}{p \leq \mu_{j} \leq D}
\end{align*}
Let $\ell \in S_{\mathfrak{m}}$, and let $\mathfrak{m}_{-\ell} := \mathfrak{m}/y_{\ell}^{\mu_{\ell}}$. In other words, $\mathfrak{m} = y_{\ell}^{\mu_{\ell}} \cdot \mathfrak{m}_{-\ell}$. Then, 
\begin{align*}
    \mathfrak{m} \; = \; (y_{\ell}^{p} - y_{\ell} + y_{\ell}) \; y_{\ell}^{\mu_{\ell} - p} \cdot \mathfrak{m}_{-\ell} \\
    = y_{\ell}^{\mu_{\ell} - p + 1} \cdot \mathfrak{m}_{-j} \; + \; y_{\ell}^{\mu_{\ell} - p} \cdot \mathfrak{m}_{-j} \cdot (y_{\ell}^{p} - y_{\ell})
\end{align*}
The monomial $y_{\ell}^{\mu_{\ell} - p} \cdot \mathfrak{m}_{-j}$ is a monomial in the polynomial $G_{\ell}(\mathbf{y})$. We repeat the above step on the monomial $y_{\ell}^{\mu_{\ell} - p + 1} \cdot \mathfrak{m}_{-j}$. In each step with respect to the variable $y_{\ell}$ (for the monomial $\mathfrak{m}$), the degree of $y_{\ell}$ is reducing by $(p-1)$. Thus this step can be repeated $\leq D/(p-1)$ times because the individual degree of $H(\mathbf{y})$ is $\leq D$. In each step, we get one monomial for the polynomial $G_{\ell}(\mathbf{y})$, and thus we get $D/(p-1)$ monomials in the polynomial $G_{\ell}(\mathbf{y})$ from the monomial $\mathfrak{m}$. We do this for every variable in the set $S_{\mathfrak{m}}$.\\

\noindent
Finally, we iterate the above steps for every monomial with non-zero coefficient in the polynomial $H(\mathbf{y})$. The sparsity of the polynomial $H(\mathbf{y})$ is at most $D^{r}$, since the individual degree of $H(\mathbf{y})$ is $\leq D$. For each monomial in the support of $H(\mathbf{y})$, each polynomial $G_{\ell}(\mathbf{y})$ gets at most $D/(p-1)$ monomials, and hence each polynomial $G_{\ell}(\mathbf{y})$,
\begin{align*}
    \mathrm{sparsity}(G_{\ell}(\mathbf{y})) \; = \; \mathrm{sparsity}(H(\mathbf{y})) \cdot D/(p-1) \; \leq \; D^{r} \cdot D/(p-1)
\end{align*}
This finishes the proof of \Cref{claim:inddeg-p}.
\end{proof}

\subsection{Proof of \Cref{claim:partial-multilinearize-prod-univariates}}\label{app:claim-partial-uni}

\partialmlproduni*

\begin{proof}[Proof of \Cref{claim:partial-multilinearize-prod-univariates}]
We will prove this via induction on $k$.

\paragraph{Base case:}For $k = 1$, this is exactly \Cref{claim:multilinearize-one-var} where $z = z_{1}$ and $\mathbf{y} = (z_{2},\ldots,z_{n})$.

\paragraph{Induction case:}Assume the claim is true up to $k$. Let $Q(\mathbf{z}) \; = \; h_{k+2}(z_{k+2}) \cdots h_{n}(z_{n})$. By induction hypothesis, we have,
\begin{align*}
    \ml_{\leq k}\brac{\prod_{i=1}^{n} h_{i}(z_{i})} \; = \; h_{k+1}(z_{k+1}) \; \cdot \; \underbrace{\prod_{i = 1}^{k} L_{i}(z_{i}) \cdot \prod_{i = k+2}^{n} h_{i}(z_{i})}_{= Q(\mathbf{y})},
\end{align*}
where $\mathbf{y} = (z_{1},\ldots,z_{k},z_{k+2},\ldots,z_{n})$. From \Cref{obs:partial-ml-sequence},
\begin{align*}
    \ml_{\leq k+1}\brac{\prod_{i=1}^{n} h_{i}(z_{i})} \; = \; \ml_{z_{k+1}}\brac{\ml_{\leq k}\brac{\prod_{i=1}^{n} h_{i}(z_{i})}},
\end{align*}
Now applying \Cref{claim:multilinearize-one-var} on $h_{k+1}(z_{k+1}) \cdot Q(\mathbf{y})$ with respect to the variable $z_{k+1}$, we get the claim for $k+1$. This finishes the proof of \Cref{claim:partial-multilinearize-prod-univariates}.
\end{proof}

\end{document}